\documentclass[11pt]{article}
\usepackage[paperwidth=8.5in,paperheight=11in,top=1in, bottom=1in, left=0.8in, right=0.8in]{geometry}
\usepackage{float}
\usepackage{amsthm}
\usepackage{amsmath}
\usepackage{amssymb}
\usepackage{graphicx}
\usepackage{esint}
\usepackage{setspace}
\usepackage{epstopdf}
\usepackage{color}
\makeatletter
\usepackage{dsfont}
\usepackage{pstricks}
  \usepackage{natbib}
     \newtheorem*{lem*}{\protect\lemmaname}
  \providecommand{\lemmaname}{Lemma}
  \newtheorem*{prop*}{\protect\propositionname}
\newtheorem{theorem}{Theorem}

\newtheorem{proposition}[theorem]{Proposition}

\newtheorem{example}[theorem]{Example}
\newtheorem{lemma}[theorem]{Lemma}
\newtheorem{prop}[theorem]{Proposition}

\numberwithin{equation}{section}
\numberwithin{theorem}{section}

\newcommand{\be}{\begin{equation}}
\newcommand{\ee}{\end{equation}}
\def\G{\mathbb{G}}
\def\H{\mathbb{H}}
\def\R{\mathbb{R}}
\def\Q{\mathbb{Q}}

\def\F{\mathcal{F}}

\makeatother

  \providecommand{\propositionname}{Proposition}

\begin{document}

\title{Accounting for  Earnings Announcements in  the Pricing of \\
Equity Options}

\author{Tim Leung\thanks{{Industrial Engineering \& Operations Research Department, Columbia University, New York, NY 10027, USA. \mbox{E-mail: leung@\,ieor.columbia.edu}.
Corresponding author. }} \and Marco Santoli\thanks{{Industrial Engineering \& Operations Research Department, Columbia
University, New York, NY 10027, USA.   \mbox{E-mail: ms4164@columbia.edu}. 
} }}

\date{\today}

\maketitle

\begin{small}
\begin{abstract}
We study an option pricing framework that accounts for the price impact of an earnings announcement (EA), and analyze the behavior of the implied volatility surface prior to  the event. On the announcement date, we incorporate  a random jump to the stock price to represent the  shock due to earnings. We consider  different  distributions of the scheduled
earnings jump as well as different underlying stock price dynamics before and after the EA date.    Our main contributions include analytical option pricing formulas when the underlying stock price  follows the Kou model  along with a double-exponential or  Gaussian EA jump on the   announcement date.  Furthermore, we  derive  analytic  bounds and asymptotics for the pre-EA implied volatility   under various models.  The calibration results demonstrate adequate fit of the entire implied volatility surface prior to an announcement. We also compare  the risk-neutral  distribution of the EA jump  to its historical distribution. Finally, we discuss the valuation and exercise strategy of pre-EA American options, and illustrate an analytical approximation and numerical results.

\end{abstract}
 

\tableofcontents
\end{small}
\section{Introduction}\label{sec:introduction}
Public companies routinely  release  summaries of their operations and  performance,  including   income statements,  balance sheets, and other reports.  Such  events are  commonly referred to as  \textit{earnings announcements}.   In the   market session following a scheduled announcement\footnote{Earnings reports are commonly released after the market closes. Therefore,  the   price impact of earnings are first reflected  at the beginning of the  next market session.},  empirical studies (see, for example, \cite{Patell1984}) suggest that  the opening  stock price can  move drastically. For instance, \cite{Dubinsky2006}  study a sample of stocks and report that the variance of  stock price returns on the earnings announcement date is over  five   times greater than those on other dates.  In addition to the immediate price impact, earnings releases may  also   affect the drift of the
stock price over a longer horizon. This empirical  observation is commonly  referred to as the post-earnings-announcement drift; see \cite{Chordia2006} and references therein.

Since many public companies also have options written on their stock prices, this motivates us to investigate the problem of pricing equity options prior to an earnings announcement (EA).  As options  are intrinsically  {forward-looking}   contracts, their prices should account for  the  uncertain stock price impact of an upcoming  earnings release. A natural question is how to extract  some information on such an impact from observed options prices, especially a few days before the announcement. Since traders  often quote or study {the} implied volatility (IV) for each option, in practice it is important to  better understand the  behavior of the \textit{pre-earnings announcement implied volatility} (PEAIV).

\subsection{Market Observations\label{sub:market_observations}}
One main feature of the  PEAIV is that it tends to rise, often  rather  drastically, up until the announcement date.  In Figure \ref{fig:vol_example},  we show  the price and IV  of the front month ({i.e.} with the nearest maturity date) at-the-money (ATM) call option written on    IBM.   Following an earnings announcement on July 17, 2013, the option expires  on the Friday   in  the same week.  As we can see on the right panel, the IV increases rapidly from 30\% to over 70\%  as time  approaches  the earnings announcement, and then drops significantly to close to 20\% after the earnings report is released.   At the same time, the ATM call price   stays well above zero  even though the time to maturity is very short (Figure \ref{fig:vol_example} (left)).   Moreover, we  observe  that the IV reaches a maximum just before an earnings event.

Let us briefly discuss the intuitions behind these observations.  First, the market   price of an  option expiring after earnings announcement  reflects the possibility of a  jump  in the stock price on the announcement date.  Since the Black-Scholes formula, which is  used to convert option prices to IVs,  assumes a log-normal  {model for the} underlying price without jumps,  the only way to account for the deterministically timed yet  random jump in the stock price  is to apply a higher
volatility parameter. Furthermore, since the volatility parameter in the Black-Scholes formula is coupled with time-to-maturity, an even higher volatility parameter is needed as time progresses.  As a result, even though the stock price may not experience  a higher (historical) volatility or clear directional movements before earnings, the (risk-neutral)  IV tends to increase.    However, as soon as     the  earnings announcement is made, the jump in the stock price is realized and the volatility effect will also disappear. 

To illustrate that the increasing PEAIV is more than a coincidence, in Figure \ref{fig:IV_time_series} (left) we plot the time series of the IV of the ATM call option for
IBM from January 1996 to August 2013. In the same figure,  we mark in red the IV values on days during which earnings announcements were made. As we can see, spikes in volatility most often  coincide with the  earnings release dates. We can also gain some insight from the shape of the implied volatility surface.  In Figure \ref{fig:IV_time_series} (right), we
show the IV surface of IBM on July 15, 2013, two days before an earnings
release. It is also clear that the front month options, which have 1 day to expiration, have
very high volatility. The implied volatility is  significantly lower for longer-term options, suggesting a  less pronounced effect of the upcoming earnings announcement.  

\begin{figure}[th]
\centering
\includegraphics[scale=0.43]{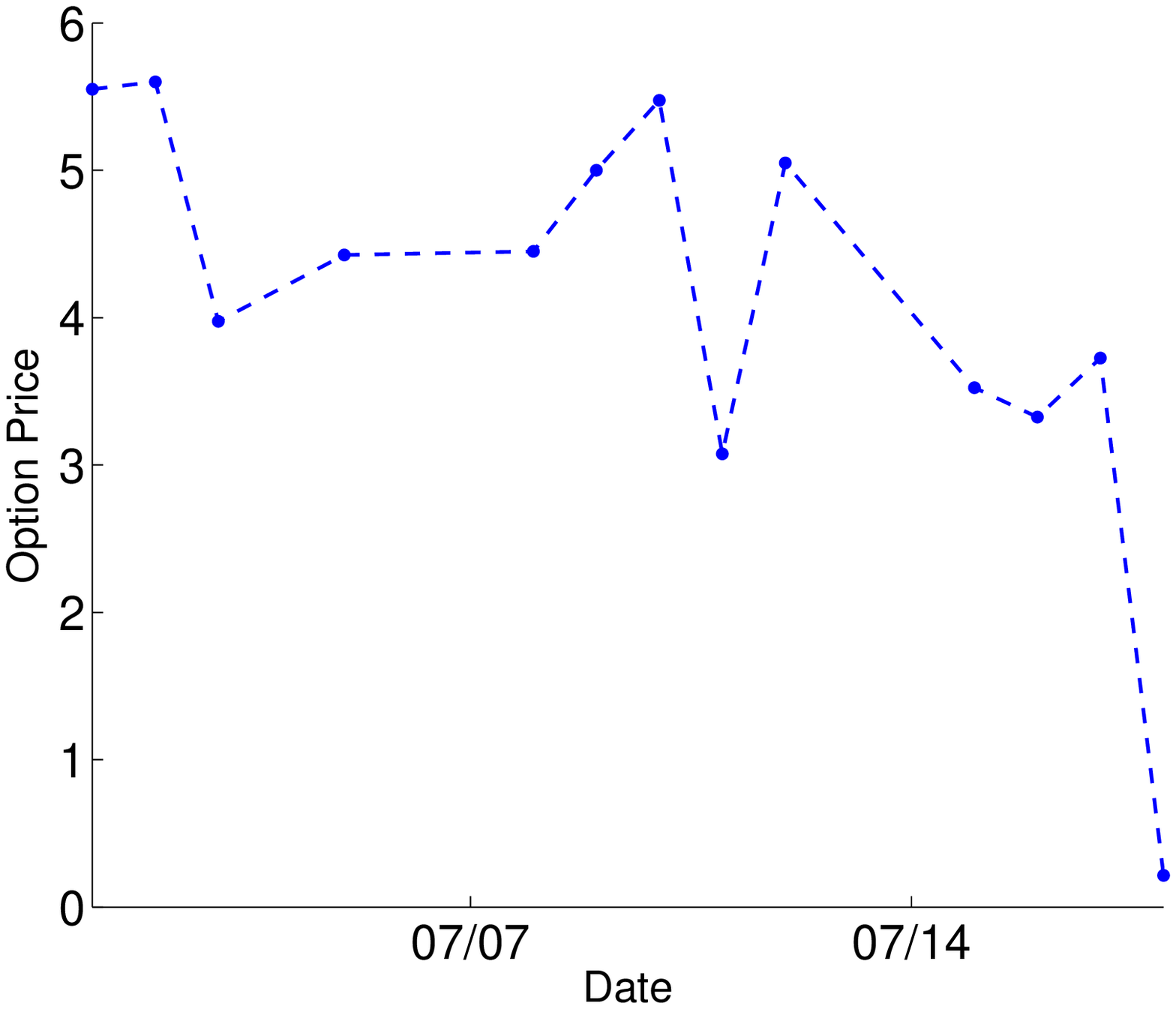}   \includegraphics[scale=0.43]{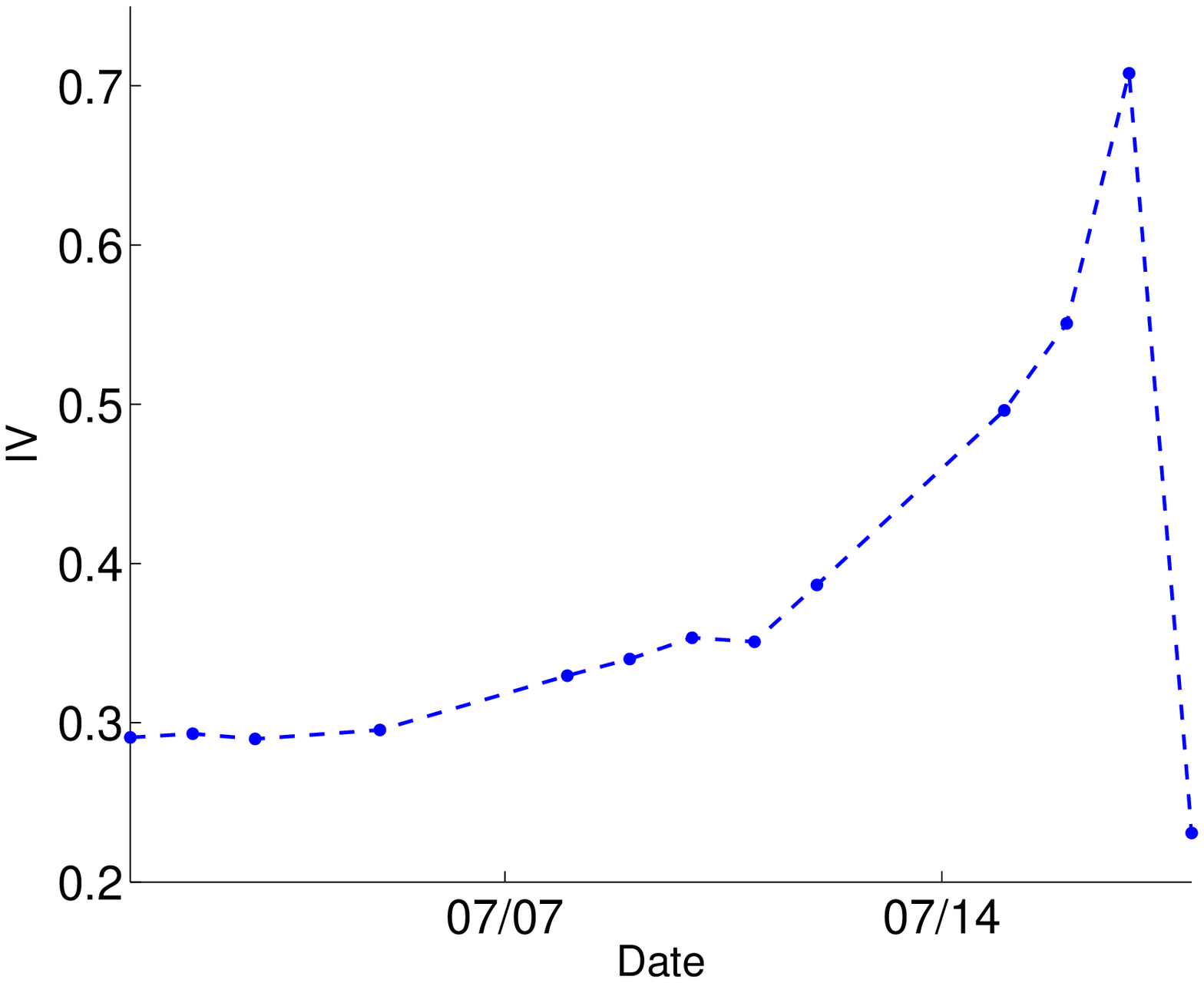}\\
\caption{\small The price (left) and implied volatility (right) of the front month ATM
option  on IBM with expiration date July 18, 2013.  The earnings announcement date is July 17, 2013.\label{fig:vol_example}
}\end{figure}

\begin{figure}[th]
\centering{}\includegraphics[scale=0.44]{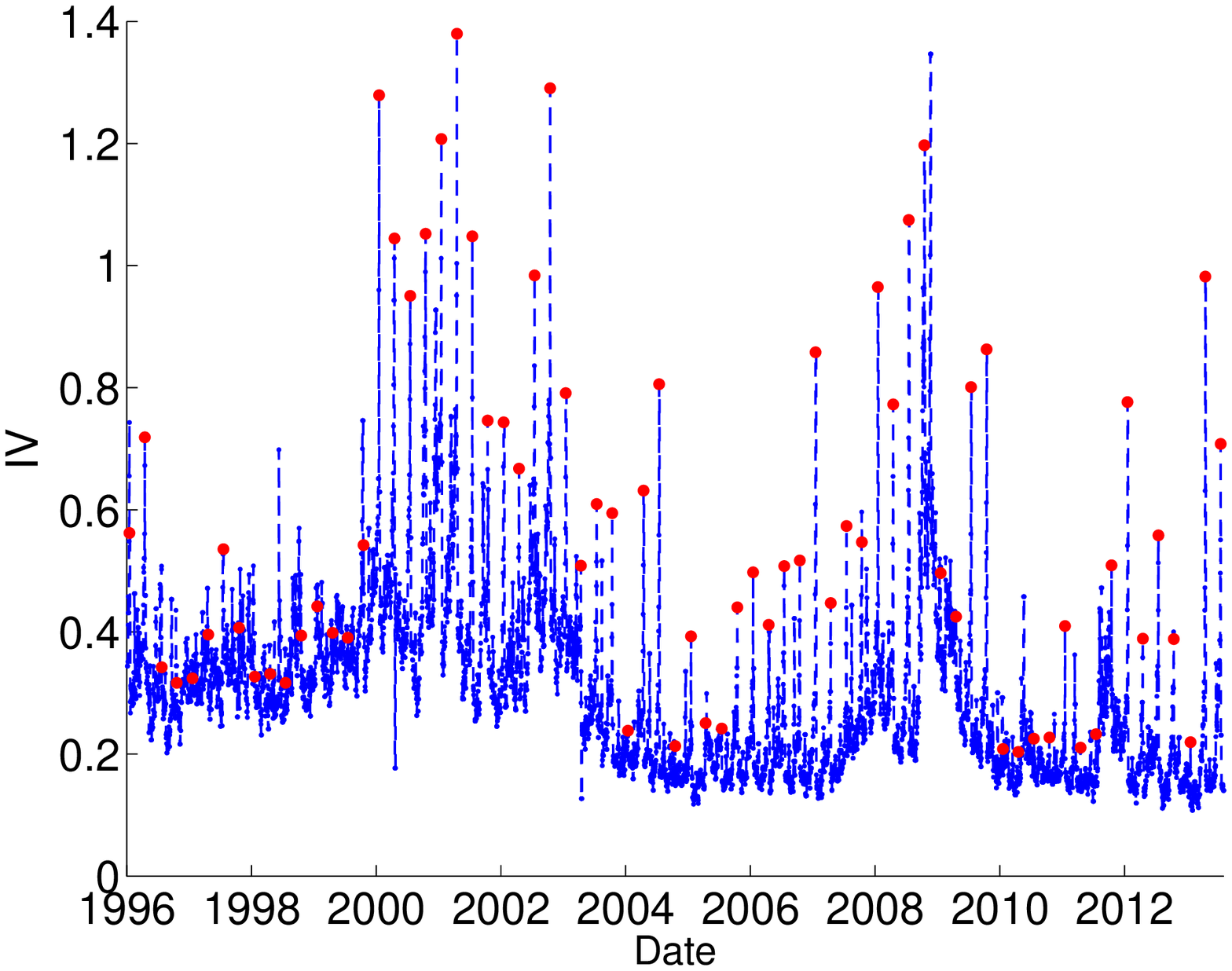}\includegraphics[scale=0.43]{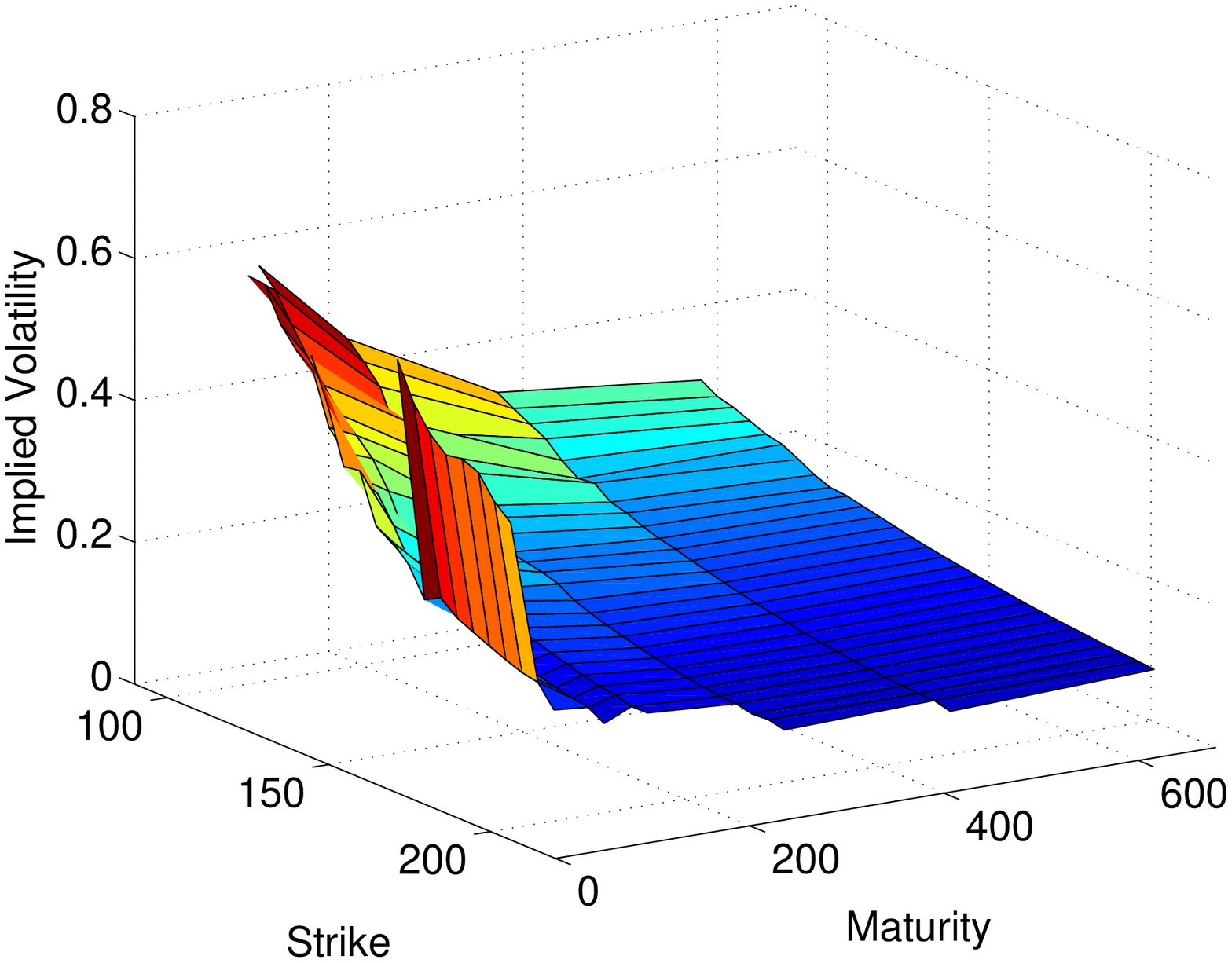}\caption{\small The time series (left) of the IV of the front month
ATM call on IBM. The red dots mark the value of the IV on EA dates. The IV surface (right) of IBM at
the market close on July 15, 2013, two days before an EA.
\label{fig:IV_time_series} \label{fig:Cal_example_surface}}
\end{figure}

\subsection{Objectives and Related Studies\label{sub:related_literature}}
These market observations call for a  pricing framework  that  explicitly incorporates the jump to be realized on the earnings announcement date.   In this paper, we analyze an option pricing framework that accounts for the price impact of an earnings announcement, with emphasis on the  behavior of the implied volatility  prior to  the event. Specifically, we   introduce  a random-sized jump scheduled on the earnings announcement date to represent the shock to the company stock price. The introduction of a jump related to earnings is consistent
with prior  empirical studies (\citet{Maheu2004,Piazzesi2005,Dubinsky2006,Lee2008}). We first apply this idea to the \cite{Black1973} (BS) model, where we  obtain formulas for  the pre-EA European option price, its IV and various Greeks.  In this model, we find that the IV increases with a specific pattern as time approaches the announcement date.  We further incorporate additional  random-time jumps in the stock price by extending the \cite{Kou2002} double-exponential (DE) jump-diffusion model to incorporate an EA jump. Specifically, we assume that the EA jump size for the log stock price is either a DE or Gaussian random variable.  In both cases,  analytic option pricing formulas are derived and used to generate  IVs.  
 
To better understand the behavior of the PEAIV, we  derive analytic bounds and asymptotics for extreme strikes for the IV under a general class of stock price models with an EA jump.  In the bounds, we observe the explicit dependence on EA jump parameters and a common time-dependence of the PEAIV across these models.  In addition, we  calibrate a number of models, namely, the extended Heston and Kou models with an EA jump,  to market option prices.  Our results demonstrate a more adequate fit of the entire PEAIV surface, as compared to these models without the EA jump. Furthermore, we    extract  the risk-neutral  distribution of the EA jump and compare it to its historical distribution.  The discrepancy between the two distributions  can be interpreted as a risk premium
for the  jump risk due to earnings. We conclude the paper analyzing the pricing of American options before and after an EA date. For implementation, we  discuss  a numerical  method, as well as an analytical approximation  based on the  work of \cite{Barone-Adesi1987}.

There is a wealth of literature in empirical  finance and accounting that examines how earnings announcements affect stock
and option prices.  Most of the studies focus on  the informational
content, market reactions and  price patterns associated with  earnings releases. Systematic reviews of the empirical studies  can be found in  \citet{Dubinsky2006},
\citet{Barth2011}, \citet{Rogers2009} and \citet{Billings2010}. For the price impact of news on futures, we refer to \cite{Kiseop} and references therein.   Below, we summarize a number of related studies that focus on option prices  prior to an earnings announcement. 

The model by  \citet{Patell1981} (PW)  is among  the first empirical studies on  the prices of  options in the presence of an earnings announcement. They build on the pricing models of \cite{Black1973}  and \citet{Merton1973}, and assume the instantaneous volatility
as a deterministic piece-wise constant  function of time, with   two volatility levels over  two time windows around the earnings release in order to  reflect the uncertainty surrounding  the earnings impact.  Since their seminal work, other papers have adopted the same model or its variations and performed  empirical tests
(see \citet{Donders1996}, \citet{Isakov2001}, \citet{Barth2011}, among others). In particular, \citet{Barth2011}
extend PW's model by allowing the deterministic piece-wise constant volatility to take different values before, during, and after the announcement date.  In all these empirical studies, there are no jumps incorporated in the stock price at anytime.  Moreover, without stochastic volatility or jumps in the  stock price, the IV will not exhibit any skew, though this  is commonly observed in the options market. 

In this paper, we adopt  a different approach by incorporating a deterministic-time random jump in the  stock price dynamics. We consider  different  distributions of   the scheduled
earnings jump as well as different underlying stock price dynamics before and after the EA date.  Our study generalizes that of  \citet{Dubinsky2006}, which  focuses on   the BS  model with a Gaussian EA jump and provide a formula for the time-deterministic  implied volatility function.    Our main contributions include analytical option pricing formulas when the underlying stock price admits doubly exponential jumps over time (the Kou model), along with a double-exponential or  Gaussian EA   jump on the   announcement date (see Theorem \ref{prop:KOU_extension}). Furthermore, we  derive  analytic  bounds and asymptotics for the pre-EA implied volatility under various models (see Propositions \ref{prop:lower_bound} and \ref{prop:upper_bound}). In addition, we study the  pricing of pre-earnings  American options  under the Kou model and  discuss an analytic approximation to the    option price by extending the method of  \cite{Barone-Adesi1987}.

  Our analysis  draws motivation from  empirical studies and market observations to  introduce more sophisticated models that  help reproduce   the entire implied volatility surface  accurately  
and tractably.   For  calibration, we employ recent  options  data  from the whole observed IV surface, instead of ATM only options used in \citet{Dubinsky2006}.  Using earnings data from 1994 to 2013, we   also compare  the risk-neutral  distribution of the EA jump  to its historical distribution. The framework and methods  introduced herein can be readily  generalized to other option pricing models. They can also be useful for monitoring the PEAIV, and  estimating the historical and risk-neutral  distributions of the EA jump and the associated risk premium.

The rest of the paper is structured  as follows. In Section \ref{sec:BS}, we extend the Black-Scholes model by  incorporating an  EA jump in the  stock price, and derive the corresponding implied volatility function and Greeks. In Section \ref{sec:other_models}, we discuss the extensions of  other models, including  the Kou and Heston models. In Section \ref{sec:IV}, we study the  bounds and  asymptotics  for the IV surface under various models.  In Section \ref{sec:calibration}, we discuss our calibration results from  the extended BS,
Kou, and Heston models, and compare the risk-neutral  and historical distributions of  EA jumps. In Section \ref{sec:american}, we study  the valuation of
American options prior to an earnings announcement. The Appendix contains the proofs and details for our analytical results.

\section{Black-Scholes Model with an EA Jump\label{sec:BS}}\label{sect-BS}
The effect of a scheduled earnings  announcement on the stock price is modeled by a deterministic-time
random-size price  jump in the price dynamics. Naturally, this modification can be applied to virtually any model, and we start with  an extension of the  \citet{Black1973} model.

Let     $(W_{t})_{t \ge 0}$  be a standard Brownian motion defined
on the risk-neutral probability space $\left(\Omega,\mathcal{F},\mathbb{Q}\right)$. Throughout, we assume a constant risk-free rate $r\ge 0$.  Let $T_{e}$ be the
EA date, and we call the r.v.  $Z_{e}$ the EA jump, which is the jump size of the  log stock price immediately after the earnings announcement. We assume  $Z_{e}$ to be  independent of
$W$. Under  the risk-neutral measure, the company stock price $(S_{t})_{t \ge 0}$ evolves according to 
\begin{equation}
\frac{dS_{t}}{S_{t_{-}}}=rdt+\sigma dW_{t}+d\left(1_{\left\{ t\geq T_{e}\right\} }\left(e^{Z_{e}}-1\right)\right),\label{eq:GBMe}
\end{equation}
where $1_{\left\{ \cdot\right\} }$ denotes the indicator function,
$r$ is  the positive  constant interest rate, and $\sigma$
represents the constant volatility. The martingale condition $S_0=\mathbb{E}\left\{ e^{-rt}S_{t}\right\} ,\ 0\leq t\leq T,$ where $\mathbb{E}\left\{ \cdot\right\} $ denotes the expectation
under $\mathbb{Q}$, implies that $\mathbb{E}\left\{ e^{Z_{e}}\right\} =1$. 

In this section, we assume that $Z_{e}$ is \textit{normally} distributed.
This yields closed form expressions for the price and the IV
surface. The martingale
condition $\mathbb{E}\left\{ e^{Z_{e}}\right\} =1$ implies that $Z_{e}\sim N\left(-\frac{\sigma_{e}^{2}}{2},\sigma_{e}^{2}\right)$, so the EA jump is parametrized by $\sigma_{e}$ only. We notice
that, for $T\geq T_{e}$, 
\begin{equation}
{\log}\left(\frac{S_{T}}{S_{t}}\right) \sim N\left(\left(r-\frac{\sigma^{2}}{2}-\frac{\sigma_{e}^{2}}{2(T-t)}\right)(T-t)\,,\,{\sigma^{2}(T-t)+\sigma_{e}^{2}}\right).
\end{equation}
Therefore, the price of a European call with strike $K$ and maturity $T$ is
given by
\begin{equation}
C\left(t,S_{t}\right)=C_{BS}\left(T-t,S_{t};\sqrt{\sigma^{2}+\frac{\sigma_{e}^{2}}{T-t}},K,r\right),\quad0\leq t<T_{e},\label{eq:GBMe_price}
\end{equation}
where $C_{BS}\left(\tau,S;\sigma,K,r\right)$ represents the usual
BS formula with time to maturity $\tau$ and spot price $S$.

\subsection{Implied Volatility\label{sub:BS_IV}}
The
price formula (\ref{eq:GBMe_price}) allows us to express the implied
volatility (IV) as the deterministic function
\begin{equation}
I\left(t;K,T\right)=\begin{cases}
\sqrt{\sigma^{2}+\frac{\sigma_{e}^{2}}{T-t}} & \quad0\leq t<T_{e},\\
\sigma & \quad T_{e}\leq t<T.
\end{cases}\label{eq:BS_IV_ts}
\end{equation}
As we can see, although the IV surface, for a fixed time $t$, is
flat across strikes, it has a decreasing term structure in $T$. In
addition, the IV for any particular option increases in time as we
approach the earnings announcement. Alternatively, the option price
formula (\ref{eq:GBMe_price}) can be obtained as if the stock has
no jump and follows the dynamics 
\[
\frac{dS_{t}}{S_{t}}=rdt+I\left(t;K,T\right)dW_{t}.
\]
This model  is used in \citet{Patell1981}, who assume $I\left(t;K,T\right)$
to be a deterministic, piece-wise constant function to reflect the
impact of earnings on the implied volatility.

In Figure \ref{fig:BS_cal_example} (left), we compare the IV function
in (\ref{eq:BS_IV_ts}) for fixed $\left(K,T\right)$, against the
IV time series of the front month ATM IBM option\footnote{The IVs of the ATM call and put are observed to be very close but not identical. Each point of the time series represents the   average of the two IVs. } with expiration date 7/18/2013. The parameters of the model have been
chosen by a least-square regression. As we can see, the observed ATM IV
and the model IV increase in a similar fashion. In the same figure (right), we
also compare the IV function in (\ref{eq:BS_IV_ts}) for fixed $\left(t,K\right)$,
against the term structure of the IV of the ATM IBM call option
on 7/17/2013---just before the earnings announcement and one day before
the expiration date. Again, the parameters of the model are obtained via an additional least-square regression and they differ from the one obtained from the IV time series analysis. The model term structure approximates the observed one even if the IV surface presents no skew and a term structure which allows only for the functional
form (\ref{eq:BS_IV_ts}). 

\begin{figure}[th] \centering{}
\includegraphics[scale=0.42]{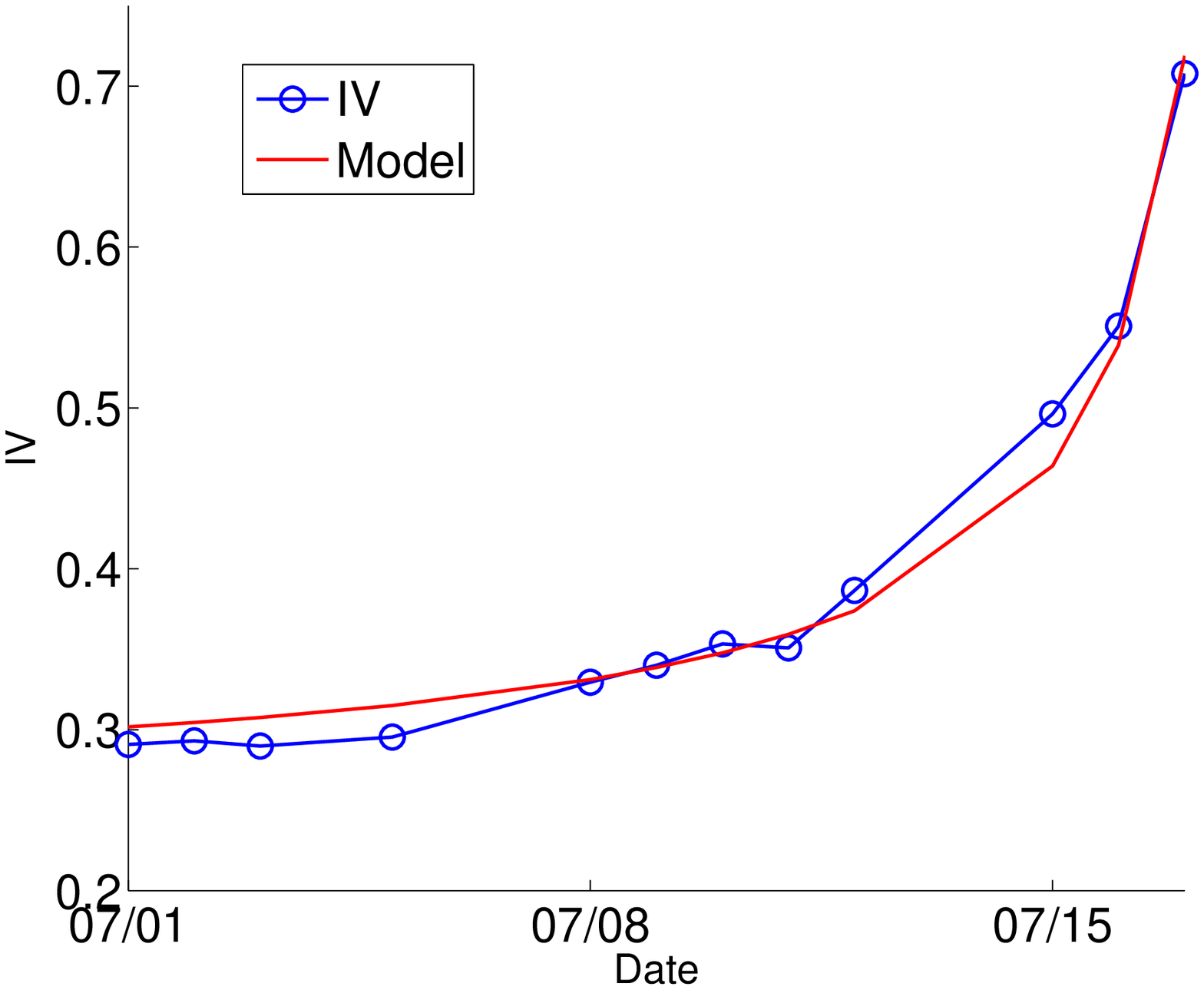}\includegraphics[scale=0.42]{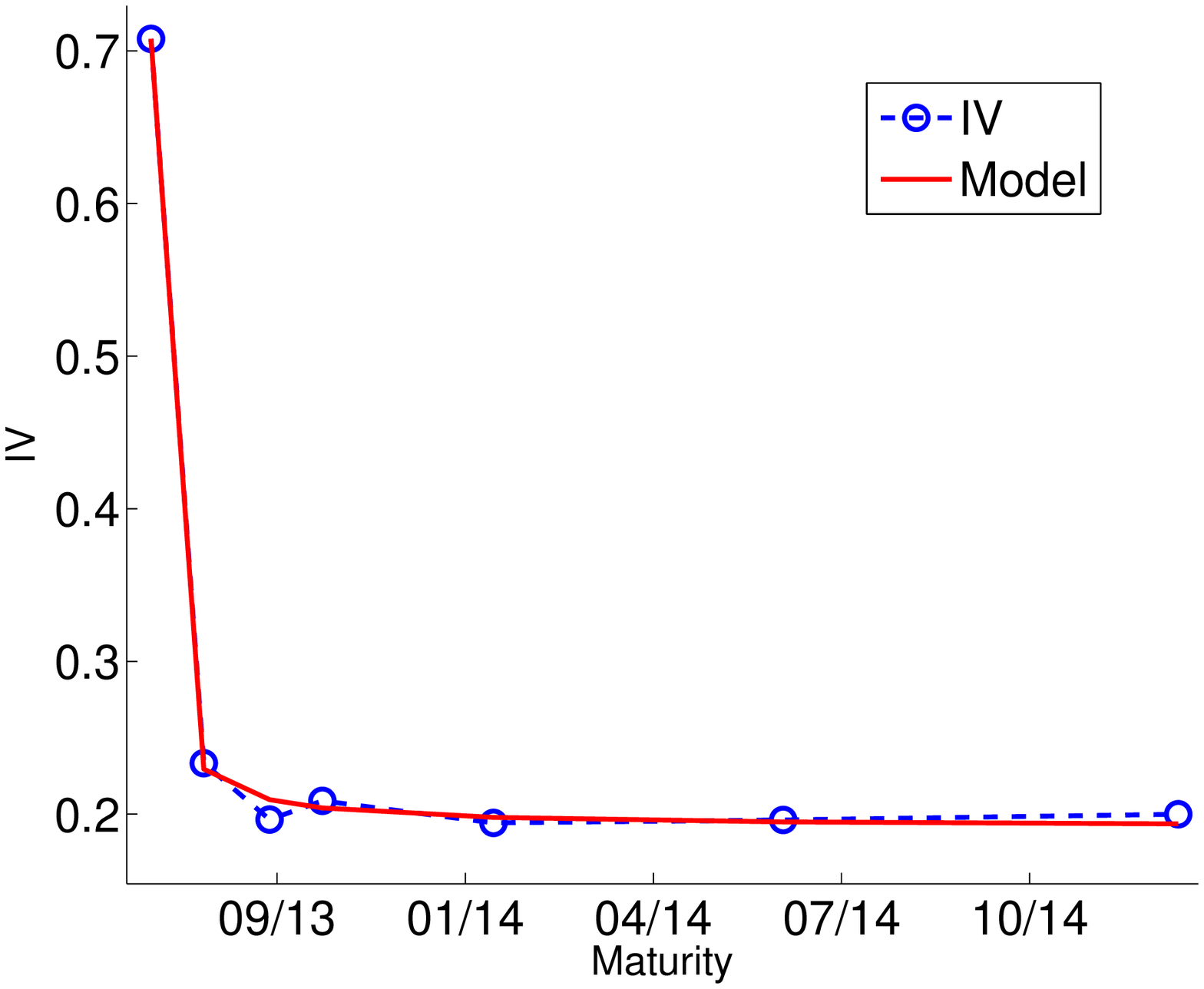}\caption{\small \small{Left: The IV function (\ref{eq:BS_IV_ts}) (solid line) prior to an EA, with $\sigma=0.2538,\ \sigma_{e}=0.0424$ and
expiration date 7/18/2013, as compared  to the market-observed IV (circle)
of the front month IBM ATM option with the same expiration
date. Right: The term structure of the IV (solid) according to (\ref{eq:BS_IV_ts}),
with $\sigma=0.1912,\ \sigma_{e}=0.0429$, on 7/17/2013, compared
to the market-observed term structure (circle) for ATM IBM
options on the same date.\label{fig:BS_cal_example}}}
\end{figure}

%\[
%\Delta =\Delta_{BS}\left(T-t,S_{t};I\left(t\right),K,r\right), \quad \Gamma\equiv\frac{\partial^{2}C}{\partial S^{2}}=\Gamma_{BS}\left(T-t,S_{t};I\left(t\right),K,r\right),\]
%where $\Delta_{BS}\left(\tau,S;\sigma,K,r\right)=\Phi\left(d_{1}\right)$
%denotes the usual Black-Scholes Delta function, with $\Phi\left(\cdot\right)$
%denoting the standard normal  c.d.f., and $\Gamma_{BS}\left(\tau,S;\sigma,K,r\right)=\frac{1}{S\sigma\sqrt{\tau}}\phi\left(d_{1}\right)$, with $d_{1}\equiv\frac{{log}\left(\frac{S}{K}\right)+\left(r+\sigma^{2}\right)\tau}{\sigma\sqrt{\tau}}$
%These Greeks are computed  Black-Scholes 

%This Delta is applicable to hedging the option up to, but not including,
%$T_{e}$. In this case, the functional form of the Delta is thus the
%same of the

% Black-Scholes one but the volatility at which it is calculated
%is in fact the observed implied volatility, which is higher than that
%of volatility driving the diffusion in (\ref{eq:GBMe}). A similar result holds for \textbf{Gamma}
%\[
%,\]
%where denotes the Black-Scholes gamma, with $\phi\left(\cdot\right)$ denoting
%the standard normal PDF. 

\subsection{Greeks\label{sub:BS_Greeks}}
In order to understand the sensitivity of options prices approaching
the earnings announcement, we derive and analyze the Greeks based
on the price function (\ref{eq:GBMe_price}). 

For a call or put with maturity $T$ after the EA date $T_{e}$, the  {Delta} and  {Gamma} are the simply the Black-Scholes Delta and Gamma functions, but with  the volatility parameter  set at higher value $I(t)$, for $t< T_e$.  On the other hand, we notice that there are two parameters related to the volatility
of the stock price. Consequently, in addition to the standard Black-Scholes   {Vega}, 
%\[
%\mathcal{V}\equiv\frac{\partial C}{\partial\sigma}=\frac{1}{\sqrt{1+\frac{\sigma_{e}^{2}}{\sigma^{2}\left(T-t\right)}}}\mathcal{V}_{BS}\left(T-t,S;I\left(t\right),K,r\right), 
%\] for $S>0$ and $0\leq t<T$, 
we   introduce the \emph{EA-Vega}
\begin{eqnarray*}
\mathcal{V}_{e} & \equiv & \frac{\partial C}{\partial\sigma_{e}}=\frac{1}{\sqrt{T-t}\sqrt{\frac{\sigma^{2}\left(T-t\right)}{\sigma_{e}^{2}}+1}}\mathcal{V}_{BS}\left(T-t,S;I\left(t\right),K,r\right),
\end{eqnarray*}
where $\mathcal{V}_{BS}\left(\tau,S;\sigma,K,r\right)=S\phi\left(d_{1}\right)\sqrt{\tau}$
represents the usual Black-Scholes Vega function with   spot price $S$ and time-to-maturity $\tau$.

For the  Theta of a call, we obtain
\begin{equation}
\Theta\equiv\frac{\partial C}{\partial t}=\Theta_{BS}\left(T-t,S;I\left(t\right),K,r\right)+\frac{1}{2I(
t)}\left(\frac{\sigma_{e}}{T-t}\right)^{2}\mathcal{V}_{BS}\left(T-t,S;I(t),K,r\right),\quad0\leq t<T.\label{eq:GBMe_theta}
\end{equation}
where $\Theta_{BS}\left(\tau ,S;\sigma,K,r\right)=-S\frac{\sigma}{2\sqrt{\tau}}\phi\left(d_{1}\right)-rKe^{-r\tau}\Phi\left(d_{2}\right)$
is  the Black-Scholes Theta function. First, we note that  $\Theta_{BS}\left(T-t,S;I(t),K,r\right)\leq\Theta\leq0.$
The left part of the inequality is a direct consequence of (\ref{eq:GBMe_theta})
while the right part is a consequence of the fact that the option
price is decreasing in time. On the other hand,  it is not true in general that $\Theta\leq\Theta_{BS}\left(S,K,r,T-t,\sigma\right).$
This means that the option may   lose   value less rapidly  over time 
than an option priced with a lower volatility.  To  illustrate this point, we suppose  $r=0$, and    the Black-Scholes PDE implies that 
\[
\frac{\partial C}{\partial t}=-\sigma^{2}S^{2}\frac{\partial^{2}C}{\partial S^{2}}=-\sigma^{2}S^{2}\Gamma_{BS}\left(T-t,S;I(t),K,0\right).
\]
On the other hand, we also have $\Theta_{BS}\left(T-t,S;\sigma,K,0\right)=-\sigma^{2}S^{2}\Gamma_{BS}\left(T-t,S;\sigma,K,0\right).$
Therefore,
\[
\frac{\partial C}{\partial t}-\Theta_{BS}\left(T-t,S;\sigma,K,0\right)=-\sigma^{2}S^{2}\left(\Gamma_{BS}\left(T-t,S;I(t),K,0\right)-\Gamma_{BS}\left(T-t,S;\sigma,K,0\right)\right).
\]
For ATM options, we have $\Gamma_{BS}\left(T-t,S;\sigma,S,0\right)=\frac{1}{\sqrt{2\pi}\sigma S}e^{-\frac{\sigma{}^{2}(T-t)}{2}},$
which is decreasing in $\sigma$.  This implies that $\frac{\partial C}{\partial t}>\Theta_{BS}\left(T-t,S;\sigma,S,0\right)$ for 
 $\sigma_{e}>0$, implying   a less rapid time decay in the option value. In fact, the same is true for  other puts and calls whose $\Gamma_{BS}$ is  decreasing in $\sigma$. 

In Figure \ref{fig:Greeks} (left),  we illustrate   $\Theta$ as a function of  the spot price for a call  with 5 days to maturity.  For comparison, we plot two additional benchmarks based on the BS Theta with different volatility parameter values $I(0)$ and $\sigma$, respectively. As expected from (\ref{eq:GBMe_theta}), the  time-decay of the call price is  less severe than  $\Theta_{BS}\left(T,S;I(0),K,r\right)$. In addition, for spot prices near  the strike $K$, we observe that  $\Theta>\Theta_{BS}\left(T,S;\sigma,S,r\right)>\Theta_{BS}\left(T,S;I(0),S,r\right)$.
Figure \ref{fig:Greeks} (right) shows how $\Theta$ changes in time for an
ATM ($S=K=100$) call, up to one day before  earnings. Again, we notice the same dominance of Theta's, but their differences increase as time approaches the EA date. Interestingly,   $\Theta$, which accounts for the EA jump, appears to be significantly more  constant as compared to  the other two BS Theta's. 
\\
\begin{figure}[H]
\centering{}\includegraphics[scale=0.45]{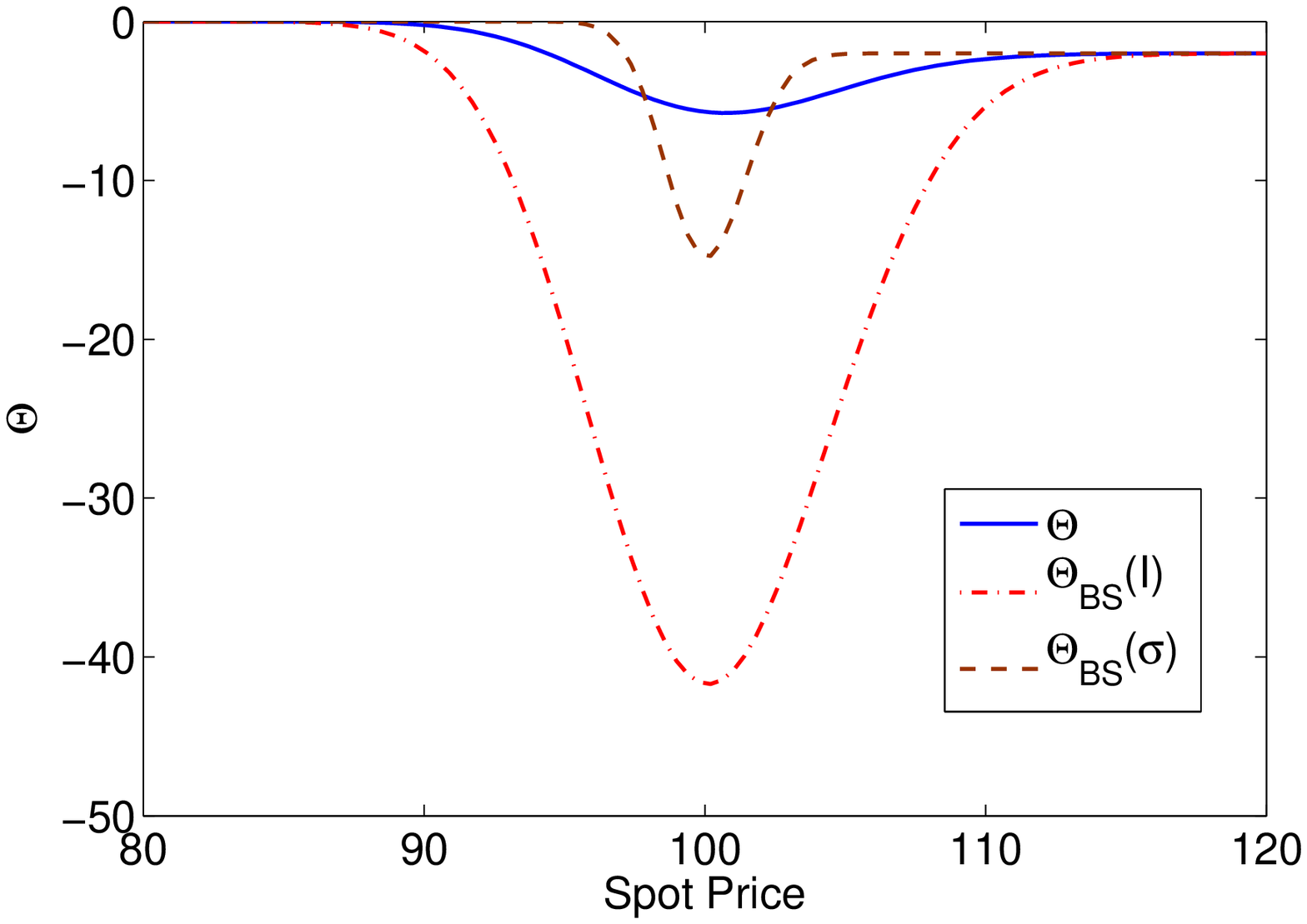}\includegraphics[scale=0.45]{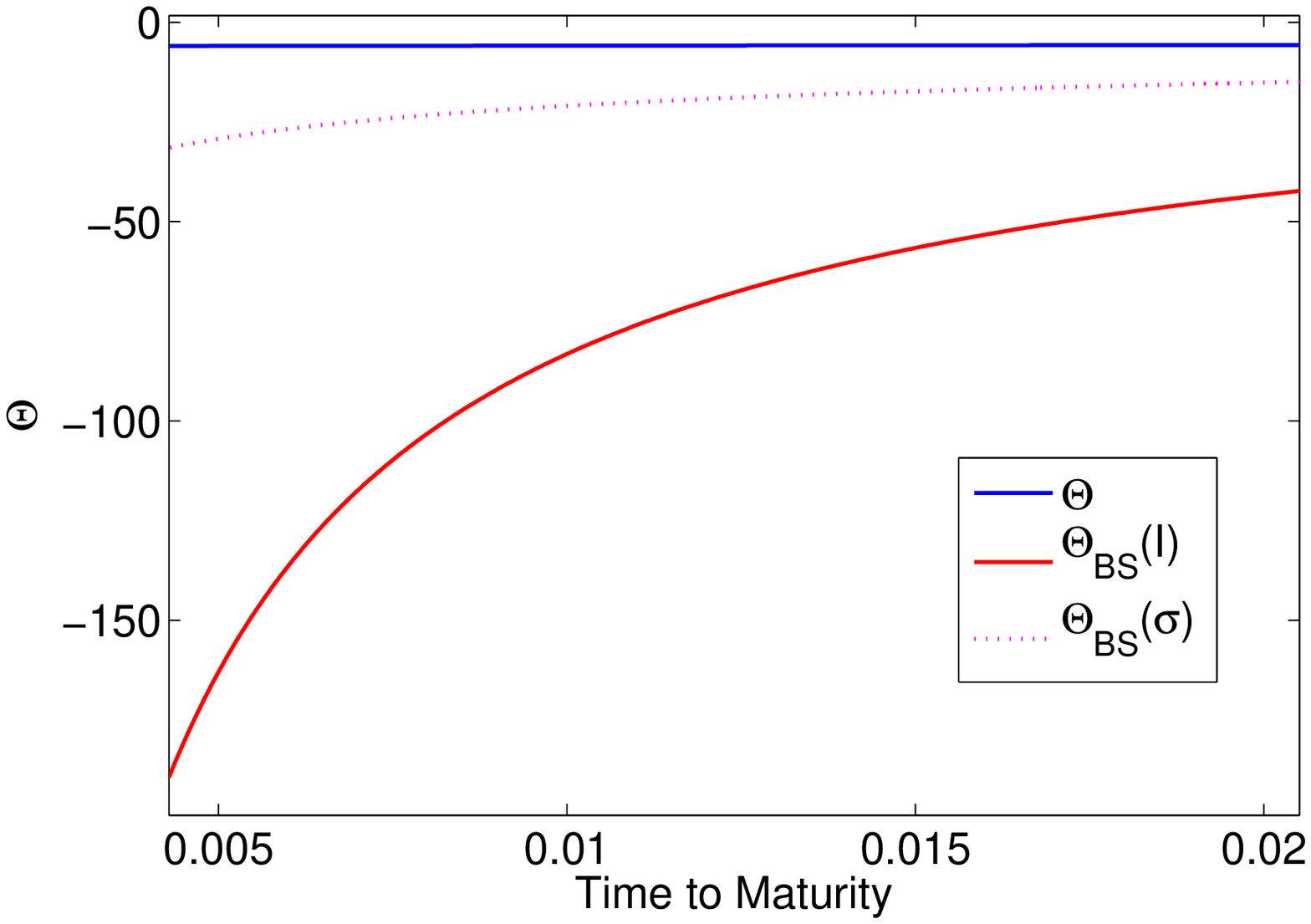}\\
\caption{\small{The shape of $\Theta$ and $\Theta_{BS}$ changes
w.r.t. the spot price (with $T=\frac{5}{252}$) and time to maturity (with $S=K=100$). In both panels, we see that $\Theta$ for ATM options is less negative than the ordinary BS Theta $\Theta_{BS}$. Common parameters: $r=0.02,\ \sigma=0.1,\ \sigma_{e}=0.04$.} \label{fig:Greeks}}
\end{figure}

\section{Incorporating EA Jumps into Other Models\label{sec:other_models}}
Although the extended BS model (\ref{eq:GBMe}) is capable of showing an increasing IV approaching an earnings announcement, the IV 
has no skew and its term structure only admits  the  particular two-parameter functional form $\sqrt{\sigma^{2}+\frac{\sigma_{e}^{2}}{T-t}}$. On the other hand, in addition to the scheduled jump, the stock price may also experience randomly timed jumps which cannot be adequately captured by diffusion.  This has motivated many models that incorporate jumps of various distributions  into the stock price dynamics, with notable examples in the \cite{Merton1976}, \cite{Kou2002}, Variance Gamma (\cite{Madan1998}), CGMY (\cite{Carr2002})) models. In this section, we  present an  analytic formula for  pricing European options under an extension  of the jump-diffusion model of \citet{Kou2002}, and also discuss the  pricing of options prior to an earnings announcement using a transform method. 

\subsection{Extensions of the Kou Model\label{sub:KOU}}
We consider an extension of the Kou model and derive an
analytic formula for the price of a European option with an EA jump. Under this model, the risk-neutral terminal stock price follows
\begin{equation}
{\log}\left(\frac{S_{T}}{S_{0}}\right)=-\frac{\sigma^{2}}{2}T+\sigma W_{T}-m \kappa  T+\sum_{i=1}^{N_{T}}J_{i}+1_{\left\{ t\geq T_{e}\right\} }\left(Z_{e}-{\log}\left(\mathbb{E}\left\{ e^{Z_{e}}\right\} \right)\right),\label{eq:dedyn}
\end{equation}
where each jump $J_{i}\sim DE\left(p,\lambda_{1},\lambda_{2}\right)$ is  \textit{double-exponentially} distributed with the p.d.f.
$f_{J_{i}}\left(x\right)=1_{\{x\geq 0\}} p \lambda_{1}e^{-\lambda_{1}x}+1_{\{x\leq 0\}}\left(1-p\right)\lambda_{2}e^{\lambda_{2}x}$. The number of randomly timed jumps is modeled by 
$N_T \sim \text{Poi}(\kappa T)$. To ensure that the martingale condition holds,  we set  $m=p\frac{\lambda_{1}}{\lambda_{1}-1}+(1-p)\frac{\lambda_{2}}{\lambda_{2}+1}-1$.
 We shall consider separately two distributions for the  EA log jump size $Z_{e}$: (i) double-exponential $Z_{e}\sim DE\left(u,\eta_{1},\eta_{2}\right)$, and (ii)  Gaussian $Z_{e}\sim N\left(0,\sigma_{e}^{2}\right)$. 

Let $C\left(S\right)=\mathbb{E}\left\{ e^{-rT}\left(S_{T}-K\right)^{+}\vert S_{0}=S\right\} $
be the price of the European call option at time $t=0$ when the spot
price is   $S$. We now present the price  formula  when the EA jump $Z_e$  is a double exponential random variable.
\begin{prop}
\label{prop:KOU_extension}Suppose the terminal stock price  follows
(\ref{eq:dedyn}), with $Z_{e}\sim DE\left(u,\eta_{1},\eta_{2}\right)$
and $\eta_{1}\neq\lambda_{1},\ \eta_{2}\neq\lambda_{2}$. Then, the European call option price   is given by 
\begin{align}
C(S)&=e^{-\alpha}S\Upsilon\left(S,K,r+\sigma^{2},T,\sigma,\hat{\kappa},\hat{p},\hat{\lambda}_{1},\hat{\lambda}_{2},\hat{u},\hat{w},\hat{\eta}_{1},\hat{\eta}_{2}\right)\notag\\
&\quad  -Ke^{-rT}\Upsilon\left(S,K,r,T,\sigma,\kappa,p,\lambda_{1},\lambda_{2},u,w,\eta_{1},\eta_{2}\right),\label{eq:KOU_extension_formulae}
\end{align}
 where the function $\Upsilon$ is given in Appendix A.1, and  the constants  are
\begin{align*}
\hat{\eta}_{1}=\eta_{1}-1,\quad\hat{\eta}_{2}&=\eta_{2}+1,\quad
\hat{\lambda}_{1}=\lambda_{1}-1  ,\quad\hat{\lambda}_{2}=\lambda_{2}+1, \quad \hat{u}=u\frac{\eta_{1}}{\eta_{1}-1}  ,\quad w=1-u, \quad \hat{w}=w\frac{\eta_{2}}{\eta_{2}+1}, \\  \hat{\kappa}&=\left(m+1\right)\kappa,\quad\hat{p}=\frac{\lambda_{1}}{\lambda_{1}-1}\frac{p}{m+1}, \quad 
\alpha={\log}\left(u\frac{\eta_{1}}{\eta_{1}-1}+w\frac{\eta_{2}}{\eta_{2}+1}\right).
\end{align*}
\end{prop}

\vspace{10pt}

When the jump size parameters of the EA and randomly timed jumps are the same, i.e. $\eta_{1}=\lambda_{1}$ or $\eta_{2}=\lambda_{2}$, an analytic pricing formula can also be  derived but omitted here. In practice, the   parameters $\lambda_1$, $\lambda_2$ are usually  of an order of magnitude greater
than $\eta_1$ and $\eta_2$, as we will observe from our  calibration in Section \ref{sec:calibration}. \\

 Alternatively, if the EA jump  is \textit{normally} distributed, i.e. $Z_{e}\sim N\left(0,\sigma_{e}^{2}\right)$, then one can directly adapt   the result from \cite{Kou2002} to account for the EA jump. Specifically, the EA jump parameter $\sigma_e$ can be incorporated in the  volatility coefficient of the Brownian motion $W_T$:
\[
{\log}\left(\frac{S_{T}}{S_{0}}\right)\overset{d}{=}\left(-\frac{\sigma^{2}+\frac{\sigma_{e}^{2}}{T}}{2}-m\kappa\right)T+\sqrt{\sigma^{2}+\frac{\sigma_{e}^{2}}{T}}\,W_{T}+\sum_{i=1}^{N_{T}}J_{i}\,.
\] 
From this, we   see that the resulting analytic formula is in fact identical to the original formula without the EA jump, except with $\sigma$ replaced by $\sqrt{\sigma^{2}+\frac{\sigma_{e}^{2}}{T}}$. \\

The analytic formula (\ref{eq:KOU_extension_formulae}) allows for the
fast computation of the price, and simultaneously its delta, $\Delta=e^{-\alpha}\Upsilon\left(\cdot\right)$,
where $\Upsilon\left(\cdot\right)$ is the first term on RHS of
(\ref{eq:KOU_extension_formulae}). In Table
\ref{tab:KOU_extension}, we apply (\ref{eq:KOU_extension_formulae}) to calculate option prices and we report the corresponding implied volatilities. As we can see, the IV decreases as time to maturity increases, which is typical of the IV before an earnings announcement. In addition,  the IV skew becomes  flatter as  maturity increases. In Figure \ref{fig:Kou_rising_IV}, we plot the IV of the ATM option when the underlying follows (\ref{eq:dedyn}) with either a Gaussian or DE EA jump. For comparison, we have chosen the jumps parameters so that the variances of the Gaussian and DE EA jumps coincide. We notice that the IVs have similar  dependence on time. Comparing Figures  \ref{fig:BS_cal_example} and \ref{fig:Kou_rising_IV}, it is natural to wonder whether the   IV increases in time in a similar fashion, not only  in the extended Black-Scholes and Kou, but also in other  models. This  motivates us to explore the properties of the IV under different models in Section \ref{sec:IV}.

% Table generated by Excel2LaTeX from sheet 'Sheet1'
\begin{table}[htbp]
  \centering
  \footnotesize
    \begin{tabular}{r|rrrrrrrrr}

    \hline
    $K$\,\textbackslash\,$T$ & \multicolumn{2}{c}{1 Week}        & \multicolumn{2}{c}{1 Month}         & \multicolumn{2}{c}{3 Months}              & \multicolumn{2}{c}{1 Year}               \\
   \hline
         & \multicolumn{1}{c}{$C$} & \multicolumn{1}{c}{$IV$} & \multicolumn{1}{c}{$C$} & \multicolumn{1}{c}{$IV$} & \multicolumn{1}{c}{$C$} & \multicolumn{1}{c}{$IV$} & \multicolumn{1}{c}{$C$} & \multicolumn{1}{c}{$IV$} \\

\hline

   %       &  &  &  & &  &  &  & &     \\
    90    & 11.031 & 0.799 & 11.380 & 0.425 & 12.348 & 0.300 & 16.050 & 0.239 \\
    92.5  & 8.958 & 0.767 & 9.400 & 0.415 & 10.529 & 0.297 & 14.485 & 0.239 \\
    95    & 7.048 & 0.738 & 7.598 & 0.407 & 8.871 & 0.295 & 13.027 & 0.239 \\
    97.5  & 5.357 & 0.714 & 6.007 & 0.401 & 7.384 & 0.294 & 11.675 & 0.239 \\
    100   & 3.945 & 0.699 & 4.651 & 0.397 & 6.075 & 0.293 & 10.428 & 0.239 \\
    102.5 & 2.849 & 0.696 & 3.536 & 0.396 & 4.942 & 0.292 & 9.284 & 0.239 \\
    105   & 2.047 & 0.704 & 2.651 & 0.396 & 3.979 & 0.292 & 8.240 & 0.238 \\
    107.5 & 1.475 & 0.719 & 1.968 & 0.399 & 3.173 & 0.292 & 7.292 & 0.238 \\
    110   & 1.069 & 0.738 & 1.455 & 0.403 & 2.509 & 0.293 & 6.434 & 0.238 \\

\hline
    \end{tabular}%
  \label{tab:KOU_extension}%
\caption{\small Option prices and IVs under model (\ref{eq:dedyn}). Prices are computed via  formula (\ref{eq:KOU_extension_formulae}). Parameters: $r=2\%$, $S_0=100$, $\kappa=10$, $p=0.6$, $\lambda_1=60$, $\lambda_2=50$, $u=0.55$, $\eta_1=15$, $\eta_2=12$, $\sigma=20\%$.}
\end{table}%

\begin{figure}[H]
\centering{}\includegraphics[scale=0.43]{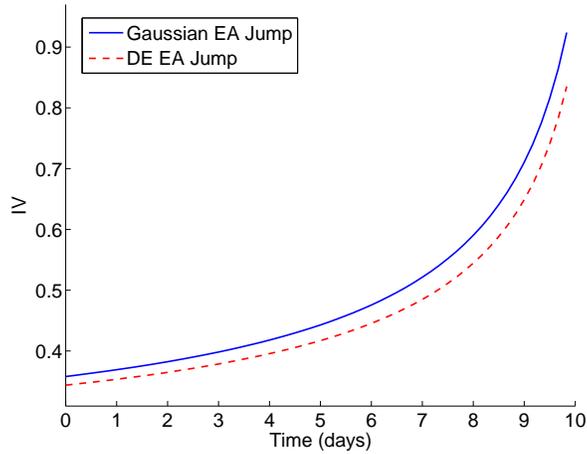}\\
\caption{\small{The time-series of the ATM option IV under the extended Kou model (\ref{eq:dedyn}) when  the EA and expiration dates are on the 10th and 11th days respectively. Parameters: $T_e=\hbox{14 days},T=\hbox{15 days},r=2\%$, $S_0=100$, $\kappa=10$, $p=0.5$, $\lambda_1=50$, $\lambda_2=50$, $u=0.5$, $\eta_1=25$, $\eta_2=25$, $\sigma=20\%$. }}\label{fig:Kou_rising_IV}
\end{figure}

\subsection{Pricing via the Characteristic Function\label{sub:fourier}}

In general, let us write  the terminal stock price as $S_T = S_0 e^{X_T + Z_e}$. If $X_{T}$ and $Z_{e}$ are  independent, and they both 
admit analytic characteristic functions, then we  obtain the characteristic formula of the log-price  
\[
\Psi\left(\omega\right):=\mathbb{E}\left\{e^{\mbox{i\ensuremath{\omega}log}\left(\frac{S_{T}}{S_0}\right)}\right\} = e^{\hat{\psi}\left(\omega\right)+\psi_{e}\left(\omega\right)},  
\]
where \[\hat{\psi}\left(\omega\right):={\log}\left(\mathbb{E}\left\{ e^{\mbox{i\ensuremath{\omega}log}\left(X_{T}\right)}\right\} \right), \quad \text{ and } \quad  \psi_{e}\left(\omega\right):={\log}\left(\mathbb{E}\left\{ e^{\mbox{i\ensuremath{\omega}log}\left(Z_{e}\right)}\right\} \right).\]
It is then possible to make use of available methods to price vanilla
as well as exotic options. For example, the methods of \citet{Carr1999}, \cite{Duffie2000}, \cite{Lee2004}, 
and \citet{Raible2000}, among others,  can be used to price European options. Alternatively,
the methods developed  in \citet{Jackson2008} or \citet{Lord08afast} can be adapted to price
both European and American options in models incorporating the jump
$Z_{e}$. 

In the following sections, we will also consider an extension of the
Heston model with dynamics
\begin{align}
\frac{dS_{t}}{S_{t}} & =r dt+\sigma_{t}dW_{t}+d\left(1_{\left\{ t\geq T_{e}\right\} }\left(e^{Z_{e}}-1\right)\right),\nonumber \\
d\sigma_{t}^{2} & = \nu \left(\vartheta-\sigma_{t}^{2}\right)dt+\zeta\sigma_{t}d\tilde{W}_{t},\label{eq:heston_dyn}
\end{align}
where $W$ and $\tilde{W}$ are standard Brownian motions, with $\mathbb{E}\{dW_{t}d\tilde{W_{t}}\}=\rho dt,$
and $Z_{e}$ is independent of both $W$ and $\tilde{W}$. Write the option price as
\begin{align*}
C\left(0,S\right) & =  \int_{-\infty}^{\infty}\left(e^{{\log}\left(S+x\right)}-K\right)^+f(x)dx \\
& =  \left(e^{\,\cdot}-K\right)^+\ast f(-\,\cdot)\,\left({\log}(S)\right),
\end{align*} 
where $f$ denotes the p.d.f. of log$(\frac{S_T}{S_0})$ and $\ast$ denotes the convolution operator. 
Denoting with
$\mathcal{F}\left(g\left(x\right)\right)\equiv\int_{\mathbb{R}}e^{-i\omega x}g\left(x\right)dx$
the Fourier operator acting on the function $g$, we can then write the
price of the European call as 
\begin{eqnarray}
C\left(0,S\right)=e^{\gamma {\log}(S)}\mathcal{F}^{-1}\left\{\mathcal{F}\left\{e^{-\gamma x}\left(Se^{x}-K\right)^{+}\right\}\Psi\left(\omega-i\gamma\right)\right\}.\label{eq:FFT}
\end{eqnarray}
where the introduction of the dampening factor $e^{-\gamma x}$ is necessary because the payoff is not integrable. The same observations have been  made in \citet{Lord08afast}, where    (\ref{eq:FFT}) is implemented as part of a new  pricing algorithm. \citet{Jackson2008} also derive the same formula by analyzing the associated pricing PIDE in Fourier space. One can then apply a Fast Fourier transform (FFT) algorithm to price according to (\ref{eq:FFT}). In addition,  these methods can be adapted to price American options, as we will do in Section \ref{sec:american}.

\section{Pre-Earnings Announcement Implied Volatility Properties\label{sec:IV}}\label{sect-IVbehave}
The market observations in Figures 1-3 motivate us to   analyze some characteristics of the pre-earnings announcement implied volatility.  In Section \ref{sect-IVbounds}, we provide  upper and lower
bounds for the IV under a class of models.  
In Section \ref{sect-IVasy},  we study some IV asymptotics  with focus on   small and large strikes. 

\subsection{Implied Volatility Bounds \label{sub:IV_bounds}}\label{sect-IVbounds}
For our analysis of the IV bounds of   European options, we consider a  general
framework where the terminal stock price is written in the form 
\begin{equation}
S_{T}=S_{t}e^{X_{t,T}+Z_{e}},\label{eq:static_dyn}
\end{equation}
with $S_t$ the stock price at time $t<T_e$.
The martingale condition on  $S$ implies that    $(X_{t,T})_{0\le t \le T}$  satisfies   $\mathbb{E}\left\{ e^{X_{t,T}}\right\} =e^{r(T-t)}$. The r.v.   $Z_{e}$ is a continuous mixture of Gaussian r.v.'s with  p.d.f. $f_{Z_{e}}\left(y\right)=\int_{\mathbb{R}^{+}}\phi\left(y;-\frac{\hat{\sigma}^{2}}{2},\hat{\sigma}\right)\G\text{\ensuremath{\left(d\hat{\sigma}\right)}}$,
where $\phi\left(\cdot;a,b\right)$
represents the p.d.f. of a Gaussian r.v. with mean $a$ and variance
$b^2$, and $\G\left(\cdot\right)$ is a measure
over the space $\mathbb{R}^{+}$ with $\G[0,\infty)=1$.

 Note that we have not specified the distribution
of $X_{t,T}$, therefore the base model can be very general. We have
the following lower bound for the IV volatility function.

\begin{prop}
\label{prop:lower_bound}Suppose the terminal stock price $S_T$ follows (\ref{eq:static_dyn}). Then, the implied volatility $I(t;K,T)$ admits the lower bound
\begin{align}
I\left(t;K,T\right)\geq\frac{\hat{\sigma}_{min}}{\sqrt{T-t}}, \quad  t<T_{e}\label{eq:lower_bound}, 
\end{align} where $\hat{\sigma}_{min} := \inf\{ \hat{\sigma} \in \mathbb{R}_+ \,:\, \G[0, \hat{\sigma}]>0\}.$
In addition, if $X_{t,T}$ is also distributed as a Gaussian mixture,
$f_{X_{t,T}}\left(y\right)=\int_{\mathbb{R}^{+}}\phi\left(y;r(T-t)-\frac{\tilde{\sigma}^{2}}{2},\tilde{\sigma}\right)\H\left(d\tilde{\sigma}\right)$,
then the lower bound improves to
\[
I\left(t;K,T\right)\geq\sqrt{\frac{\tilde{\sigma}^{2}_{min}+\hat{\sigma}_{min}^{2}}{T-t}}\ ,\ t<T_{e},
\]
where $ \tilde{\sigma}_{min} := \inf\{ \tilde{\sigma} \in \mathbb{R}_+ \,:\, \H[0, \tilde{\sigma}] >0\}.$
\end{prop}
We note that   $\hat{\sigma}_{min} \ge 0$,  so the bound \eqref{eq:lower_bound} is nontrivial only if   $\hat{\sigma}_{min}>0$. This means that  the measure $G$ has zero weight on Gaussian r.v.'s with  variance smaller than  $\hat{\sigma}_{min}>0$. Such a  condition is satisfied, for example, by any finite  mixture of Gaussian r.v. 

%
%The lower bound not only shows that the increasing behavior of the
%IV is shared by different models, but that also that the rate of increase
%is bounded from below from $\frac{\hat{\sigma}_{min}}{\sqrt{T-t}}$
%which is of the same functional form as under the framework (\ref{eq:GBMe}).
%The same functional behavior is also shared by the upper bound of
%the ATM-forward IV. 

In addition, we obtain an upper bound for the implied volatility.
\begin{prop}
\label{prop:upper_bound}Suppose the terminal stock price $S_T$ follows (\ref{eq:static_dyn}) and
assume that both $X_{T}$ and $Z_{e}$ are distributed as continuous
Gaussian mixtures,
\begin{equation}
f_{Z_{e}}\left(y\right)=\int_{\mathbb{R}^{+}}\phi\left(y;-\frac{\hat{\sigma}^{2}}{2},\hat{\sigma}\right)\G\text{\ensuremath{\left(d\hat{\sigma}\right)}},\quad f_{X_{t,T}}\left(y\right)=\int_{\mathbb{R}^{+}}\phi\left(y;r(T-t)-\frac{\tilde{\sigma}^{2}}{2},\tilde{\sigma}\right)\H\left(d\tilde{\sigma}\right).\label{eq:mix_cond}
\end{equation}
Then, the following upper bound for the implied volatility of a European
ATM-forward call option, $K=e^{r\left(T-t\right)}S_{t}$, holds:
\begin{align}
 & I\left(t;e^{r\left(T-t\right)}S_{t},T\right)\leq\sqrt{\int_{\mathbb{R}^{+}}\frac{\tilde{\sigma}^{2}}{T-t}\H\left(d\tilde{\sigma}\right)+\int_{\mathbb{R}^{+}}\frac{\hat{\sigma}^{2}}{T-t}\G\left(d\hat{\sigma}\right)}\ ,\ t<T_{e}.\label{eq:upper_bound}
\end{align}
 
\end{prop}

Notable examples of models that satisfy conditions (\ref{eq:mix_cond}) include the extended Merton and Heston (when $\rho=0$) models. However, we remark that  Gaussian mixtures can also be used to approximate other distributions. Moreover, our bounds can serve as analytical benchmarks for the IV under different models. As an example, we derive the explicit  expressions for
the bounds (\ref{eq:upper_bound}) under the Heston model.

According to Propositions \ref{prop:lower_bound} and \ref{prop:upper_bound},   the IV bounds  under different models exhibit similar   behaviors as time approaches the earnings announcement. Comparing the bounds to the IV function $I$ in \eqref{eq:BS_IV_ts}, it is not too surprising that  the simple extended  BS model was able to fit the observed ATM IV over time  (see Figure \ref{fig:BS_cal_example}).

\begin{example}\label{example1}
Assume that $S$ follows the Heston dynamics (\ref{eq:heston_dyn}).  
%stochastic volatility dynamics
%\begin{equation}
%{\log}\left(\frac{S_{T}}{S_t}\right)=\left(r-\frac{\tilde{\sigma}^2}{2}\right)(T-t)+\int_{t}^{T}\sigma_{u}dW_{u}+\left(Z_{e}-{\log}\left(\mathbb{E} e^{Z_{e}}\right) \right),\nonumber
%\end{equation}
%where $\sigma_{u}$ is a stochastic process,  \tilde{\sigma}^{2}\equiv\frac{1}{T-t}\int_{t}^{T}\sigma_{u}^{2}du$ and $Z_{e}$ is independent
%of both $W_{u}$ and $\sigma_{u}$. 
In the case $\rho=0$ it is known that, conditioned on the   path of $(\sigma_u)_{t\le u \le T}$, $\int_{t}^{T}\sigma_{u}dW_{u} \sim N\left(0,\tilde{\sigma}^{2}\right)
$, where $\tilde{\sigma}^{2}\equiv\int_{t}^{T}\sigma_{u}^{2}du$. From this we observe that $X_{t,T}\equiv \left(r-\frac{\tilde{•\sigma}^2}{2(T-t)}\right)(T-t)+\int_{t}^{T}\sigma_{u}dW_{u}$ satisfies the second part of \eqref{eq:mix_cond}.  In turn, direct computation yields  that $\int_{\mathbb{R}^{+}}\frac{\tilde{\sigma}^{2}}{T-t}\H\left(d\tilde{\sigma}\right)=\vartheta+\frac{\sigma_{t}^2-\vartheta}{\nu\left(T-t\right)}\left(1-e^{-\nu\left(T-t\right)}\right).$
Therefore, for example, in the case of a Gaussian EA jump, the term $\int_{\mathbb{R}^{+}}\frac{\hat{\sigma}^{2}}{T-t}\G\left(d\hat{\sigma}\right)$
is  equal to $\frac{\sigma_{e}^{2}}{T-t}$ and the bound (\ref{eq:upper_bound}) reads 
\begin{equation}
I\left(t;e^{r\left(T-t\right)}S_{t},T\right)\leq\sqrt{\vartheta+\frac{\sigma_{t}^2-\vartheta}{\nu\left(T-t\right)}\left(1-e^{-\nu\left(T-t\right)}\right)+\frac{\sigma_{e}^{2}}{T-t}}.
\label{eq:up_bound_exG}
\end{equation}
If we instead assume that the jump is distributed as a symmetric double-exponential,
$Z_{e}\sim DE\left(\frac{1}{2},\eta,\eta\right)$, then we use the fact that $Z_{e}\overset{d}{=}\sqrt{\frac{2}{\eta^{2}}\epsilon}Z,$
where $\epsilon\sim\mbox{Exp}\left(1\right)$ and $Z\sim N\left(0,1\right)$ are independent.
In fact, in that case, $\G\left(x\right)=\eta^{2}xe^{-\frac{\eta^{2}x^{2}}{2}}$ , $\int_{\mathbb{R}^{+}}\frac{\hat{\sigma}^{2}}{T-t}\G\left(d\hat{\sigma}\right)=\frac{2}{\eta^{2}\left(T-t\right)},$
and the bound (\ref{eq:upper_bound}) reads 
\begin{equation}
I\left(t;e^{r\left(T-t\right)}S_{t},T\right)\leq\sqrt{\vartheta+\frac{\sigma_{t}^2-\vartheta}{\nu\left(T-t\right)}\left(1-e^{-\nu\left(T-t\right)}\right)+\frac{2}{\eta^{2}\left(T-t\right)}}.
\label{eq:up_bound_exDE}
\end{equation} 
\end{example}

In Figure \ref{fig:time_bounds} we plot the explicit bounds (\ref{eq:up_bound_exG}) and (\ref{eq:up_bound_exDE}) from  Example  \ref{example1}. As we can see, the upper bound is relatively
close to the model IV curve and shares a very similar time dependence.   In particular,  when the EA jump is Gaussian   (Figure \ref{fig:time_bounds}, left), the bound is almost indistinguishable from the model IV. To see this, as $t \rightarrow T$, the lower and upper bounds ((\ref{eq:lower_bound}) and (\ref{eq:up_bound_exG}))  share a common leading term  $\frac{\sigma_e^2}{T-t}$. In practical cases as in this example, the term $\int_{\mathbb{R}^{+}}\frac{\tilde{\sigma}^{2}}{T-t}\H\left(d\tilde{\sigma}\right)\approx \sigma^2_t$   is typically at least one order of magnitude smaller than $\sigma_e^2/(T-t)$ since $T-t$ is very small and $\sigma_e$ and $\sigma_t$ are of the same order.  Also, in Figure \ref{fig:time_bounds}, we observe that the model IV with  a non-zero  $\rho$ still admits a similar time behavior and are very close to the bounds (\ref{eq:up_bound_exG}) and (\ref{eq:up_bound_exDE})  for the case $\rho=0$.  As a curious note,  the coefficient $\frac{2}{\eta^{2}}$ of  the leading term in (\ref{eq:up_bound_exDE}) as $t \rightarrow T$, is exactly the variance of $Z_e$ when $Z_e\sim DE(\frac{1}{2},\eta,\eta)$, the same way that  the coefficient $\sigma_e^2$ of the  leading term in (\ref{eq:up_bound_exG})  is the variance of the Gaussian $Z_e$.  Therefore, an interesting question is whether    the rate of change of the IV is approximately proportional to the standard deviation of the EA jump, at least as time approaches the EA date (see also Figure \ref{fig:Kou_rising_IV}).

 To conclude, we   recall that  Propositions \ref{prop:lower_bound} and \ref{prop:upper_bound} also give us information about the term structure of the IV. Indeed, for models whose dynamics are time-homogeneous, we observe the relationship: $\frac{\partial I\left(t;K,T\right)}{\partial t}=-\frac{\partial I\left(t;K,T\right)}{\partial T}$. This implies  that one should expect a decreasing term structure for  options prior to the EA date.

\begin{figure}[H] \centering{}
\includegraphics[scale=0.4]{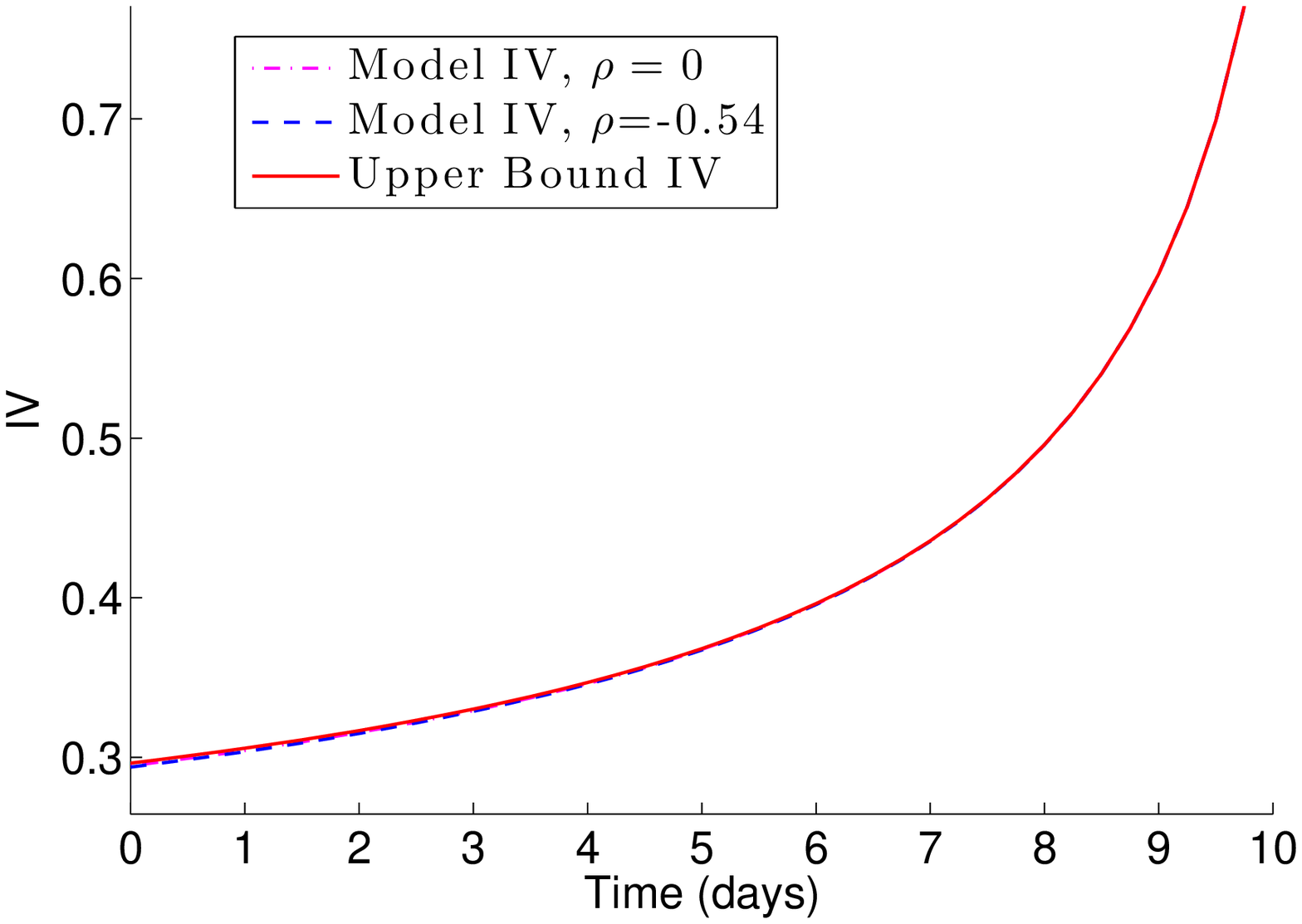}\includegraphics[scale=0.4]{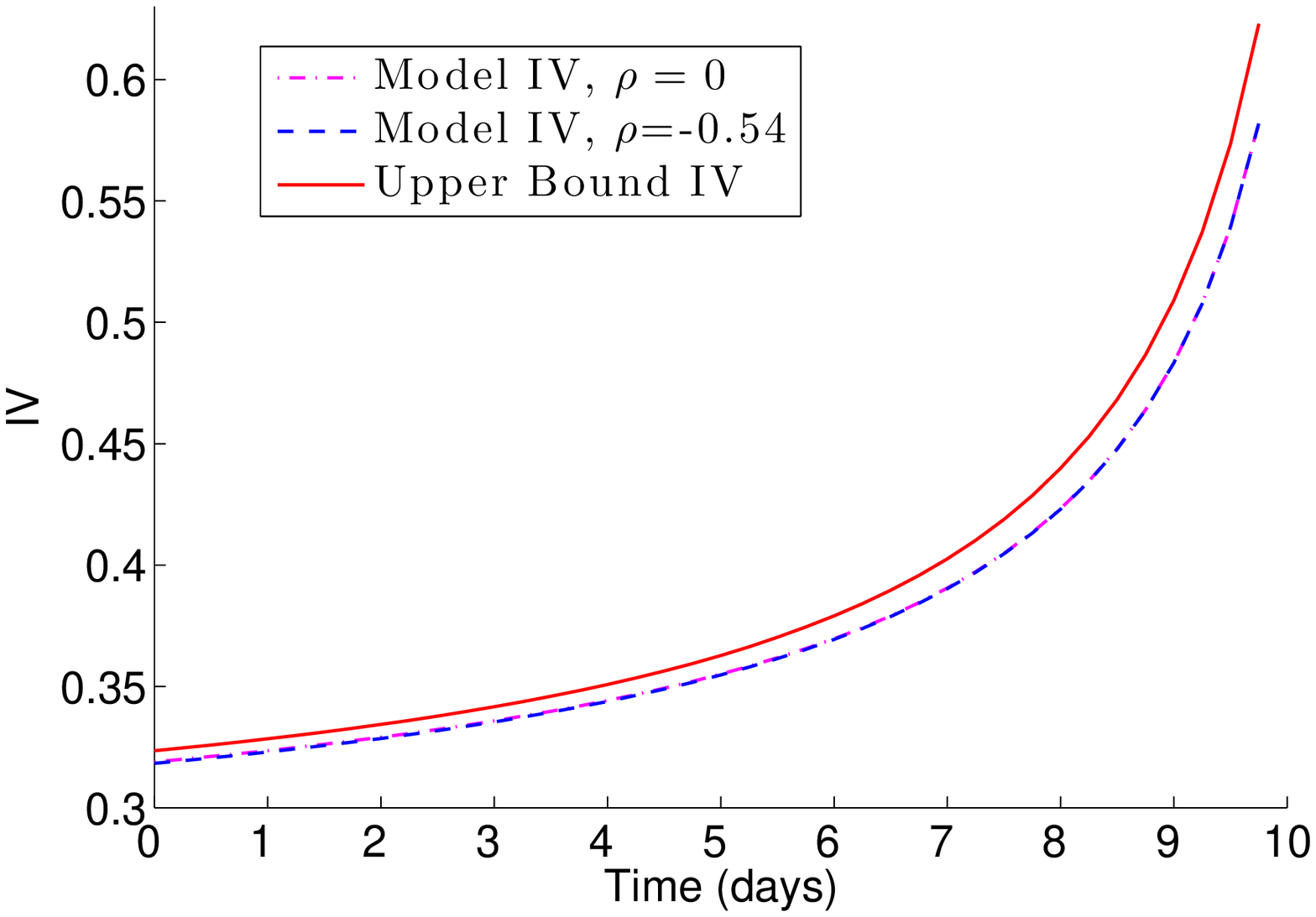}\caption{\small The model IV plotted against the upper bounds in (\ref{eq:up_bound_exG})  and (\ref{eq:up_bound_exDE})  
under the extended Heston model with a Gaussian EA jump (left) and  DE EA jump (right) respectively. Parameters: $\nu=4.04,\ \vartheta=0.05,\ \sigma_{0}=1.01,\ \rho\in \{0,-0.54\},\ \zeta=0.03,\ \sigma_{e}=.0473,\ u=0.5,\ \eta_{1}=\eta_{2}=40$.
\label{fig:time_bounds}}
\end{figure}

\subsection{Small and Large Strikes   Asymptotics\label{sub:IV_asy}}\label{sect-IVasy}
We now  analyze the asymptotics of the IV surface for small
and large strikes for the extended Heston  and Kou model.  Our asymptotics  follow from  an adaptation  of the results of \citet{Benaim2008} (see also \citet{Benaim2008a}). We state the results here  and provide the proofs in Appendix \ref{app-strike}.

%First, notice that the models we have extended
%admit a moment generating function. Denote with $M\left(\omega\right)\equiv\mathbb{E}\left\{ e^{\omega X}\right\} $
%the m.g.f. of the r.v. $X$, and with $F$ its c.d.f. If $r^{*}\equiv\mbox{inf}\left\{ \omega\ s.t.\ M\left(\omega\right)<\infty\right\} $
%is finite, than it easy to show that $\mbox{lim sup}_{x\rightarrow\infty}\frac{-{\log}\left(1-F\left(x\right)\right)}{x}=r^{*}$.
%If $F$ is well-behaved, the lim sup can be replaced by a lim and
%the tail asymptotics $\frac{-{\log}\left(1-F\left(x\right)\right)}{x}\sim r^{*}x$,
%along with the condition $r^{*}>1$, are sufficient to provide asymptotics
%for the implied volatility (see \citet{Benaim2008}). 

\begin{prop}\label{prop:asy1}
Let $S$ satisfy the extended Heston dynamics (\ref{eq:heston_dyn})
where $Z_{e}$ is either: 

Case 1: normally distributed, $Z_{e}\sim N\left(0,\frac{\sigma_{e}^{2}}{2}\right)$;
or

Case 2: double-exponentially distributed, $Z_{e}\sim DE\left(u,\eta_{1},\eta_{2}\right)$.
\\
Then for any fixed $t<T_{e}$, the  implied volatility  $I\left(t;K,T\right)$ satisfies 
\begin{align}
\frac{I^{2}\left(t;K,T\right)\left(T-t\right)}{{\log}\left(\frac{K}{S_{t}}\right)} & \sim\mbox{\ensuremath{\xi}\ensuremath{\ensuremath{\left(q^{*}\right)}}},\ as\ \ensuremath{K}\rightarrow\ensuremath{0},\label{eq:asymptotics-11} \\
\frac{\mbox{\ensuremath{I^{2}\left(t;K,T\right)\left(T-t\right)}}}{{\log}\mbox{\ensuremath{\left(\frac{K}{S_{t}}\right)}}} & \sim\mbox{\ensuremath{\xi}\ensuremath{\ensuremath{\left(r^{*}-1\right)}}},\ as\ \ensuremath{K}\rightarrow\ensuremath{\infty},\label{eq:asymptotics-12}
\end{align}
where $\xi\left(x\right)$ is defined by $\xi\left(x\right)\equiv2-4\left(\sqrt{x^{2}+x}-x\right)$
and 
\[
q^{*}=\begin{cases}
p_- & case\ 1,\\
{\min}\left\{ p_-,\eta_{2}\right\}  & case\ 2,
\end{cases}\quad,\quad r^{*}=\begin{cases}
p_+ & case\ 1,\\
\min\left\{ p_+,\eta_{1}\right\}  & case\ 2,
\end{cases}
\]
and $p_{\pm}$ is the smallest positive solution to, respectively,
\[
\nu\mp\rho\zeta p_{\pm}+\mbox{\ensuremath{\sqrt{\left(\nu\mp\rho\zeta p_{\pm}\right)^{2}+\zeta^{2}\left(\pm p_{\pm}-(p_{\pm})^{2}\right)}}}\,{\coth}\left(\frac{\left(T-t\right)}{2}\mbox{\ensuremath{\sqrt{\left(\nu\mp\rho\zeta p_{\pm}\right)^{2}+\zeta^{2}\left(\pm p-(p_{\pm})^{2}\right)}}}\right)=0.
\]
\end{prop}
 
 \begin{prop}\label{prop:asy2}
Let $S$ follow  the extended Kou dynamics (\ref{eq:dedyn}) where
$Z_{e}$ is either: 

Case 1: normally distributed, $Z_{e}\sim N\left(0,\frac{\sigma_{e}^{2}}{2}\right)$;
or

Case 2: double-exponentially distributed, $Z_{e}\sim DE\left(u,\eta_{1},\eta_{2}\right)$. \\
Then for any fixed $t<T_{e}$, the implied volatility  $I\left(t;K,T\right)$ satisfies 
\begin{align}
\frac{I^{2}\left(t;K,T\right)\left(T-t\right)}{{\log}\left(\frac{K}{S_{t}}\right)} & \sim\mbox{\ensuremath{\xi}\ensuremath{\ensuremath{\left(q^{*}\right)}}},\ as\ \ensuremath{K}\rightarrow\ensuremath{0},\label{eq:asymptotics-21} \\
\frac{\mbox{\ensuremath{I^{2}\left(t;K,T\right)\left(T-t\right)}}}{{\log}\mbox{\ensuremath{\left(\frac{K}{S_{t}}\right)}}} & \sim\mbox{\ensuremath{\xi}\ensuremath{\ensuremath{\left(r^{*}-1\right)}}},\ as\ \ensuremath{K}\rightarrow\ensuremath{\infty},\label{eq:asymptotics-22}
\end{align}
where $\xi\left(x\right)\equiv2-4\left(\sqrt{x^{2}+x}-x\right)$
and
\[
q^{*}=\begin{cases}
\lambda_{2} & case\ 1,\\
\min\left\{ \lambda_{2},\eta_{2}\right\}  & case\ 2,
\end{cases}\quad,\quad r^{*}=\begin{cases}
\lambda_{1} & case\ 1,\\
\min\left\{ \lambda_{1},\eta_{1}\right\}  & case\ 2.
\end{cases}
\]

\end{prop}

First, we observe that if the EA jump is Gaussian, then it has no role in  the IV asymptotics in strikes (see \eqref{eq:asymptotics-11} and \eqref{eq:asymptotics-21}). Hence,  in either the Heston or Kou model, the large/small strikes asymptotics  with and without the EA jump are in fact identical.   On the other hand, if the tails of the EA jump are fatter than those
of the base model, then the IV asymptotics are determined by the EA
jump parameters. In such cases,
  the asymptotics are observed  for less extreme strikes
with  short maturities, which is   when the EA jump variance dominates. For longer
maturities, however,  the asymptotics  hold for  more extreme strikes.
An intuitive explanation is that, as time-to-maturity increases, the EA jump variance is relatively
low and the tails behavior is manifest only for extreme values. 

In Figure \ref{fig:asymptotics}, we show the IV asymptotics under the extended Kou (left) and Heston (right) models compared to the IV obtained by inverting the BS formula on option prices calculated via Fourier transform (see Section \ref{sub:fourier}). We plot the asymptotic volatility function $I(t;K,T)=c + \sqrt{{\log}\left( {K}/{S_t}\right)\xi(\omega)}$ for fixed $(t,T)$, where $c$ is a constant chosen so that   the asymptotics and the model IVs   coincide at the most extreme strikes considered. In each particular case,  $\omega$ is a constant  set according to (\ref{eq:asymptotics-21})-(\ref{eq:asymptotics-22}) (for the extended Kou model) and (\ref{eq:asymptotics-11})-(\ref{eq:asymptotics-12}) (for the extended Heston model). In Figure \ref{fig:asymptotics} (left) the double-exponential EA jump tails dominates those of
the daily jumps, as this   generally holds in practice. In Figure \ref{fig:asymptotics}  (right), the EA jump is Gaussian and thus does not affect the IV asymptotics from  the base model.

\begin{figure}[H] \centering{}
\includegraphics[scale=0.4]{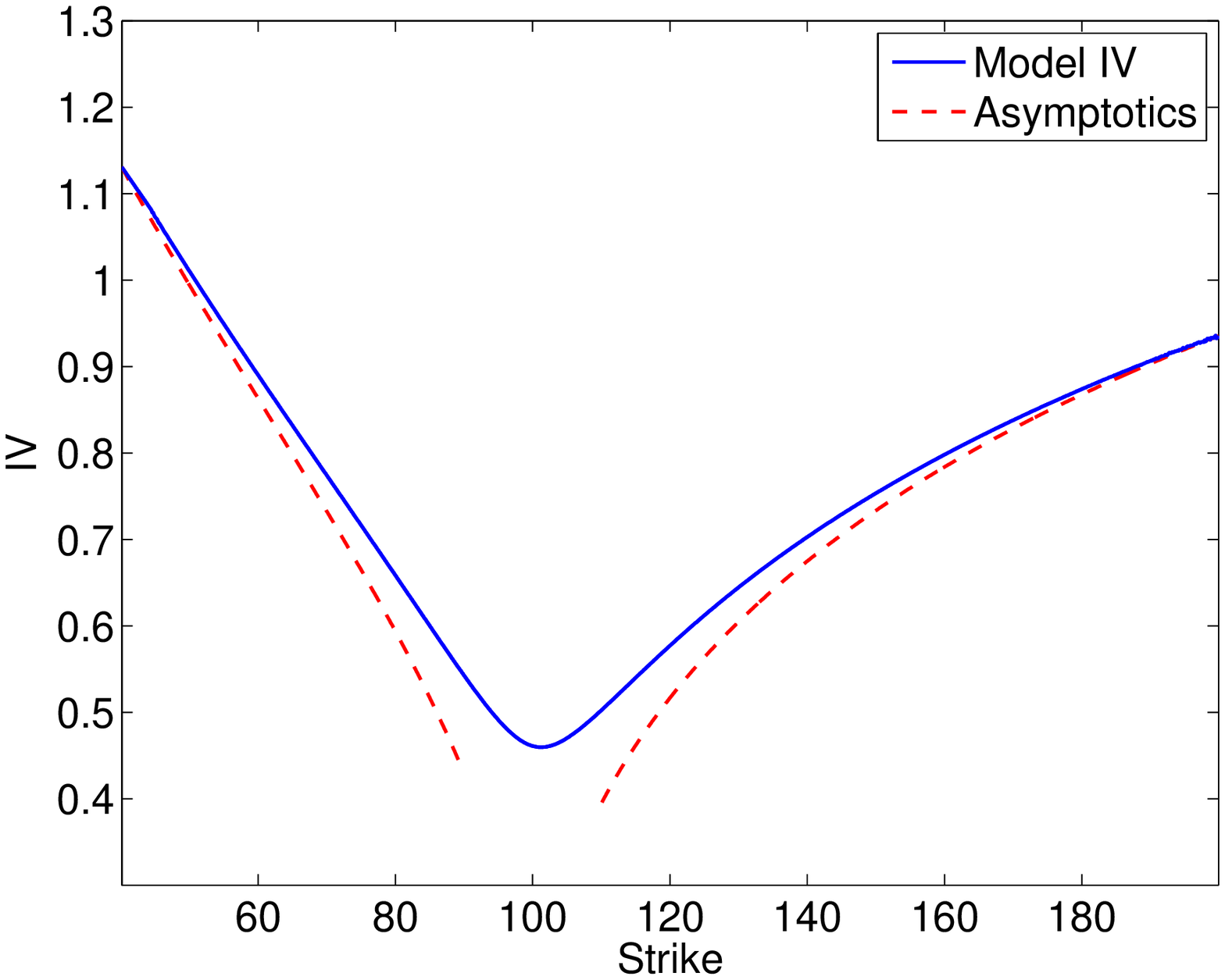}\includegraphics[scale=0.4]{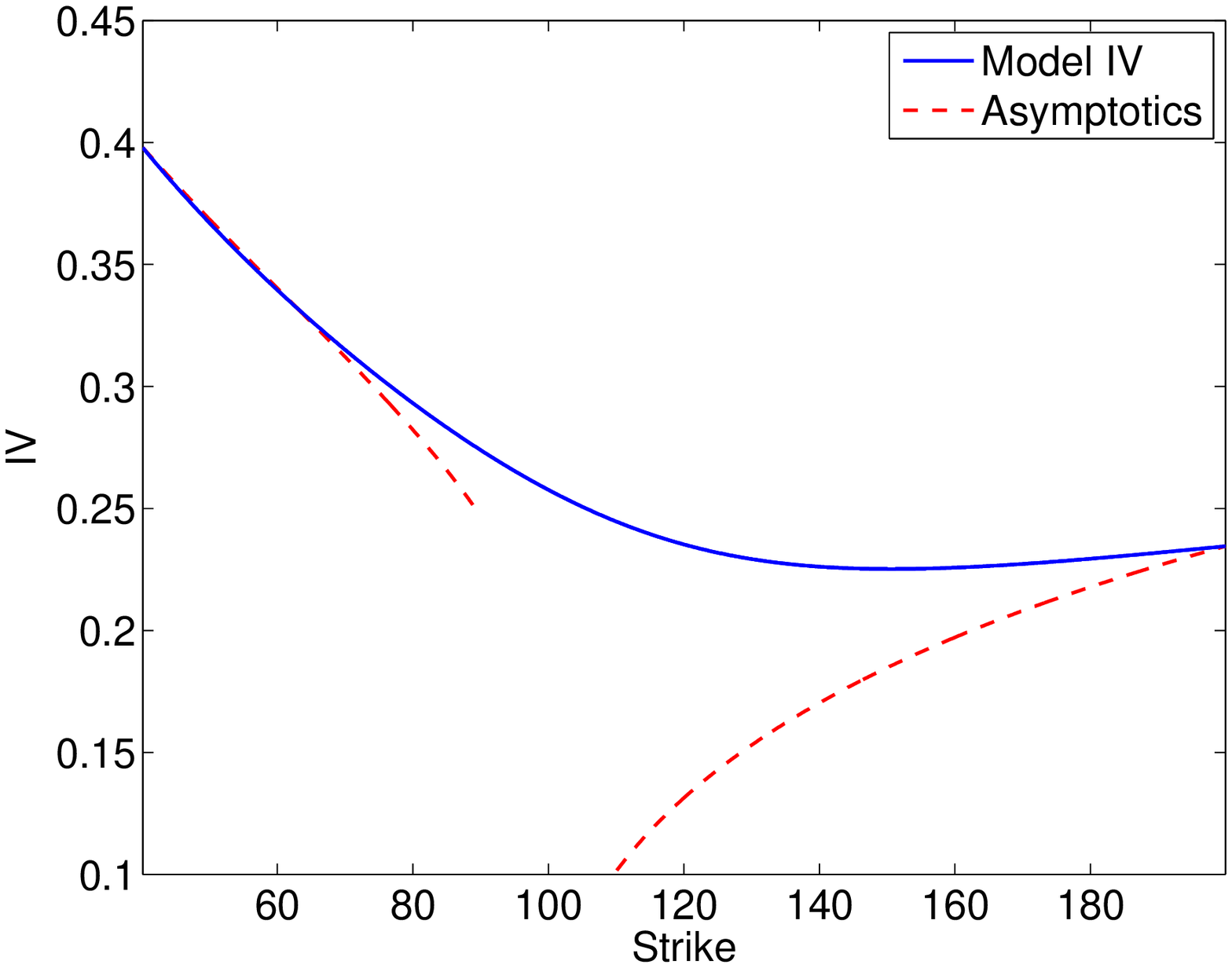}\caption{\small (Left) The  IV obtained from numerical option prices against the asymptotics (\ref{eq:asymptotics-21})-(\ref{eq:asymptotics-22})
under the extended Kou model with a log-DE EA jump. Parameters: $S=100,\ r=0.02,\ T=\frac{4}{252},\ \kappa=300,\ p=0.5,\ \lambda_{1}=\lambda_{2}=100,\ u=0.5,\ \eta_{1}=30,\ \eta_{2}=25$. 
(Right) The IV obtained from numerical option prices against the asymptotics (\ref{eq:asymptotics-11})-(\ref{eq:asymptotics-12})
under the extended Heston model with a log-normal EA jump. Parameters: $S=100,\ r=0.02,\ T=1,\ \nu=2.7,\ \vartheta=0.077,\ \sigma_{0}=0.7075,\ \rho=-0.54,\ \zeta=0.073,\ \sigma_{e}=0.04.$}
\label{fig:asymptotics}
\end{figure}

%Concluding, we note that we do not analyze the limiting behaviors for short and long time-to-maturities since known results for the base models can be directly applied in those cases. For example, since  the earnings announcement is made at a fixed  time before expiration, after an EA and approaching maturity, the model reduces to the base model without the EA jump. Alternatively, for very large time to maturity, the effect of the EA jump will be insignificant w.r.t. to the rest of the price dynamics.

\section{Calibration and Parameter Estimators\label{sec:calibration}}

In this section, we perform  calibrations of the extended
BS, Kou and Heston models to  the observed market prices of options near an earnings announcement. This allows us to   evaluate whether
the model extensions improve the accuracy of the calibration compared
to that of the base models. Calibration results can also be used to
infer information about the distribution of the EA jump and, in one
simple example, we compare the estimators obtained through calibration
on IBM options data to estimates given by its historical distribution.

\subsection{Model Calibration\label{sub:calibration}}
In our calibration procedure, we consider  a set of $N$ vanilla calls and puts with observed market prices,  $\hat{C}_{i}$, $i=1,...,N$. These  options have different  contractual features, such as strike, maturity, and option type. For a given  model,  the set of model parameters is denoted by $\Theta$. In turn, the model price of   option  $i$ is denoted by $C_i\left(\Theta\right)$.  To calibrate a given  model, we minimize the sum of squared errors (see e.g. \citet{Dennis1977,Andersen2000,Bates1996,Cont2002}):
\begin{equation}
\min_{\Theta}\limits  \, \sum_{i=1}^{N}\left(C_i\left(\Theta\right)-\hat{C}_{i}\right)^{2},\label{eq:cal_problem}
\end{equation}
We use the best-bid and best-ask mid-point prices. Furthermore, we  adapt  a trust-region-reflective gradient-descent algorithm (\citet{Coleman1994}; \citet{Coleman1996}), starting from different initial points to guarantee a better exploration of the parameters space. While  our  numerical tests show the adopted method  results in an effective calibration, we  remark that there are many alternative, possibly more advanced, calibration procedures available (see e.g. \cite{Cont2002} and references therein).

The majority of equity options in the US are of
the American type.  While in Section \ref{sec:american} we will discuss  the pricing of   American options, the methods are generally  too computationally-intensive to be practically used in conjunction
with a gradient-descent method for calibration, especially as the number of options and parameters increase. Related studies typically circumvent this issue by simply assuming the American options are European-style (see e.g. \cite{Dubinsky2006, Broadie2009}). In contrast, our procedure  begins by inverting the market prices of American puts and calls via   a relatively fast American option pricer under the Black-Scholes model.  This gives us the observed IVs. In turn,  we apply the   Black-Scholes European put/call pricing formula,  with the volatility parameter being the  observed IV, and derive the associated  \textit{European} put or call  price. We then use the resulting prices as  inputs to calibrate against the option prices generated from a model.  In all our experiments, we obtain option price data, available up to August 2013,  from the OptionMetrics Ivy database.

We now present an   example using  the extended BS, Kou and Heston models with  different  distributions for the EA jump. Our objective is to illustrate the  calibrated  IV surfaces under these models and compare them with the empirical IV surface.    Recall from Figure \ref{fig:Cal_example_surface} the empirical  implied volatility surface of IBM on July 15, 2013. That is observed   2 days prior to the earnings announcement by IBM on  July  17, 2013 after market. The closest options expiration date  was
Friday  July 19, 2013.  As noted earlier,   the front month IVs are significantly higher than those for options with longer maturities. In Figure \ref{fig:Cal_example_surfaces}
we show the associated calibrated IV surfaces for the 3  base models and their extensions, resulting in  a total of 9  calibrated models.   Table
\ref{tab:cal_ex_params}  summarizes  the calibrated parameters.  In the original Black-Scholes model, the implied volatility surface is flat and takes a high value of 28.11\%. As we incorporate the EA jump, under both Gaussian and DE distributions, the calibrated  values  of the stock price volatility $\sigma$ are lower.  More importantly, in every case, the base model is unable to generate the characteristic shape of the IV surface before the EA.  Between the extended Heston and Kou models, the  IV surface generated from the Kou model tends to flatten more rapidly as maturity lengthens. Overall,  the Heston model seems to be  able to reproduce the IV surface more accurately, and the incorporation of  a Gaussian EA jump seems to reproduce the IV surface better than with a DE EA jump in this example. We will further compare the two EA jump distributions in Section \ref{sub:premia}.

%For comparison, notice that the analytic estimators given by (\ref{eq:BS_TS_estimator}) and (\ref{eq:BS_ts_estimator}) are: $ \sigma^{TS}=55.63\%,\: \sigma^{TS}_e=6.51\%,\: \sigma^{ts}=66\%,\:$and $\sigma^{ts}_e=6.38\% $. 

% Table generated by Excel2LaTeX from sheet 'Sheet1'
\begin{table}[H]
  \centering
  \small
        \begin{tabular}{rrrrrrrrrr}
Black-Scholes & $\sigma$ & $\sigma_e$ & u     & $\eta_1$  & $\eta_2$  &       &       &       &  \\
\cline{1-6}    Base & 28.11\% & \multicolumn{1}{c}{---} & \multicolumn{1}{c}{---} & \multicolumn{1}{c}{---} & \multicolumn{1}{c}{---} & \multicolumn{1}{c}{} &       &       &  \\
    Gaussian jump & 27.69\% & 7.11\% & \multicolumn{1}{c}{---} & \multicolumn{1}{c}{---} & \multicolumn{1}{c}{---} &       &       &       &  \\
    DE jump & 21.42\% & \multicolumn{1}{c}{---} & \multicolumn{1}{c}{38.28\%} &         23.03  & 6.3   &       &       &       &  \\
\cline{1-6}  
          &       &       &       &       &       &       &       &       &  \\
    Heston & $\nu$ & $\vartheta$ & $\zeta$  & $\rho$   & $\sigma_0^2$    & $\sigma_e$ & u     & $\eta_1$  & $   \eta_2$ \\
     \hline
    Base &           3.70  &           0.05  &           0.90  & -0.51 &           0.05  & \multicolumn{1}{c}{---} & \multicolumn{1}{c}{---} & \multicolumn{1}{c}{---} & \multicolumn{1}{c}{---} \\
    Gaussian jump &           4.04  &           0.05  &           1.01  & -0.55 &           0.03  & 4.73\% & \multicolumn{1}{c}{---} & \multicolumn{1}{c}{---} & \multicolumn{1}{c}{---} \\
    DE jump &           3.10  &           0.05  &           0.84  & -0.54 &           0.03  & \multicolumn{1}{c}{---} & 42.06\% &         34.77  &         24.95  \\
    \hline
          &       &       &       &       &       &       &       &       &  \\
    Kou   & $\sigma$ & $\kappa$  & p     & $\lambda_1$ & $\lambda_2$ & $\sigma_e$  & u     & $\eta_1$  & $\eta_2$ \\
    \hline
    Base & 2.80\% & \multicolumn{1}{c}{232.3} & \multicolumn{1}{c}{46.85\%} & \multicolumn{1}{c}{275.9} & \multicolumn{1}{c}{82.1} & \multicolumn{1}{c}{---} & \multicolumn{1}{c}{---} & \multicolumn{1}{c}{---} & \multicolumn{1}{c}{---} \\
    Gaussian jump & 0.03\% & 193.6 & 51.17\% & 998.9 & 70.0  & 3.61\% & \multicolumn{1}{c}{---} & \multicolumn{1}{c}{---} & \multicolumn{1}{c}{---} \\
    DE jump & 0.39\% & 85.1  & 79.46\% & 990.3 & 30.0  & \multicolumn{1}{c}{---} & 98.73\% & 32.9  & 2.0 \\
    \hline
    \end{tabular}%

  \caption{\small Summary of calibrated parameters  based on the observed IV  surface in  Figure \ref{fig:Cal_example_surface}. The corresponding calibrated IV surfaces are displayed in Figure \ref{fig:Cal_example_surfaces}}
  \label{tab:cal_ex_params}
\end{table}

%\begin{figure}[H]
%\centering{}\includegraphics[scale=0.5]{pictures/IBM_130715_cal}\caption{\small The observed  implied volatility surface for  IBM on July 15, 2013, two days before an earnings announcement. \label{fig:Cal_example_surface}}
%\end{figure}

\begin{figure}[H]
\centering{}\includegraphics[scale=0.25]{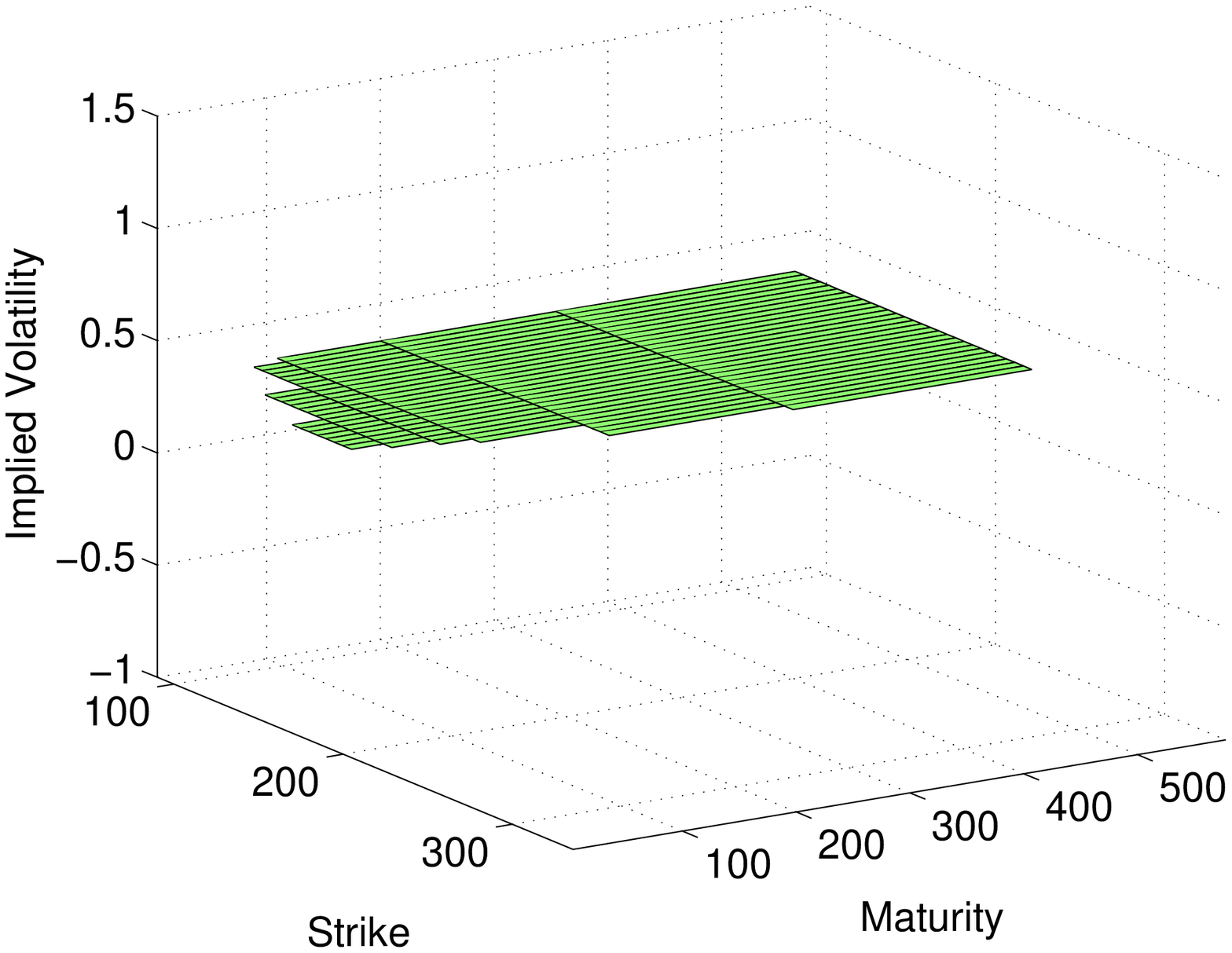}\includegraphics[scale=0.25]{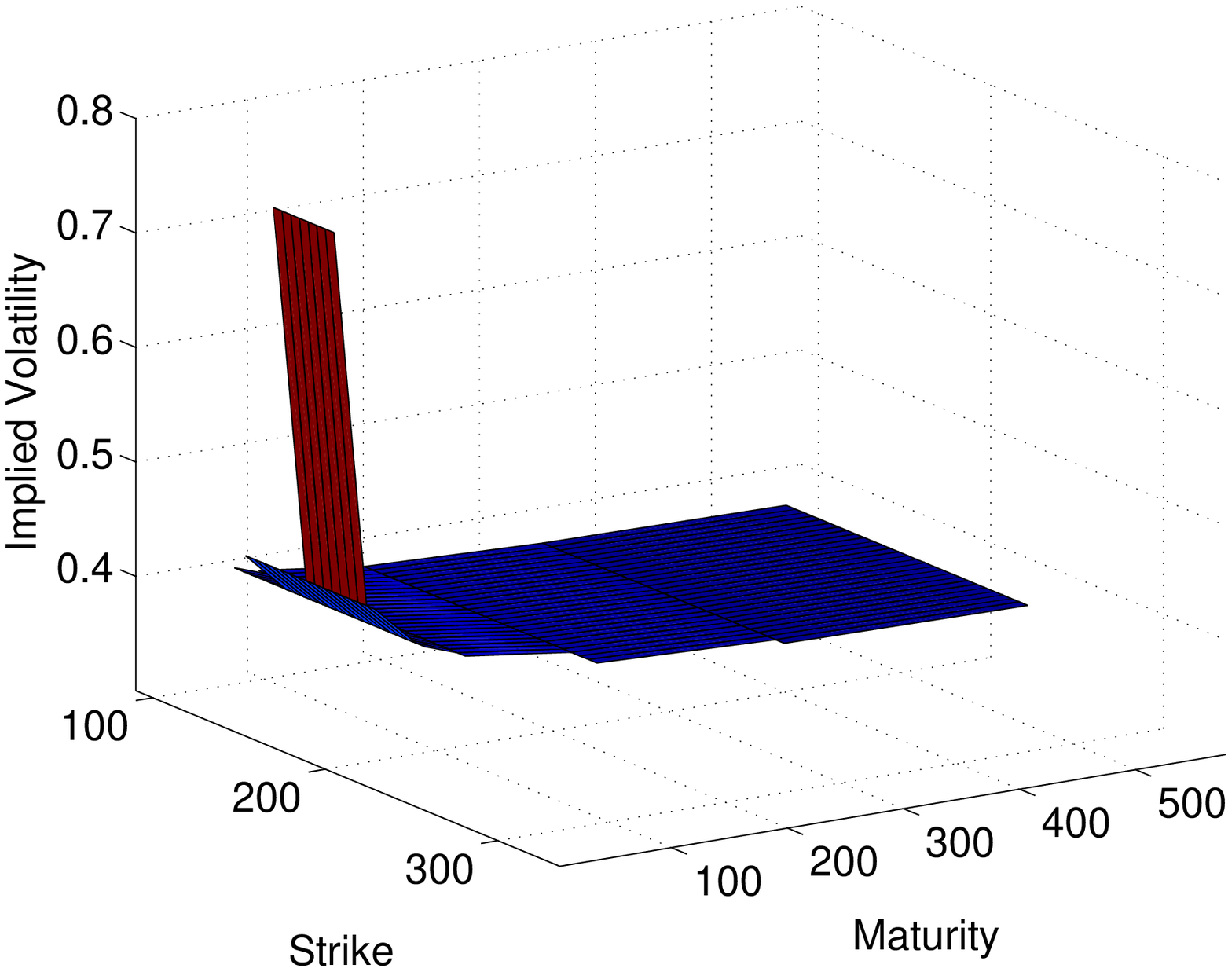}\includegraphics[scale=0.25]{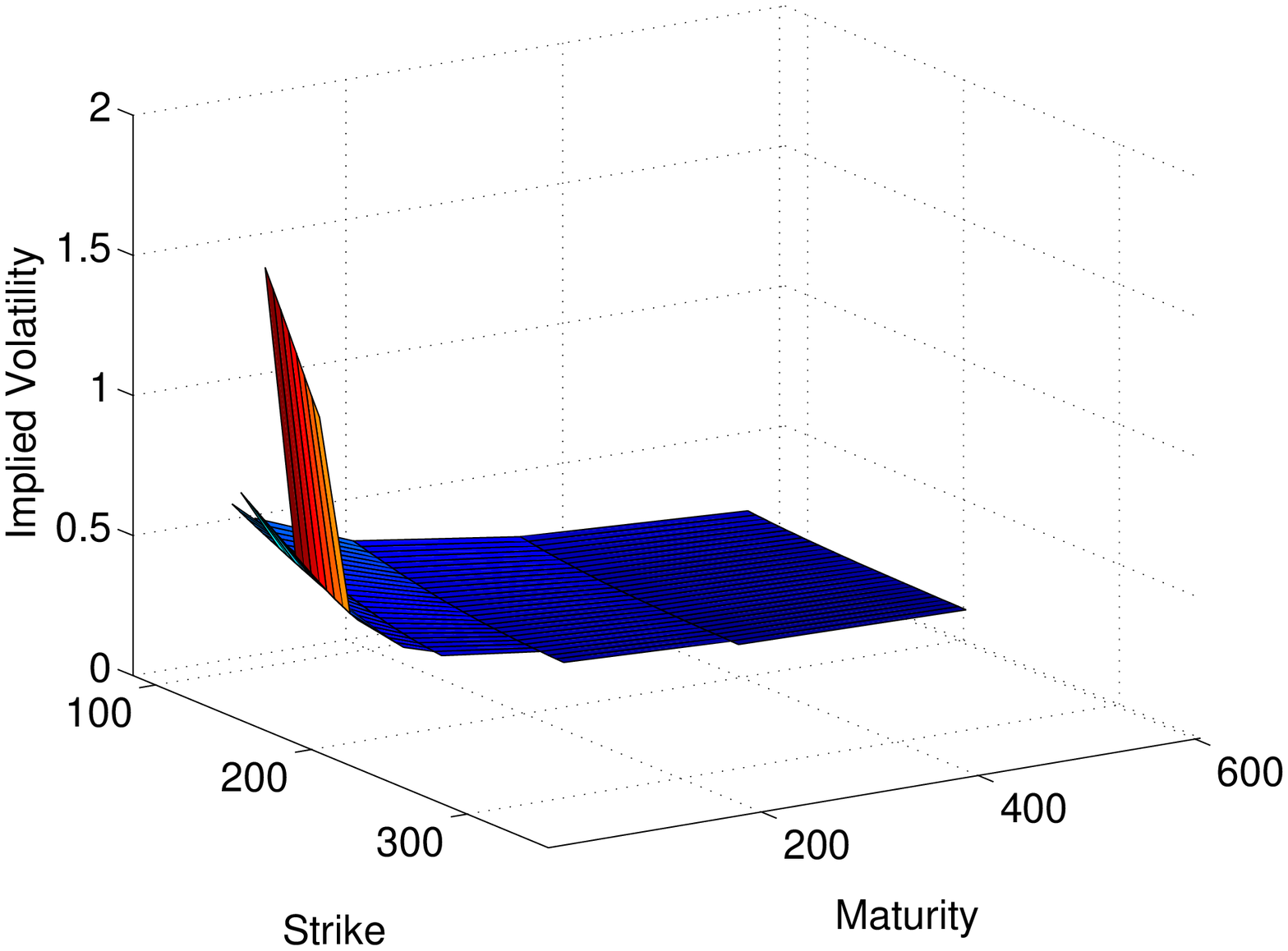}\\
\includegraphics[scale=0.25]{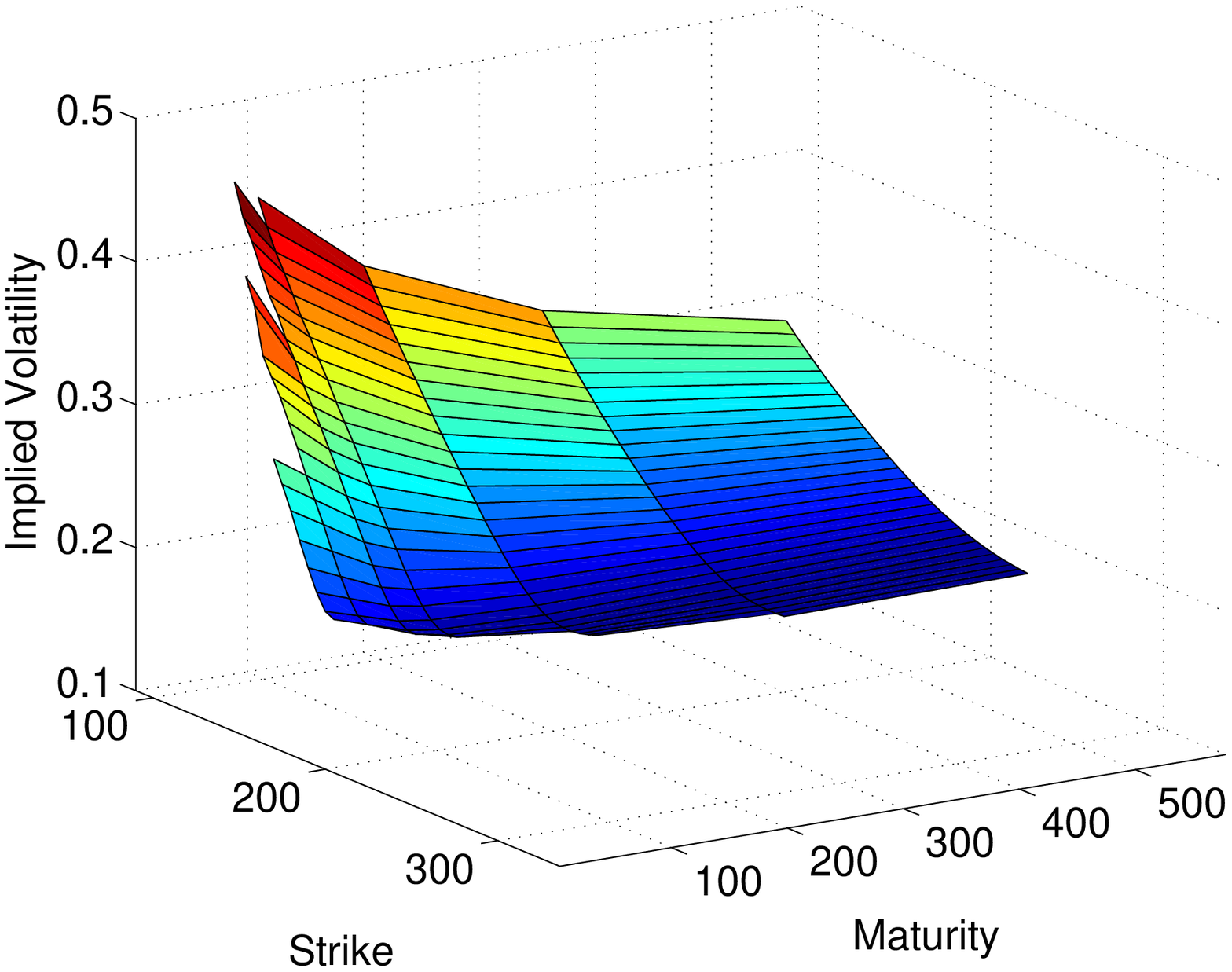}\includegraphics[scale=0.25]{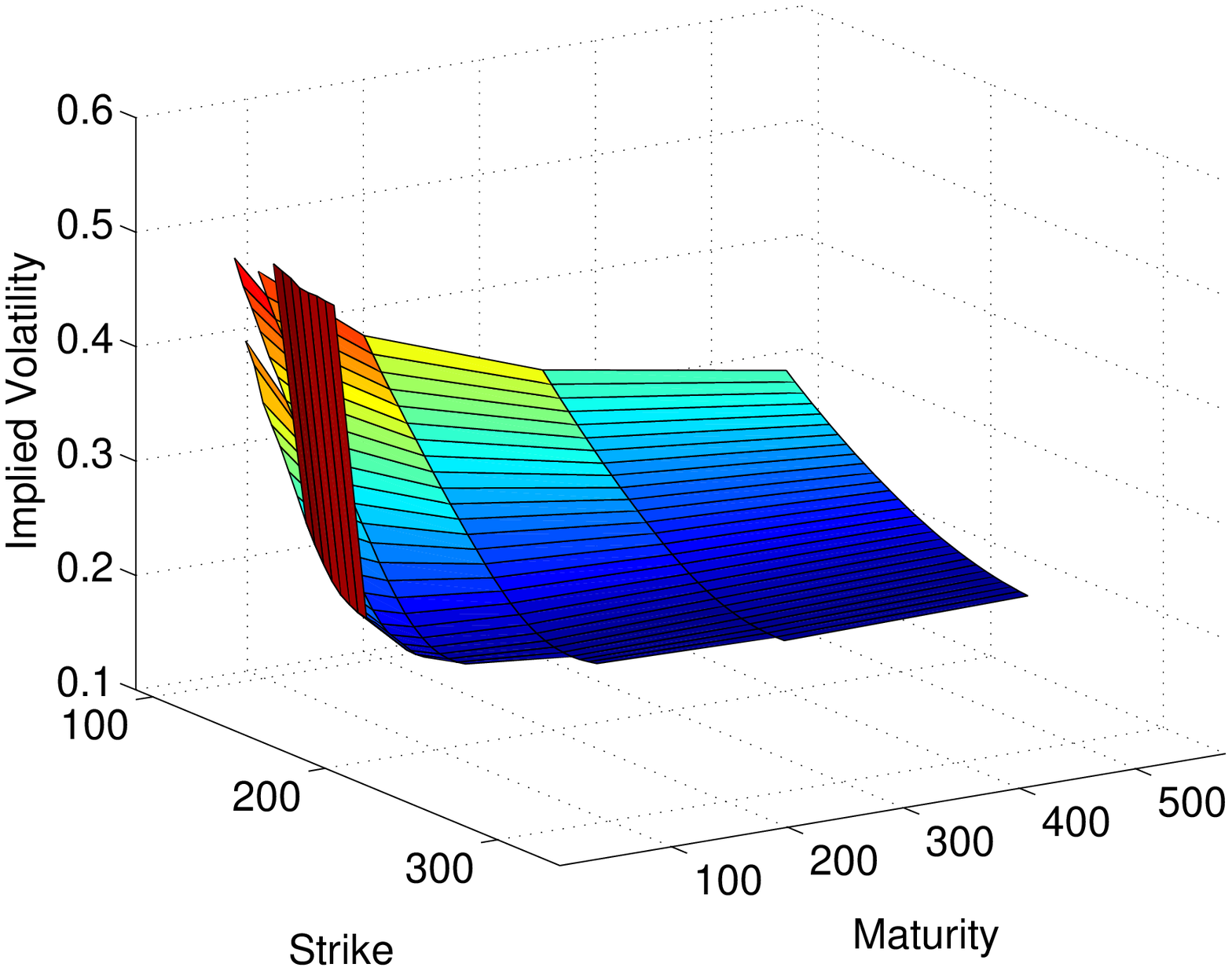}\includegraphics[scale=0.25]{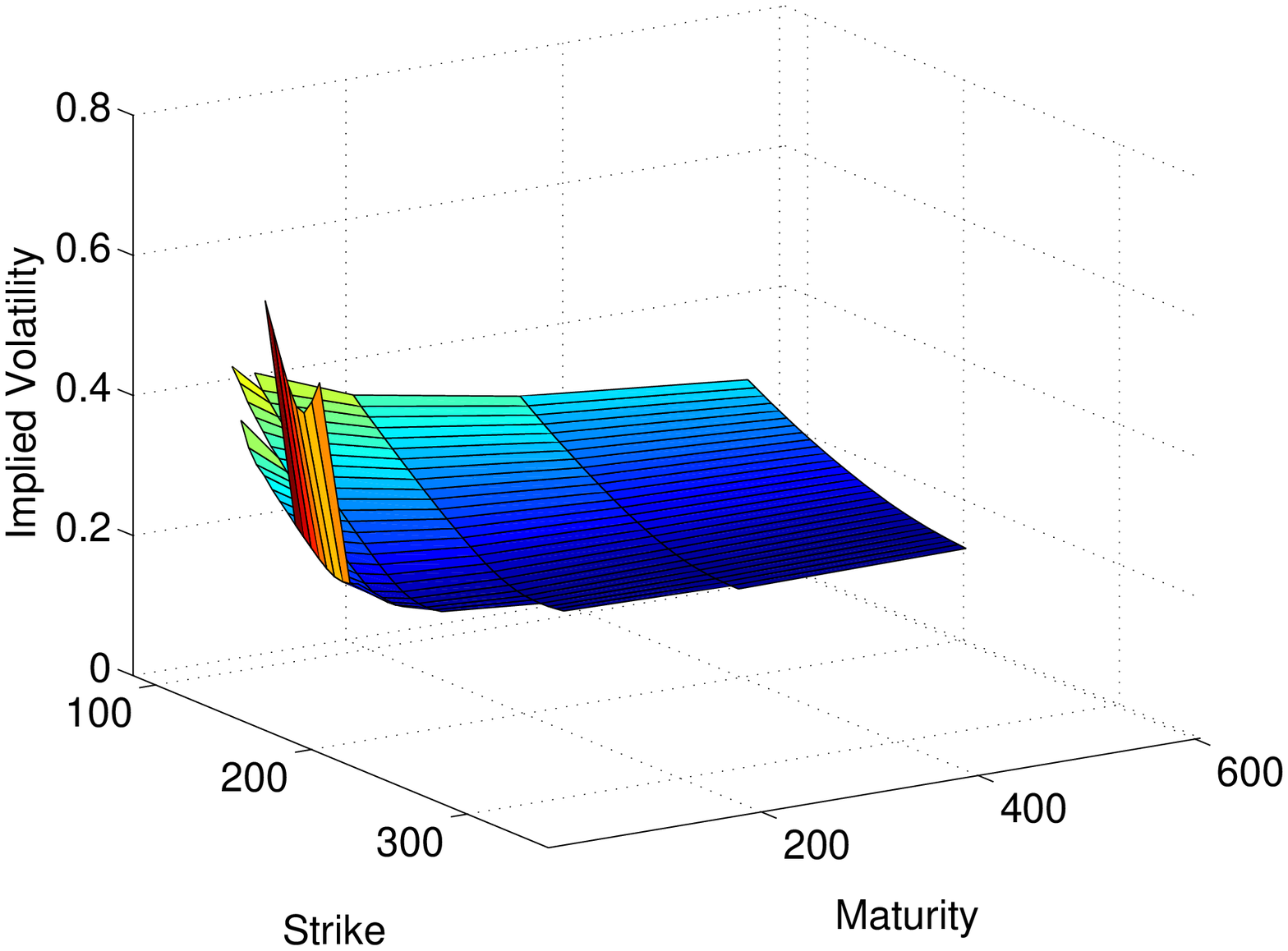}\\
\includegraphics[scale=0.25]{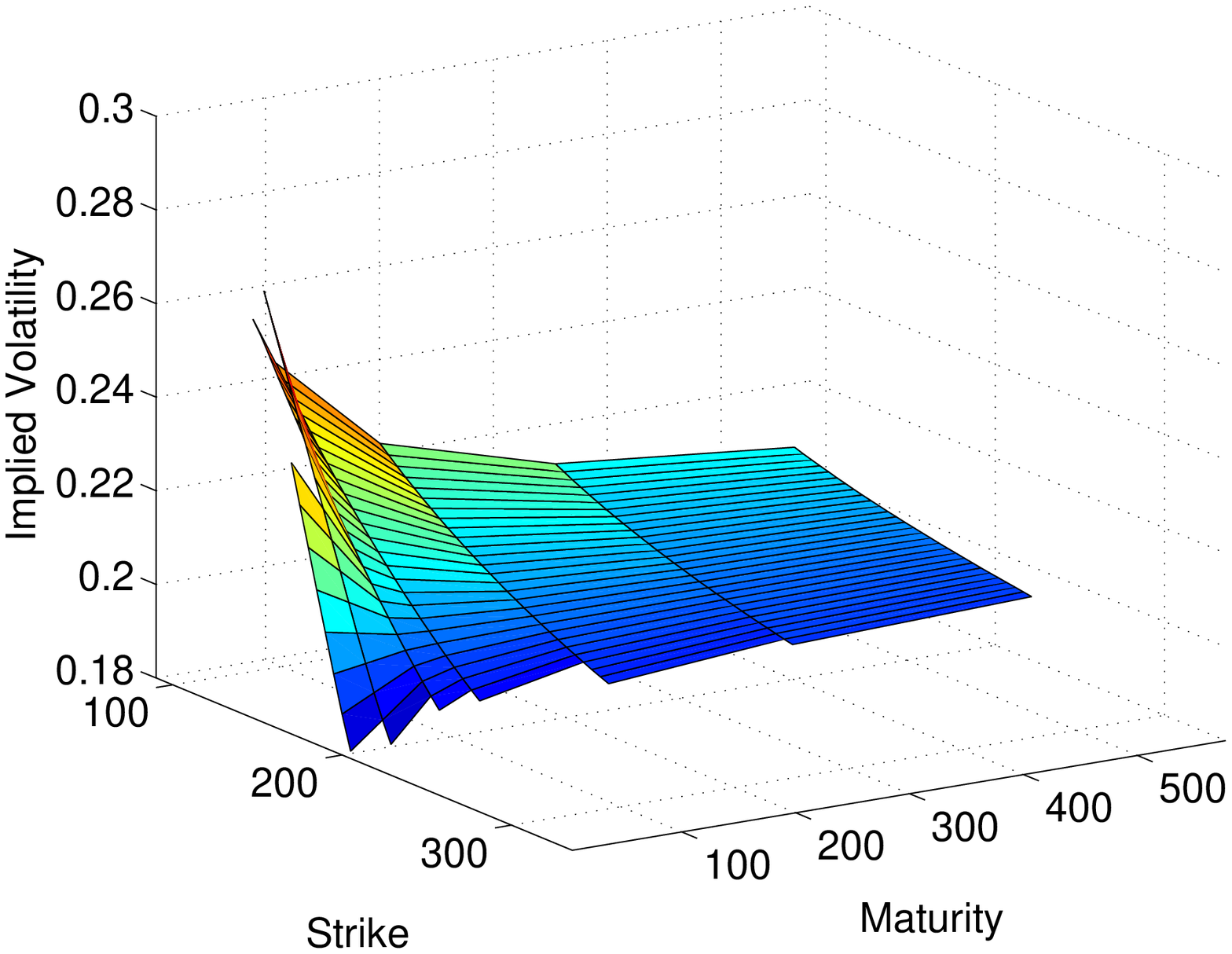}\includegraphics[scale=0.25]{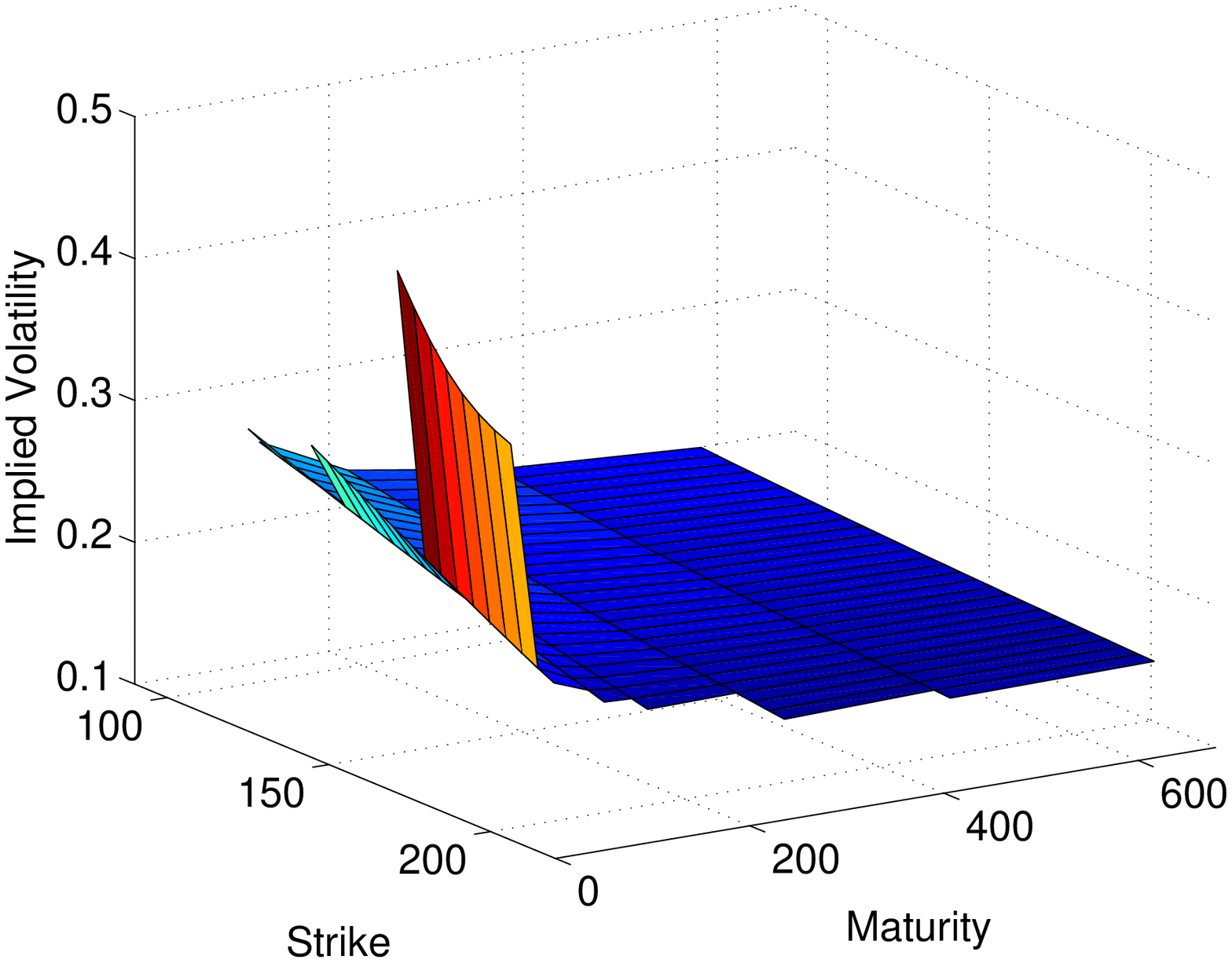}\includegraphics[scale=0.25]{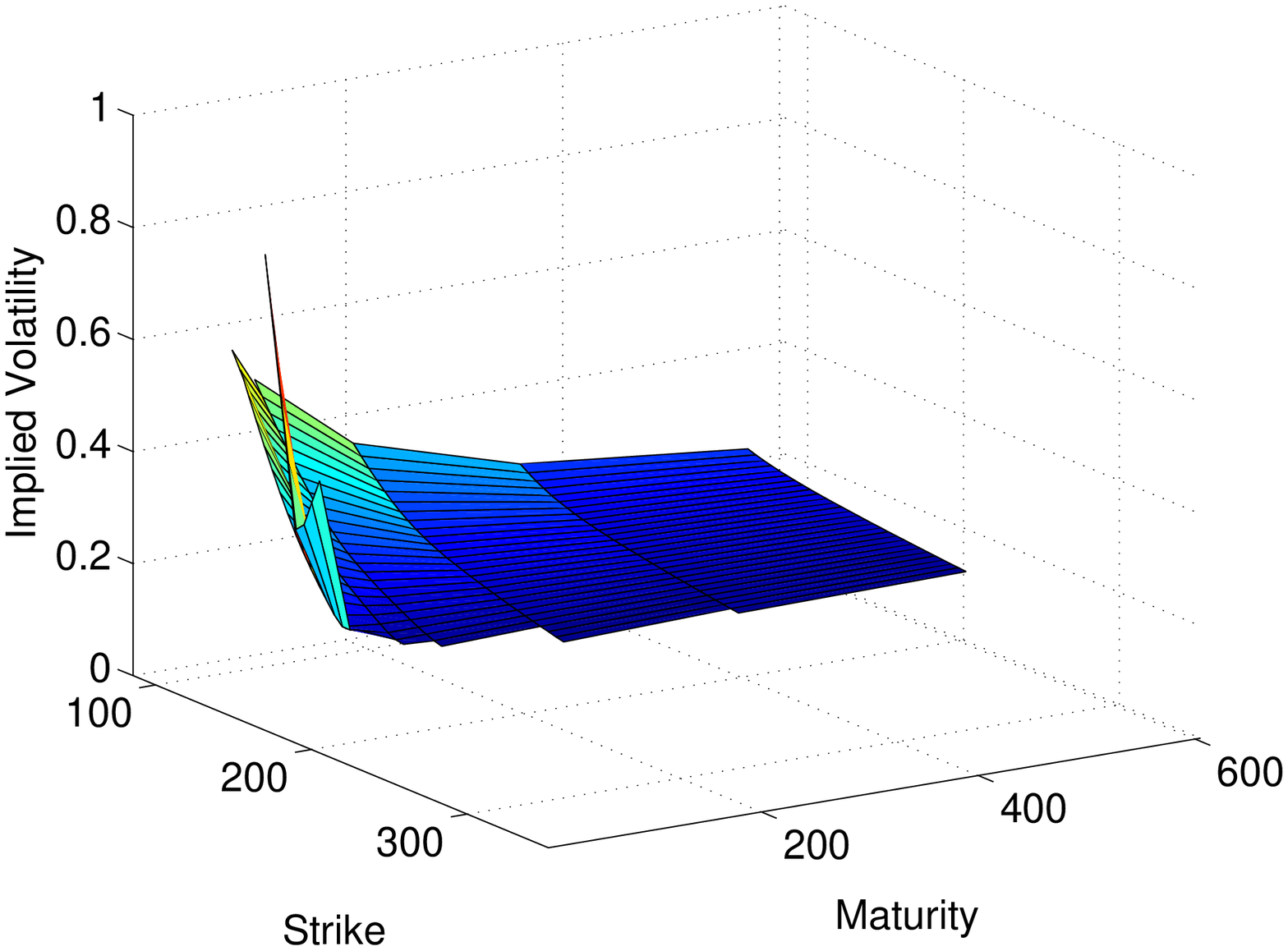}\caption{\small Calibrated surfaces for the Black-Scholes, Heston, and Kou models (1st, 2nd, 3rd rows respectively) without EA jump (left column), and with  Gaussian (middle) and DE (right) EA jumps. The calibrated model parameters are listed in Table \ref{tab:cal_ex_params}.}\label{fig:Cal_example_surfaces}
\end{figure}

\subsection{Analytic Estimators under the Extended BS Model}\label{sub:Analytic-Estimators}
In  the extended BS model (\ref{eq:GBMe}), it is  possible to derive analytical estimators for the models parameters $(\sigma, \sigma_e)$, as discussed in \cite{Dubinsky2006}. We apply these estimators to compare with the estimators obtained by calibrating other models (see Table \ref{tab:risk_premia}).  First, we consider the extended BS model where the EA jump $Z_{e}$ is  normally distributed.  To calibrate this model, it suffices to use any pair of  options of  different maturities. Let $\sigma_{IV}\left(T_{1}\right),\ \sigma_{IV}\left(T_{2}\right)$
represent the implied volatilities of two options with maturities $T_{1}$ and $T_{2}$, respectively. Then, applying  (\ref{eq:BS_IV_ts}), the model  parameters $(\sigma,\sigma_e)$  can be estimated by  
\begin{align}
\sigma^{TS}  =  \sqrt{\frac{\left(T_{1}-t\right)\sigma_{IV}^{2}\left(T_{1}\right)-\left(T_{2}-t\right)\sigma_{IV}^{2}\left(T_{2}\right)}{T_{1}-T_{2}}}, \quad
\sigma_{e}^{TS}  =  \sqrt{\frac{\sigma_{IV}^{2}\left(T_{1}\right)-\sigma_{IV}^{2}\left(T_{2}\right)}{\frac{1}{T_{1}-t}-\frac{1}{T_{2}-t}}}.\label{eq:BS_TS_estimator}
\end{align}
where the superscript $TS$ indicates the relevance of the   IV term structure to these  estimators. In particular, we call  $\sigma_{e}^{TS}$ the \emph{term structure
estimator} of the jump volatility under the risk-neutral measure $\mathbb{Q}$.  We stress that a set of two options
with identical maturity would not allow us to estimate  $\sigma$ and $\sigma_{e}$ separately, but only the aggregate value $\sigma^{2}+\frac{\sigma_{e}^{2}}{T-t}$. 

Alternatively, one can  utilize option prices at different times for parameter estimation. In fact, given the implied volatilities $\sigma_{IV,t_{1}}$
and $\sigma_{IV,t_{2}}$ at times $t_{1}$ and $t_{2}$, with $t_{1}<t_{2}<T_{e}$, we apply \eqref{eq:BS_IV_ts} to get  the following estimators:
\begin{align}
\sigma^{ts}  = \sqrt{\frac{\left(T-t_{1}\right)\sigma_{IV,t_{1}}^{2}\left(T\right)-\left(T-t_{2}\right)\sigma_{IV,t_{2}}^{2}\left(T\right)}{t_{2}-t_{1}}},\quad
\sigma_{e}^{ts}  =  \sqrt{\frac{\sigma_{IV,t_{1}}^{2}\left(T\right)-\sigma_{IV,t_{2}}^{2}\left(T\right)}{\frac{1}{T-t_{1}}-\frac{1}{T-t_{2}}}}.\label{eq:BS_ts_estimator}
\end{align}
They are called   the \emph{time-series estimators} (see also \cite{Dubinsky2006}). 

We observe from  (\ref{eq:BS_TS_estimator}) that one must require that 
$\sigma_{IV}\left(T_{1}\right)>\sigma_{IV}\left(T_{2}\right)$ in order to obtain well-defined estimators.  Similarly for  $\sigma^{ts}$ and $\sigma^{ts}_e$ in (\ref{eq:BS_ts_estimator}), their definitions suggest  that $\sigma_{IV,t_{2}}>\sigma_{IV,t_{1}}$
must hold.  In our empirical tests, we find that, before an earnings announcement, $\sigma_{IV}\left(T_{1}\right)>\sigma_{IV}\left(T_{2}\right)$ always holds, but  the condition $\sigma_{IV,t_{1}}>\sigma_{IV,t_{2}}$ is sometimes  violated. Similar observations are also discussed in  \citet{Dubinsky2006}, who have also conducted  a comprehensive empirical test using ATM options. 

We emphasize that these analytical estimators   are based on a specific extension of the BS model.  Since market prices are not necessarily generated by this model,  the analytical estimators and the calibrated parameters  may  not coincide. Moreover, they  also depend on the choice of  options whose IVs are inputs to  the estimator formulas. On the other hand, the main advantage of these analytical estimators is that they can be computed instantly, and they are also used in practice  (see e.g. \citet{jpmorgan}) and related studies.

\subsection{Implied EA jump Distribution and Risk Premia}\label{sub:premia}
With the choice of a pricing model, our  calibration procedure  extracts  the implied distribution of the EA jump.  One useful application is to compare  the risk-neutral and historical  distributions of the EA jump. Their discrepancy will shed some light on the risk premium associated with the EA jump. 
 As an example, let us consider the empirical  EA jumps of the IBM stock starting from 1994. We assume that both  risk-neutral and historical  distributions are Gaussian, which is amenable for comparison  since  we only need to estimate a single parameter, i.e.  the EA jump volatility. In Table \ref{tab:risk_premia}, we   report the estimate of EA jump volatility obtained by calibrating the Heston model extended with a Gaussian EA jump, $\sigma_e^Q$, and the empirical EA jump volatility,  $\sigma_e^P$ based on data  from 1994 up to the given date. For comparison, we  also list   the EA jump volatility estimators  according to  (\ref{eq:BS_TS_estimator}).  As we can see, for each given date, the ratio $\sigma_e^P/\sigma_e^Q$ is very  close to 1. This suggests that under the extended Heston model, the EA jump distributions are very similar under both historical and risk-neutral measures. On the other hand,  the ratio $\sigma_e^P/\sigma_e^{TS}$  is smaller and less than  1 in this example, suggesting that the extended BS model would imply a  higher EA jump volatility than the empirical one.  In summary,  the volatility  $\sigma_e^Q$ calibrated from  the extended Heston model  is smaller than the term structure EA jump volatility estimator $\sigma_e^{TS}$  which is based on the extended BS model without stochastic volatility.

\begin{table}[H]
  \centering
  \footnotesize
    \begin{tabular}{rrrrrrrr}
    \hline
    Date  &       &  $\sigma_e^P$ & $\sigma_e^Q$ & $\sigma_e^P/\sigma_e^Q$  &       & $\sigma_e^{TS}$ & $\sigma_e^P/\sigma_e^{TS}$ \\
%          &       &       &       &       &       &       &  \\
\hline
    18-Jul-12 &       & 4.64\% & 4.68\% & 99.3\% &       & 5.22\% & 89.0\% \\
    16-Oct-12 &       & 4.62\% & 4.18\% & 110.5\% &       & 4.70\% & 98.4\% \\
    22-Jan-13 &       & 4.62\% & 4.58\% & 100.9\% &       & 5.71\% & 81.0\% \\
    18-Apr-13 &       & 4.61\% & 4.68\% & 98.5\% &       & 5.49\% & 84.1\% \\
    17-Jul-13 &       & 4.63\% & 4.61\% & 100.4\% &       & 5.76\% & 80.3\% \\
    \hline
    \end{tabular}%
  \caption{\small The implied and historical EA jump volatilities for IBM.  For each date in the table, the historical EA jump volatility $\sigma_e^P$ is estimated using price data from 1994 up to that date. The implied volatility $\sigma_e^Q$ is calibrated from the extended Heston model. }
  \label{tab:risk_premia}
\end{table}

\section{American Options\label{sec:american}}

While index options are typically of  European style, most
US equity options are  American-style. In general, the American option pricing
problem does not admit closed-form formulas,  so we discuss a    numerical method  for computing the option price and exercise boundary. In addition, we apply the  analytic results from the European case to approximate the American option price before an earnings announcement.

\subsection{American Option Price and Exercising Boundary\label{sub:american}}

We assume that the stock price evolves according to the extended Kou model with an EA jump defined in (\ref{eq:dedyn}). The value
of the American option is defined by
\begin{equation}
A\left(t,S\right)=\sup_{t\leq\tau\leq T}\mathbb{E}\left\{ e^{-r\left(\tau-t\right)}\left(K-S_{\tau}\right)^{+}\vert S_{t}=S\right\}\,,\quad t\leq T, \label{eq:am_def}
\end{equation}
where $\tau$ is a stopping time w.r.t. the filtration generated by $S$. By the dynamic programming principle, the option price can be written as (see, e.g. \citep[Chap. 10]{Oeksendal2003})
\begin{align}
A\left(t,S\right) & =\sup_{t\leq\tau\leq T_{e}}\mathbb{E}\left\{e^{-r\left(\tau-t\right)}\left( 1_{\left\{ \tau<T_{e}\right\} }\left(K-S_{\tau}\right)^{+}+1_{\left\{ \tau=T_{e}\right\} }A\left(T_{e},S_{T_{e-}}e^{Z_{e}}\right)\right)\vert S_{t}=S\right\} .\label{eq:Bellman}
\end{align}
Therefore, we see that for $t<T_e$ the problem is equivalent to pricing an American option under the Kou model but with ``terminal" payoff $\mathbb{E}\left\{A\left(T_{e},Se^{Z_{e}}\right)\vert S_{T_{e-}}=S\right\}$ at time $T_e$.  Assume that $t<T_e<T$ and 
that $Z_{e}$ admits the p.d.f. $g\left(z\right)$. Then, the American
put price can be written as
\[
A\left(t,S\right)=D\left(t,S\right)1_{\left\{ t<T_{e}\right\} }+\widetilde{D}\left(t,S\right)1_{\left\{ t\geq T_{e}\right\} },
\]
where $D$ and $\widetilde{D}$ satisfy the linear complementarity
problems (see  \citet{Bensoussan1984} and  \citet{Cont2005}):
\begin{equation}
\begin{cases}
\widetilde{D}\left(t,S\right)\geq\left(K-S\right)^{+}, & \quad   T_e \le t<T,\ S\geq0,\\
r\widetilde{D}\left(t,S\right)-\frac{\partial\widetilde{D}}{\partial t}\left(t,S\right)-\mathcal{L}\widetilde{D}\left(t,S\right)\geq0, & \quad   T_e \le t<T,\ S\geq0,\\
\left(\widetilde{D}\left(t,S\right)-\left(K-S\right)^{+}\right)\left(r\widetilde{D}\left(t,S\right)-\frac{\partial\widetilde{D}}{\partial t}\left(t,S\right)-\mathcal{L}\widetilde{D}\left(t,S\right)\right)=0, & \quad   T_e \le t<T,\ S\geq0,\\
\widetilde{D}\left(T,S\right)=\left(K-S\right)^{+}, & \quad S\geq0;
\end{cases}\label{eq:am_eq_1}
\end{equation}
\begin{equation}
\begin{cases}
D\left(t,S\right)\geq\left(K-S\right)^{+}, & \quad  0 \le  t <T_{e} ,\ S\geq0,\\
rD\left(t,S\right)-\frac{\partial{D}}{\partial t}\left(t,S\right)-\mathcal{L}D\left(t,S\right)\geq0, & \quad  0 \le  t <T_{e} ,\ S\geq0,\\
\left(D\left(t,S\right)-\left(K-S\right)^{+}\right)\left(rD\left(t,S\right)-\frac{\partial{D}}{\partial t}\left(t,S\right)-\mathcal{L}D\left(t,S\right)\right)=0, & \quad  0 \le  t <T_{e} ,\ S\geq0,\\
D\left(T_{e},S\right)=\int_{\mathbb{R}} \widetilde{D}\left(T_{e},Se^{z}\right)g\left(z\right)dz\,, & \quad S\geq0.
\end{cases}\label{eq:am_eq_2}
\end{equation}
We have denoted by $\mathcal{L}$ the  infinitesimal generator of $S$ under model (\ref{eq:dedyn}):
\begin{equation}
\mathcal{L}V(S)\equiv \frac{\sigma^2 S^2}{2} \frac{\partial^2 V}{\partial {S^2}} + rS\frac{\partial V}{\partial S} + \kappa \int_{-\infty}^{\infty}\left(V(Se^y)-V(S)\right)f_J(y)dy,
\label{eq:Kou_generator}
\end{equation}and $f_J$ is double exponential p.d.f. 

 It is worth noting  that the integration at time $T_{e}$ in \eqref{eq:am_eq_2} 
must be approximated with a sum because $\tilde{D}$ is not in 
 closed form. The numerical computation of the integration may  introduce computational errors, but it also adds to the computational
burden since  the sum is $\mathcal{O}\left(n^{2}\right)$, where $n$
is the number of discretized  stock price values. We remark that
it is possible, for example, to reduce the complexity of the integration
to $\mathcal{O}\left(n{\log}\left(n\right)\right)$  using an  FFT algorithm. It is also worth noticing that when
$T_{e}=T$, the complexity reduces to $\mathcal{O}\left(n\right)$
if closed form formulas are available for the European option (e.g.
if the EA jump is Gaussian or double-exponential). On the other hand, the complexity also
reduces to $\mathcal{O}\left(n\right)$ when the announcement is imminent, i.e. $T_{e}=0^{+}$. This motivates us to look for a closed-form approximation to  the American option price based on these scenarios, as discussed in Sect. \ref{sub:am_approx}. It is useful 
to compare the American option prices under the same model but with
different earnings announcement dates.  
\begin{prop}
\label{prop:american_bounds}Let $A\left(t,S;u\right)$ denote the American option price as in (\ref{eq:am_def}) with $T_{e}=u$. Then, we have
\begin{equation}
\label{Aineq}
A\left(t,S;l\right)\leq A\left(t,S;u\right)\quad,\quad l\geq u>t.
\end{equation}
\end{prop}
As a result,  the American option price is monotonically decreasing as the EA date $T_{e}$ approaches maturity. Consequently, $A\left(t,S;t\right)$ and $A\left(t,S;T\right)$ become the upper and lower bounds, respectively, for the American option price $A\left(t,S;u\right)$, $t\leq u\leq T$. As an interesting comparison, the European option price is completely independent of  the exact EA date as long as it is at or prior to maturity. In Figure \ref{fig:american} (right \&  top) we show the time-value\footnote{The time-value of an American  put is defined as $A-(K-S)^+$ where $A$ is the put price, $S$ is the spot price and $K$ is the strike. With spot price $S$ fixed, inequality \eqref{Aineq} also holds for the corresponding time-values.} of an American put, with strike, maturity, and spot price fixed, over different announcement times $T_{e}$.  Consistent with  Proposition \ref{prop:american_bounds},
the time-value of the American put  is indeed monotonically decreasing
in $T_{e}$.

To solve problems (\ref{eq:am_eq_1})-(\ref{eq:am_eq_2}) we use the Fourier transform based method  presented in \cite{Jackson2008}. Unless $T_e=T$ or $T_e=0^+$, we solve backward in time for $\tilde{D}$ in (\ref{eq:am_eq_1}), perform the numerical integration and feed it as a terminal condition for problem (\ref{eq:am_eq_2}) which is solved backward in time as well. 

% Indeed, we have seen that $\mbox{lim}_{t\rightarrow T_{e_{-}}}A\left(t,S\right)=\mathbb{E}\left\{ A\left(T_e,Se^{Z_{e}}\right)\vert S_{T_{e}}=S\right\}$ which is always greater than $(K-S)^+$ and thus the boundary, when it exists, must decrease as time approaches $T_e$. 

Figure \ref{fig:american} (left) shows the exercise boundary for different
values of $T_{e}$ under the Kou model with a DE EA jump, along other common parameters. Naturally,  the
scheduled announcement introduces a discontinuity in the exercising
boundary.  We mark  the EA dates with three crosses.  As expected, after the largest date $T_e$
 the three boundaries  coincide.  Interestingly,  the exercise boundary is decreasing rapidly in time near $T_e$, which means that the option holder is  more likely  to wait until  the earnings announcement, rather than exercising immediately. This can also be seen in terms of the option's time-value. In Figure \ref{fig:american} (right \& bottom), we illustrate  that the time-value of an American put, with both strike and spot price fixed, is increasing as time approaches the EA date.     \\
 
\begin{figure} \centering{}
\includegraphics[scale=0.45]{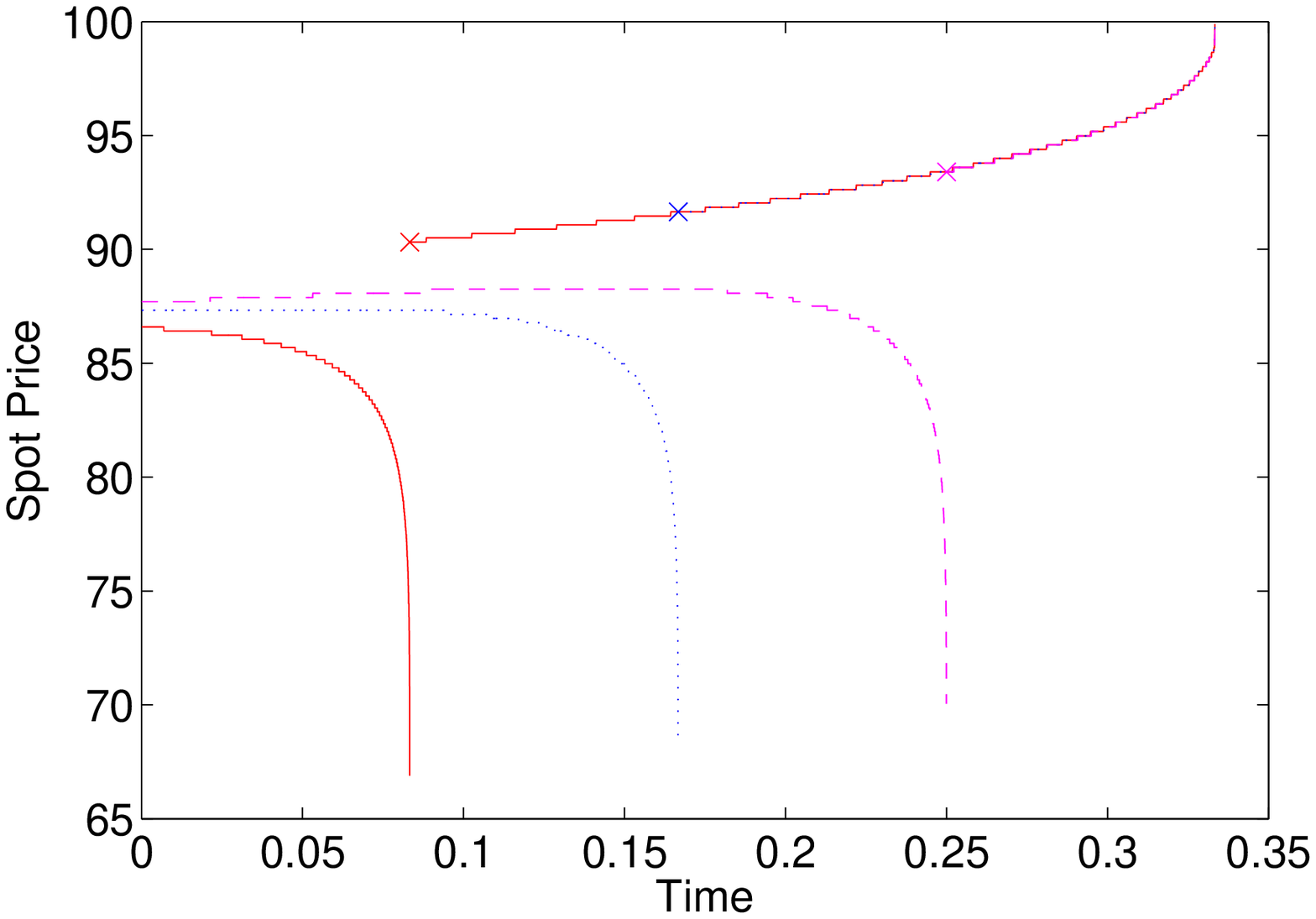}\includegraphics[scale=0.37]{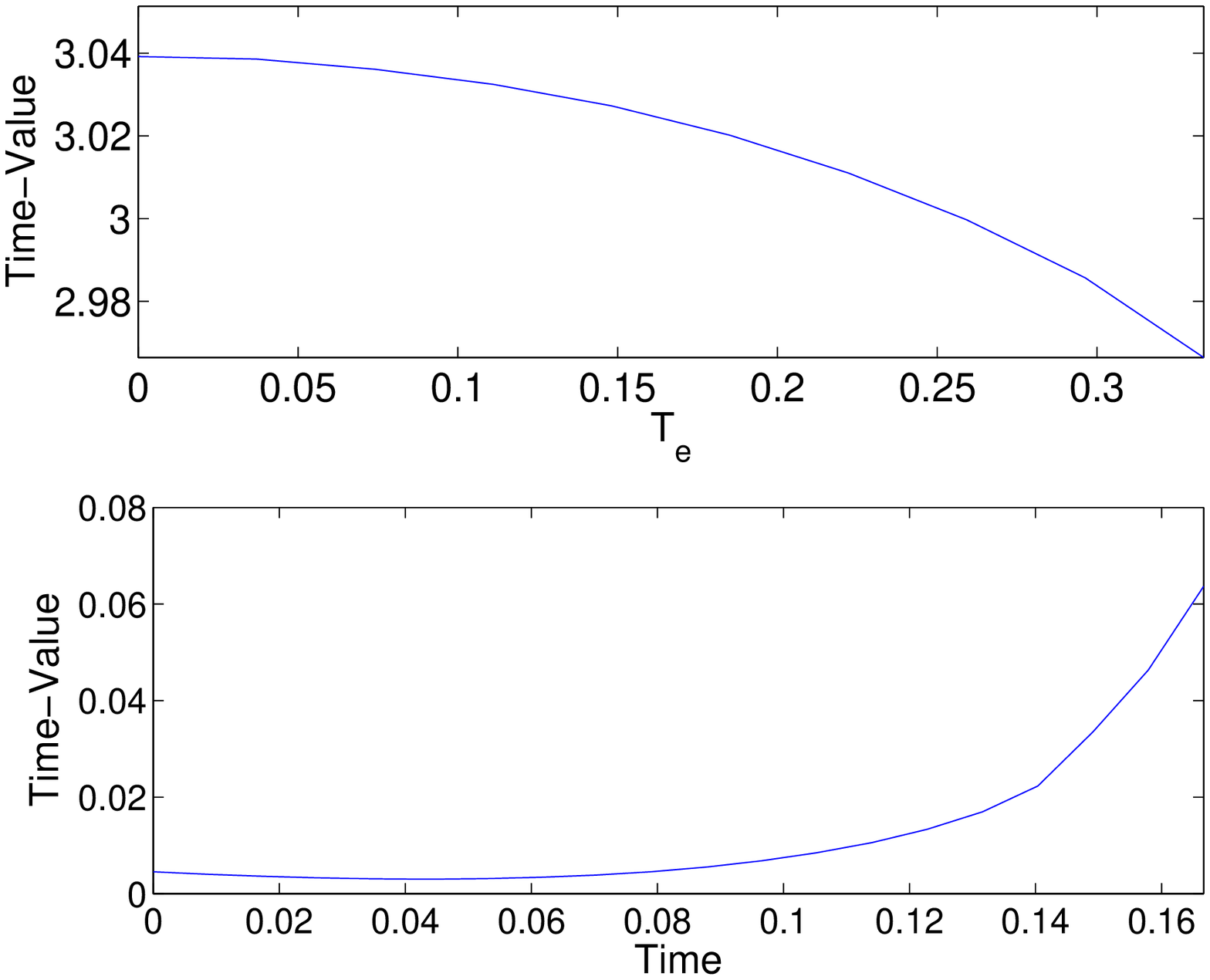}\caption{\small Left: The American put exercising boundary when $T_{e}$ assumes different
values, $T_{e}=1,2,3$ months. Right:
The time-value of a put (top) with spot and strike set at \$99.79 when
$T_{e}$ changes and the time-value of a put (bottom) with strike
\$100 and spot \$87.33 when $T_{e}=2$ months and time $t$ changes.
Parameters: $r=0.02,T=4\mbox{ months}, \sigma=0.1,\kappa=252,p=0.5,\lambda_{1}=300,\lambda_{2}=300,u=0.5,\eta_{1}=30,\eta_{2}=30$\label{fig:american} }
\end{figure}

\subsection{Analytical Approximations\label{sub:am_approx}}
One major method  proposed by  \citet{Barone-Adesi1987}  to  analytically approximate  the price of an American option is to  express  the American option price in terms of an 
 European option price plus a correction term.  The correction
term is determined as the solution of an approximation of the Black-Scholes
equation, with the addition of elementary boundary condition. \citet{Kou2004} present an analytical approximation to the   American option price
under the  Kou model. Here, we adapt 
the Barone-Adesi approximation and apply it to our extended Kou model. 

When $T_{e}=T$ the approximation is virtually identical to the original
one. Let $P_{E}\left(t,S\right)$ denote the European put price (as
in Proposition \ref{prop:KOU_extension}). The approximate price for the American put is similar to that of \cite{Kou2004} and is given by
\begin{equation}
\tilde{A}\left(t,S\right)=\begin{cases}
P_{E}\left(t,S\right)+\gamma_{1}S^{-\beta_{1}}+\gamma_{2}S^{-\beta_{2}} & \quad,\ if\ S>\alpha(t),\\
\left(K-S\right)^{+} & \quad,\ if\ S\leq\alpha(t),
\end{cases}\label{eq:AM_DE_apx-1}
\end{equation}
with the positive constants  
\begin{align}
\gamma_{1} & \equiv\frac{\alpha^{\beta_{1}}}{\beta_{2}-\beta_{1}}\left(\beta_{2}K-\left(1+\beta_{2}\right)\left(\alpha+P_{E}\left(t,\alpha\right)\right)+Ke^{-r\left(T-t\right)}q(\alpha(t)) \right),\label{eq:gamma1} \\
\gamma_{2} & \equiv\frac{\alpha^{\beta_{2}}}{\beta_{1}-\beta_{2}}\left(\beta_{1}K-\left(1+\beta_{1}\right)\left(\alpha+P_{E}\left(t,\alpha\right)\right)+Ke^{-r\left(T-t\right)}q(\alpha(t)) \right),\label{eq:gamma2}
\end{align}
where $q(s)\equiv \mathbb{Q}\left\{ S_{T}\leq K\vert S_{t}=s\right\}$, and  $\beta_{1,2}$, $0<\beta_{1}<\lambda_{2}<\beta_{2}<\infty$, are the two positive
solutions to the equation 
\begin{equation}
\frac{r}{1-e^{-r\left(T-t\right)}}=\beta\left(m\kappa-\frac{\sigma^{2}}{2}-r\right)+\frac{\sigma^{2}\beta^{2}}{2}+\kappa\left(p\frac{\lambda_{1}}{\lambda_{1}+\beta}+\left(1-p\right)\frac{\lambda_{2}}{\lambda_{2}-\beta}-1\right).\label{eq:betas}
\end{equation}
Also,  $\alpha(t)\in\left[0,K\right]$ is the solution to the equation 
\begin{equation}
c_{1}K-c_{2}\left(\alpha(t)+P_{E}\left(t,\alpha(t)\right)\right)=\left(c_{1}-c_{2}\right)Ke^{-r\left(T-t\right)}q(\alpha(t)) ,\label{eq:alphat}
\end{equation}
with $c_{1}=\beta_{1}\beta_{2}\left(1+\lambda_{2}\right)$ and $c_{2}=\lambda_{2}\left(1+\beta_{1}\right)\left(1+\beta_{2}\right)$.
The analytical expression for $P_{E}\left(t,S\right)$ and $q(s)$
can be directly obtained from Proposition \ref{prop:KOU_extension}. In contrast to the formula given in \citet{Kou2004} (see equation (7) there), the calculation of  $P_{E}\left(s,t\right)$,  $q(\alpha(t))$, $\gamma_{1,2}$ and $\alpha$  accounts for the EA jump r.v. $Z_{e}$ (see Appendix A.6).  

When $t\leq T_{e}<T$, one could apply the same method used to derive
the approximation above, writing the American option as $\tilde{A}\left(t,S\right)=\mathbb{E}\left\{ e^{-r\left(T_{e}-t\right)}\tilde{A}\left(T_{e},S_{T_{e}}\right)\vert S_{t}=S\right\} +\epsilon\left(t,S\right)$.
The functional form of $\epsilon$ would be identical, $\epsilon\left(t,S\right)=\gamma_{1}S^{-\beta_{1}}+\gamma_{2}S^{-\beta_{2}}$,
because it is derived from the same PIDE (see Appendix A.6).
The  constants $\beta_1$ and $\beta_2$ are indeed solutions to   equation (\ref{eq:betas}),
after adjusting the time parameters accordingly.  In the case  when the EA jump is imminent ($T_{e}=t^{+}$), we need to  evaluate the
expectation $\mathbb{E}\left\{ \tilde{A}\left(t_{+},Se^{Z_{e}}\right)\right\} $,
where $S$ is the stock price at time $t$. From (\ref{eq:AM_DE_apx-1}), 
this amounts to computing 
\begin{equation}
\int_{-\infty}^{+\infty} \left(\left(P_{E}\left(t_{+},Se^{z}\right)+\gamma_{1}e^{-\beta_{1}z}S^{-\beta_{1}}+\gamma_{2}e^{-\beta_{2}z}S^{-\beta_{2}}\right)1_{\left\{Se^{z}>\alpha\right\}}+\left(K-Se^{z}\right)1_{\left\{Se^{z}\leq\alpha\right\}}\right)f_{Z_e}\left(z\right)dz \label{eq:AM_DE_apx-1-1}
\end{equation}
where $\alpha$, $\gamma_{1,2}$, and $\beta_{1,2}$ are determined by
\eqref{eq:gamma1}-\eqref{eq:alphat}. Notice that at time $t$ there is no exercising boundary since  $\tilde{A}\left(t,S\right)\geq\left(K-S\right)^{+}$
due to Jensen's inequality.   In the extended Kou model 
with double-exponential EA jump, semi-closed formulas can be obtained 
in a similar fashion as in Proposition \ref{prop:KOU_extension}.  

In Table \ref{tab:DE_AM_prices},  we present numerical results for
the prices of     American put options with different strikes and maturities under the extended Kou model.
We compare the price computed  from  the Fourier transform method to that  analytical approximations (\ref{eq:AM_DE_apx-1}) and (\ref{eq:AM_DE_apx-1-1}).
For  different choices of model parameters, strikes and maturities, the approximations  perform
very  well when $T_{e}\sim t$ or $T_{e}\sim T$. For the second
set of   parameters (last 6 columns), the EA jump
has a greater impact on the options prices as the tails of the jump
are fatter, and the volatility of the other part of the dynamics is
lower. In this case, the difference between the ``true'' prices
for different values of $T_{e}$ increases, and so does the difference
between the two approximations which might then not be suitable to
approximate the options values for in-the-money options when $T_{e}$
is not close to $t$ or $T$.

\begin{table}[t!]
  \centering
  \footnotesize
    \begin{tabular}{r|rrrrrrrrrrrrr}
    
    \hline
    K\textbackslash $T_e$ & \multicolumn{1}{c}{$T-2$D} & \multicolumn{1}{c}{$1.5$M} & \multicolumn{1}{c}{$3$D} & $T$ & $0$      & --- & &  \multicolumn{1}{c}{$T-2$D} & \multicolumn{1}{c}{$1.5$M} & \multicolumn{1}{c}{$3$D} & $T$ & $0$      & --- \\
    \hline
     	     & FST   & FST   & FST   & $BA_L$ & $BA_H$ & EU & &FST   & FST   & FST   & $BA_L$ & $BA_H$ & EU \\
    \hline
    80    & 0.10  & 0.10  & 0.10  & 0.10  & 0.11  & 0.10  &       & 0.01  & 0.01  & 0.01  & 0.01  & 0.01  & 0.01 \\
    85    & 0.37  & 0.37  & 0.37  & 0.37  & 0.37  & 0.37  &       & 0.05  & 0.05  & 0.05  & 0.05  & 0.05  & 0.05 \\
    90    & 1.02  & 1.02  & 1.02  & 1.02  & 1.02  & 1.02  &       & 0.22  & 0.22  & 0.23  & 0.22  & 0.23  & 0.22 \\
    95    & 2.30  & 2.31  & 2.31  & 2.30  & 2.31  & 2.29  &       & 0.84  & 0.86  & 0.87  & 0.84  & 0.87  & 0.84 \\
    100   & 4.40  & 4.42  & 4.43  & 4.40  & 4.42  & 4.38  &       & 2.54  & 2.60  & 2.63  & 2.54  & 2.62  & 2.53 \\
    105   & 7.35  & 7.39  & 7.40  & 7.36  & 7.38  & 7.32  &       & 5.68  & 5.81  & 5.91  & 5.70  & 5.91  & 5.66 \\
    110   & 11.07 & 11.11 & 11.14 & 11.06 & 11.11 & 10.98 &       & 10.01 & 10.10 & 10.28 & 10.03 & 10.30 & 9.87 \\
    115   & 15.36 & 15.40 & 15.45 & 15.35 & 15.42 & 15.20 &       & 15.00 & 15.00 & 15.08 & 15.00 & 15.10 & 14.57 \\
    120   & 20.06 & 20.07 & 20.14 & 20.04 & 20.13 & 19.77 &       & 20.00 & 20.00 & 20.01 & 20.00 & 20.04 & 19.45 \\
    \hline
    \end{tabular}
 \caption{\small American put prices under the extended Kou model  with a DE EA jump.   The column ``FST" shows the prices calculated by  a Fourier transform method for three different values of $T_e$ with  expiration date  $T=3$ months  fixed. The extended Barone-Adesi approximations (\ref{eq:AM_DE_apx-1}) and (\ref{eq:AM_DE_apx-1-1})  are  given under columns  ``$BA_L$" and ``$BA_U$". Column ``EU" shows the corresponding  European put price (see \eqref{eq:KOU_extension_formulae}). The first 6 columns are calculated with model parameters: $S=100, r=0.02, \sigma=0.2, \kappa=252, p=0.5, \lambda_1=300, \lambda_2=300, u=0.5, \eta_1=30, \eta_2=30$. The last 6 columns with $S=100, r=0.02, \sigma=0.07, \kappa=200, p=0.5, \lambda_1=350, \lambda_2=350, u=0.5, \eta_1=25, \eta_2=25$.}
  \label{tab:DE_AM_prices}
\end{table}

\appendix
\addcontentsline{toc}{section}{Appendix} \addtocontents{toc}{\protect\setcounter{tocdepth}{-1}} 
\section{Appendix}
In this Appendix, we provide a number of detailed proofs and formulas. 
\subsection{Details for Formula \eqref{eq:KOU_extension_formulae}}\label{sect-app1}
In this section we write the expression of the function $\Upsilon$ in (\ref{eq:KOU_extension_formulae}) explicitly. For ease of notation, we refer to the parameters given as input to $\Upsilon$ as a vector  $\Theta\equiv\left(\theta_1,...,\theta_{13}\right)$. In addition, let $\tilde{\Theta}$ be a permutation of the vector $\Theta$ where only the 8th and 9th components (the parameters for  the randomly timed jumps) are switched. The function $\Upsilon$ is given by 
\begin{equation}
\Upsilon(\Theta) = \sum_{n=0}^{\infty}\frac{\left(
\theta_6\theta_4\right)^{n}e^{-
\theta_6\theta_4}}{n!}Z_{n}\left(\Theta\right), \label{eq:upsilon}
\end{equation}
where 
\begin{eqnarray*}
Z_{n}\left(\Theta\right) & = & \theta_{10}\sum_{l=0}^{n}\left(P_{n,l}T_{1,n}\left(k,\Theta\right)+Q_{n,l}T_{2,n}\left(k,\Theta\right)\right)+\theta_{11}\sum_{l=0}^{n}\left(P_{n,l}T_{3,n}\left(k,\Theta\right)+Q_{n,l}T_{4,n}\left(k,\Theta\right)\right),
\end{eqnarray*}
 \[ k={\log}\left(\frac{\theta_2}{\theta_1}\right)-\left(\theta_3-\frac{\theta_5^{2}}{2}-m\kappa\right)\theta_4-\alpha,\quad m=p\frac{\lambda_{1}}{\lambda_{1}-1}+q\frac{\lambda_{2}}{\lambda_{2}+1}-1,\]
\begin{align*}
T_{1,n+1}\left(s,\Theta\right) & =\frac{\theta_{8}}{\theta_{8}-\theta_{12}}\left(T_{1,n}\left(s,\Theta\right)-\theta_{12}\frac{e^{\frac{\left(\theta_5\theta_{8}\right)^{2}\theta_4}{2}}}{\sqrt{2\pi}}\left(
\theta_5\sqrt{\theta_4}\theta_{8}\right)^{n}I_{n}\left(s;-\theta_{8},-\frac{1}{\theta_5\sqrt{\theta_4}},-\theta_{8}\theta_5\sqrt{\theta_4}\right)\right),
\end{align*}
 
\begin{align*}
T_{2,n+1}\left(s,\Theta\right) & =\frac{\theta_{9}}{\theta_{9}+\theta_{12}}\left(T_{2,n}\left(s,\Theta\right)+\theta_{12}\frac{e^{\frac{\left(\theta_5\theta_{8}\right)^{2}\theta_4}{2}}}{\sqrt{2\pi}}\left(
\theta_5\sqrt{\theta_4}\theta_{9}\right)^{n}I_{n}\left(s;\theta_{9},\frac{1}{\theta_5\sqrt{\theta_4}},-\theta_{9}\theta_5\sqrt{\theta_4}\right)\right),
\end{align*}
 
\begin{eqnarray*}
T_{3,n+1}\left(s,\Theta\right) = & 1-T_{2,n+1}\left(-s,\tilde{\Theta}\right),\quad 
T_{4,n+1}\left(s,\Theta\right) = & 1-T_{1,n+1}\left(-s,\tilde{\Theta}\right),
\end{eqnarray*}
 
\[
T_{1,0}\left(s,\Theta\right)= T_{2,0}\left(s,\Theta\right)=\theta_{12}\frac{e^{\left(\theta_5\theta_{12}\right)^{2}\theta_4/2}}{\sqrt{2\pi}}I_{0}\left(s;-\theta_{12},-\frac{1}{\theta_5\sqrt{\theta_4}},-\theta_{12}\theta_5\sqrt{\theta_4}\right),
\]
 
\begin{eqnarray*}
T_{1,1}\left(s,\Theta\right) & = & \frac{1}{\frac{1}{\theta_{8}}-\frac{1}{\theta_{12}}}\left(\frac{e^{\frac{\left(\theta_5\theta_{8}\right)^{2}\theta_4}{2}}}{\sqrt{2\pi}}I_{0}\left(s;-\theta_{8},\frac{-1}{\theta_5\sqrt{\theta_4}},-\theta_{8}\theta_5\sqrt{\theta_4}\right)-\frac{e^{\frac{\left(\theta_5\theta_{12}\right)^{2}\theta_4}{2}}}{\sqrt{2\pi}}I_{0}\left(s;-\theta_{12},\frac{-1}{\theta_5\sqrt{\theta_4}},-\theta_{12}\theta_5\sqrt{\theta_4}\right)\right),\\
T_{2,1}\left(s,\Theta\right) & = & \frac{1}{\frac{1}{\theta_{9}}-\frac{1}{\theta_{12}}}\left(\frac{e^{\frac{\left(\theta_5\theta_{9}\right)^{2}\theta_4}{2}}}{\sqrt{2\pi}}I_{0}\left(s;-\theta_{9},\frac{-1}{\theta_5\sqrt{\theta_4}},-\theta_{9}\theta_5\sqrt{\theta_4}\right)-\frac{e^{\frac{\left(\theta_5\theta_{12}\right)^{2}\theta_4}{2}}}{\sqrt{2\pi}}I_{0}\left(s;-\theta_{12},\frac{-1}{\theta_5\sqrt{\theta_4}},-\theta_{12}\theta_5\sqrt{\theta_4}\right)\right),
\end{eqnarray*}
 
\begin{align*}
P_{n,m}&= \sum_{i=m}^{n-1}\binom{n-m-1}{i-m}\binom{n}{i}\left(\frac{\theta_{8}}{\theta_{8}+\theta_{9}}\right)^{i-m}\left(\frac{\theta_{9}}{\theta_{8}+\theta_{9}}\right)^{n-i}\theta_7^{i}(1-\theta_7)^{n-i},\\
Q_{n,m} & =\sum_{i=m}^{n-1}\binom{n-m-1}{i-m}\binom{n}{i}\left(\frac{\theta_{8}}{\theta_{8}+\theta_{9}}\right)^{n-i}\left(\frac{\theta_{9}}{\theta_{8}+\theta_{9}}\right)^{i-m}\theta_7^{n-i}(1-\theta_7)^{i},\\
P_{n,n}&=\theta_7^{n}, Q_{n,n}=(1-\theta_7)^{n}, P_{0,0}=1, Q_{0,0}=0,
\end{align*}
\[
Hh_n(x)=\frac{1}{n}\left(Hh_{n-2}(x)-xHh_{n-1}(x)\right),  Hh_0(x) = \int_{-\infty}^{-x}e^{-t^2/2}dt, 
I_{n}\left(k;\alpha,\beta,\delta\right)=-\frac{e^{\alpha k}}{\alpha}Hh_{n}\left(\beta k-\delta\right)+\frac{\beta}{\alpha}I_{n-1},
\]
\[
I_{-1}\left(k;\alpha,\beta,\delta\right)=\frac{\sqrt{2\pi}}{\beta}e^{\frac{\alpha\delta}{\beta}+\frac{\alpha^{2}}{2\beta^{2}}}\begin{cases}
\Phi\left(-\beta k+\delta+\frac{\alpha}{\beta}\right) & if\:\beta>0\ \alpha\neq0,\\
-\Phi\left(\beta k-\delta-\frac{\alpha}{\beta}\right) & if\:\beta<0\ \alpha<0.
\end{cases}
\]
As suggested by these definitions, the implementation of the formula involves successive computation of functions, followed by a summation in  (\ref{eq:upsilon}) that  must be truncated. In order to control the error, we notice   that\footnote{This can be verified directly from the proof of Proposition \ref{prop:KOU_extension} presented in \ref{apx:kou}} $Z_n\leq 2$, and the error bound 
\[ 
\Upsilon\left(\Theta\right) - \sum_{n=0}^{M}\frac{\left(
\theta_6\theta_4\right)^{n}e^{-\theta_6\theta_4}}{n!}Z_{n}\left(\Theta\right) \leq 2\sum_{n=M+1}^{\infty}\frac{\left(
\theta_6\theta_4\right)^{n}e^{-\theta_6\theta_4}}{n!} = 2 - 2\sum_{n=0}^{M}\frac{\left(
\theta_6\theta_4\right)^{n}e^{-\theta_6\theta_4}}{n!}\equiv \epsilon(\theta_6\theta_4,M).
\]
Therefore, if we truncate the summations in  (\ref{eq:KOU_extension_formulae}) at the $M$th term,   the  upper bound for the truncation error is given by
\[
e^{-\alpha}S\epsilon\left((m+1)\kappa T,M\right)+e^{-rT}K\epsilon\left(\kappa T,M\right).
\]
This can be computed quickly and  used to limit  the error to the a priori desired decimal place. For example, take  $S=K=100$, $T=1$ and an  error tolerance of  $0.01$. Then, with  $\kappa\approx 100$, and   $m, \alpha< 0.001$, we obtain $M= 143$.

%Several of the functions introduced above are intentionally defined in a recursive way. This is convenient and computationally efficient because these functions are needed in consecutive order and because the non-recursive expressions involve summations that would be redundant and would take long to compute. Finally, notice that in practice the summation series in (\ref{eq:upsilon}) must be truncated, introducing a computation error. 

\subsection{Proof of Propositions \ref{prop:KOU_extension}\label{apx:kou}}

In order to price the European call option, we need to evaluate the
terms $\mathbb{Q}\left\{ S_{T}>K\right\} $ and $\mathbb{E}\left\{ e^{S_{T}}1_{\left\{ S_{T}>K\right\} }\right\} $. We first state some useful facts (see   \citet{Kou2002} for proofs). 
\begin{lemma}
Define two i.i.d. exponential r.v.'s,  namely,  $J_{i}^{+}\sim Exp(\lambda_{1})$,
$J_{i}^{-}\sim Exp(\lambda_{2})$. For every $n\geq1$,  we have 
\begin{eqnarray*}
\sum_{i=1}^{n}J_{i} & \overset{d}{=} & \begin{cases}
\sum_{i=1}^{m} & J_{i}^{+}\,,\, w.p.\ P_{n,m},\\
\sum_{i=1}^{m} & J_{i}^{-}\,,\, w.p.\ Q_{n,m},
\end{cases}
\end{eqnarray*}
 where 
\begin{align*}
P_{n,m}&=  \sum_{i=m}^{n-1}\binom{n-m-1}{i-m}\binom{n}{i}\left(\frac{\lambda_{1}}{\lambda_{1}+\lambda_{2}}\right)^{i-m}\left(\frac{\lambda_{2}}{\lambda_{1}+\lambda_{2}}\right)^{n-i}p^{i}q^{n-i},\\
Q_{n,m} & =\sum_{i=m}^{n-1}\binom{n-m-1}{i-m}\binom{n}{i}\left(\frac{\lambda_{1}}{\lambda_{1}+\lambda_{2}}\right)^{n-i}\left(\frac{\lambda_{2}}{\lambda_{1}+\lambda_{2}}\right)^{i-m}p^{n-i}q^{i},\\
P_{n,n}&=p^{n} ,\ Q_{n,n}=q^{n},\ P_{0,0}=1,\ Q_{0,0}=0,\ q=1-p.
\end{align*}
\end{lemma}
 
Next, for every $n\geq0$, we define the functions\footnote{The $Hh$ functions can be calculated using either of the following facts: 
\begin{align}
Hh_{n}&=2^{-\frac{n}{2}}\sqrt{\pi}e^{-\frac{x^{2}}{2}}\left(\frac{_{1}F_{1}\left(\frac{n+1}{2},\frac{1}{2},\frac{x^{2}}{2}\right)}{\sqrt{2}\Gamma\left(1+\frac{n}{2}\right)}-x\frac{_{1}F_{1}\left(\frac{n}{2}+1,\frac{3}{2},\frac{x^{2}}{2}\right)}{\Gamma\left(\frac{n+1}{2}\right)}\right),\notag\\
nHh_{n}\left(x\right)&=Hh_{n-2}\left(x\right)-xHh_{n-1}\left(x\right)\ ,\ n\geq1,\notag
\end{align}
 where $_{1}F_{1}$ denotes the confluent hypergeometric function
and $\Gamma$ is the gamma function.} 
\[
Hh_{n}\left(x\right)=\int_{x}^{\infty}Hh_{n-1}\left(y\right)dy=\frac{1}{n!}\int_{x}^{\infty}\left(t-x\right)^{n}e^{-\frac{t^{2}}{2}}dt,
\]
 
\[
Hh_{-1}\left(x\right)=  e^{-x^{2}/2}\,,\,\ Hh_{0}\left(x\right)=\int_{x}^{\infty}e^{-\frac{t^{2}}{2}}dt=\sqrt{2\pi}\Phi\left(-x\right),
\]
 where $\Phi\left(x\right)$ denotes the standard normal c.d.f. 
  In addition,      for $n\geq 0$, define  the integral
\[
I_{n}\left(k,\alpha,\beta,\delta\right)=\int_{k}^{\infty}e^{\alpha x}Hh_{n}\left(\beta x-\delta\right)dx.
\]

\begin{lemma}
$\left(i\right)$ If $\beta>0$ and $\alpha\neq0$, then for all $n\geq-1$, we have
\[
I_{n}\left(k,\alpha,\beta,\delta\right)=\frac{-e^{\alpha k}}{\alpha}\sum_{i=0}^{n}\left(\frac{\beta}{\alpha}\right)^{n-i}Hh_{i}\left(\beta k-\delta\right)+\left(\frac{\beta}{\alpha}\right)^{n+1}\frac{\sqrt{2\pi}}{\beta}e^{\frac{\alpha\delta}{\beta}+\frac{\alpha^{2}}{2\beta^{2}}}\Phi\left(-\beta k+\delta+\frac{\alpha}{\beta}\right).
\]
$\left(ii\right)$ If $\beta<0$ and $\alpha<0$, then for all $n\geq-1$, we have
\[
I_{n}\left(k,\alpha,\beta,\delta\right)=\frac{-e^{\alpha k}}{\alpha}\sum_{i=0}^{n}\left(\frac{\beta}{\alpha}\right)^{n-i}Hh_{i}\left(\beta k-\delta\right)-\left(\frac{\beta}{\alpha}\right)^{n+1}\frac{\sqrt{2\pi}}{\beta}e^{\frac{\alpha\delta}{\beta}+\frac{\alpha^{2}}{2\beta^{2}}}\Phi\left(\beta k-\delta-\frac{\alpha}{\beta}\right).
\]

In particular, 
\[
I_{-1}\left(k;\alpha,\beta,\delta\right)=\frac{\sqrt{2\pi}}{\beta}e^{\frac{\alpha\delta}{\beta}+\frac{\alpha^{2}}{2\beta^{2}}}\begin{cases}
\Phi\left(-\beta k+\delta+\frac{\alpha}{\beta}\right) & if\:\beta>0, \, \alpha\neq0,\\
-\Phi\left(\beta k-\delta-\frac{\alpha}{\beta}\right) & if\:\beta<0\,, \alpha<0.
\end{cases}
\]
 In addition, under either assumption (i) or (ii) over the parameters
$\alpha$ and $\beta$ given above,   the  functions ($I_{n}$) satisfy
the   recursive relation
\[
I_{n}\left(k;\alpha,\beta,\delta\right)=-\frac{e^{\alpha k}}{\alpha}Hh_{n}\left(\beta k-\delta\right)+\frac{\beta}{\alpha}I_{n-1}\left(k;\alpha,\beta,\delta\right).
\]
\end{lemma}

In our extended Kou model, we need to understand the distribution of the
sum of double exponentially distributed randomly timed and EA jumps. Hence, we consider the associated p.d.f.'s.

\begin{lemma}
Let $Z_{e}^{+}$, $Z_{e}^{-}$ be i.i.d. exponential r.v.'s, $Z_{e}^{+},\sim Exp(\eta_{1})$,
$Z_{e}^{-},\sim Exp(\eta_{2})$. In addition, let $J_{i}^{+},\ J_{i}^{-},\ Z_{e}^{+},\ Z_{e}^{-}$
be independent. Then,  we have
\begin{equation}
f_{\sum_{1}^{n}J_{i}^{+}+Z_{e}^{+}}\left(t\right)=\left(\frac{1}{\frac{1}{\lambda_{1}}-\frac{1}{\eta_{1}}}\right)^{n}\left(A_{n}\left(t\right)e^{-\lambda_{1}t}+B_{n}\left(t\right)e^{-\eta_{1}t}\right)\ ,\quad t>0,  \lambda_{1}\neq\eta_{1},\label{eq:f_++}
\end{equation}
 
\begin{equation}
f_{-\sum_{1}^{n}J_{i}^{-}+Z_{e}^{+}}\left(t\right)=\left(\frac{1}{\frac{1}{\lambda_{2}}+\frac{1}{\eta_{1}}}\right)^{n}C_{n}\left(t\right)e^{\lambda_{2}\hbox{min}\left(0,t\right)-\eta_{1}\max\left(0,t\right)}\ ,\quad t\in \R,  \lambda_{2}\neq\eta_{1},\label{eq:f_-+}
\end{equation}
 where $f_{Y}$ denotes the p.d.f.  of $Y$, and 
\begin{align*}
B_{n}\left(t\right)= & -\eta_{1}^{-\left(n-1\right)},A_{n}\left(t\right)=1+\sum_{i}^{n-1}\eta_{1}^{-\left(i-1\right)}\left(\frac{1}{\lambda_{1}}-\frac{1}{\eta_{1}}\right)^{n-i}\lambda_{1}^{n-i}\frac{\left(t\right)^{n-i}}{\left(n-i\right)!}, \quad t >0,
\end{align*}
  
\[
C_{n}\left(t\right)=\frac{C_{n-1}\left(t\right)}{\lambda_{2}}+\lambda_{2}^{n-1}\frac{\left({\max}\left(0,t\right)-t\right)^{n-1}}{(n-1)!}, \quad t\in \R.
\]
 In addition, they  satisfy the following recursive relations: 
\begin{equation}
f_{\sum_{1}^{n+1}J_{i}^{+}+Z_{e}^{+}}\left(t\right)=\frac{\lambda_{1}}{\left(\lambda_{1}-\eta_{1}\right)}\left(f_{\sum_{1}^{n}J_{i}^{+}+Z_{e}^{+}}\left(t\right)-\lambda_{1}^{n}\frac{\left(t\right)^{n}}{n!}\eta_{1}e^{-\lambda_{1}t}\right) ,\label{eq:frec_++}
\end{equation}
 
\begin{equation}
f_{-\sum_{1}^{n+1}J_{i}^{-}+Z_{e}^{+}}\left(t\right)=\frac{\lambda_{2}}{\lambda_{2}+\eta_{1}}\left(f_{-\sum_{1}^{n}J_{i}^{-}+Z_{e}^{+}}\left(t\right)+\lambda_{2}^{n}\frac{\left(\max\left(0,t\right)-t\right)^{n}}{n!}\eta_{1}e^{\lambda_{2}t}e^{-\left(\lambda_{2}+\eta_{1}\right)\max\left(0,t\right)}\right).\label{eq:frec_-+}
\end{equation}
\end{lemma}

\begin{proof}
We first notice that  the positive  r.v., $J_{i}^{+}+Z_{e}^{+}$,  has the p.d.f.
\begin{align*}
 \int_{0}^{t}\lambda_{1}e^{-\lambda_{1}\left(t-x\right)}\eta_{1}e^{-\eta_{1}x}dx 
 & =\frac{1}{\frac{1}{\lambda_{1}}-\frac{1}{\eta_{1}}}\left(e^{-\lambda_{1}t}-e^{-\eta_{1}t}\right)\\
 & \equiv\left(\frac{1}{\frac{1}{\lambda_{1}}-\frac{1}{\eta_{1}}}\right)\left(A_1(t)e^{-\lambda_{1}t}+B_1(t)e^{-\eta_{1}t}\right), \quad t>0, \lambda_1 \neq \eta_1.
\end{align*}
 Now suppose   for a fixed $n\ge 1$, $f_{\sum_{1}^{n}J_{i}^{+}+Z_{e}^{+}}\left(t\right)=\left(\frac{1}{\frac{1}{\lambda_{1}}-\frac{1}{\eta_{1}}}\right)^{n}\left(A_n(t)e^{-\lambda_{1}t}+B_n(t)e^{-\eta_{1}t}\right)$,
we obtain 
\begin{align*}
f_{\sum_{1}^{n}J_{i}^{+}+Z_{e}^{+}}\left(t\right) & =\int_{0}^{t}\lambda_{1}^{n}\frac{\left(t-x\right)^{n-1}}{\left(n-1\right)!}e^{-\lambda_{1}\left(t-x\right)}\eta_{1}e^{-\eta_{1}x}dx\\
 & =\lambda_{1}^{n}\frac{\left(t\right)^{n}}{n!}\eta_{1}e^{-\lambda_{1}t}+\frac{\left(\lambda_{1}-\eta_{1}\right)}{\lambda_{1}}f_{\sum_{1}^{n+1}J_{i}^{+}+Z_{e}^{+}}\left(t\right),\\
\Longrightarrow\ \ f_{\sum_{1}^{n+1}J_{i}^{+}+Z_{e}^{+}}\left(t\right) & =\frac{\lambda_{1}}{\left(\lambda_{1}-\eta_{1}\right)}\left(f_{\sum_{1}^{n}J_{i}^{+}+Z_{e}^{+}}\left(t\right)-\lambda_{1}^{n}\frac{\left(t\right)^{n}}{n!}\eta_{1}e^{-\lambda_{1}t}\right)\\
 & =\left(\frac{1}{\frac{1}{\lambda_{1}}-\frac{1}{\eta_{1}}}\right)^{n+1}\left(A_{n+1}(t)e^{-\lambda_{1}t}+B_{n+1}(t)e^{-\eta_{1}t}\right) ,
\end{align*}
 where the coefficients satisfy 
\begin{align*}
B_{n+1} & =\frac{1}{\eta_{1}}B_n, \quad A_{n+1}=\frac{A_n}{\eta_{1}}-C_{n+1}, \quad C_{n+1}=\left(\frac{1}{\lambda_{1}}-\frac{1}{\eta_{1}}\right)^{n}\lambda_{1}^{n}\frac{\left(t\right)^{n}}{n!},\\
B_n= & -\eta_{1}^{-\left(n-1\right)}, \quad A_n=1+\sum_{i}^{n-1}\eta_{1}^{-\left(i-1\right)}\left(\frac{1}{\lambda_{1}}-\frac{1}{\eta_{1}}\right)^{n-i}\lambda_{1}^{n-i}\frac{\left(t\right)^{n-i}}{\left(n-i\right)!},
\end{align*} 
 and this yields (\ref{eq:f_++}). Next, we note that  the real-valued r.v. $-J_{i}^{-}+Z_{e}^{+}$ has the p.d.f.
\begin{align*}
\int_{\max\left(0,t\right)}^{\infty}\lambda_{2}e^{\lambda_{2}\left(t-x\right)}\eta_{1}e^{-\eta_{1}x}dx 
 =\frac{e^{\lambda_{2}\hbox{min}\left(0,t\right)-\eta_{1}\max\left(0,t\right)}}{\frac{1}{\lambda_{2}}+\frac{1}{\eta_{1}}}, \quad t\in \R, \lambda_2 \neq \eta_1.
\end{align*}
For a fixed $n \ge 1$, assume that  $f_{-\sum_{1}^{n}J_{i}^{-}+Z_{e}^{+}}\left(t\right)=\left(\frac{1}{\frac{1}{\lambda_{2}}+\frac{1}{\eta_{1}}}\right)^{n}C_n(t) e^{\lambda_{2}\hbox{min}\left(0,t\right)-\eta_{1}\max\left(0,t\right)}$, then we get 
\begin{align*}
f_{-\sum_{1}^{n}J_{i}^{-}+Z_{e}^{+}}\left(t\right) & =\int_{\max\left(0,t\right)}^{\infty}\lambda_{2}^{n}\frac{\left(-t+x\right)^{n-1}}{\left(n-1\right)!}e^{\lambda_{2}\left(t-x\right)}\eta_{1}e^{-\eta_{1}x}dx\\
 & =-\lambda_{2}^{n}\frac{\left(\max\left(0,t\right)-t\right)^{n}}{n!}\eta_{1}e^{\lambda_{2}t}e^{-\left(\lambda_{1}+\eta_{1}\right)\max\left(0,t\right)}+\frac{\left(\lambda_{2}+\eta_{1}\right)}{\lambda_{2}}f_{-\sum_{1}^{n+1}J_{i}^{-}+Z_{e}^{+}}\left(t\right),\\
\Longrightarrow f_{-\sum_{1}^{n+1}J_{i}^{-}+Z_{e}^{+}}\left(t\right) & =\frac{\lambda_{2}}{\lambda_{2}+\eta_{1}}\left(f_{-\sum_{1}^{n}J_{i}^{-}+Z_{e}^{+}}\left(t\right)+\lambda_{2}^{n}\frac{\left(\max\left(0,t\right)-t\right)^{n}}{n!}\eta_{1}e^{\lambda_{2}t}e^{-\left(\lambda_{1}+\eta_{1}\right)\max\left(0,t\right)}\right)\\
 & \equiv\left(\frac{1}{\frac{1}{\lambda_{2}}+\frac{1}{\eta_{1}}}\right)^{n+1}C_{n+1}\left(t\right)e^{\lambda_{2}\hbox{min}\left(0,t\right)-\eta_{1}\max\left(0,t\right)}.
\end{align*}
Matching terms yields  $C_{n+1}\left(t\right)=\frac{C_{n}\left(t\right)}{\lambda_{2}}+\lambda_{2}^{n}\frac{\left(\max\left(0,t\right)-t\right)^{n}}{n!}.$ 
\end{proof}
We now calculate the distribution of the sum of a normal r.v. and   double-exponentials.  
\begin{proposition}
\label{prop:densities}Let $W\sim N\left(0,\sigma^{2}\right)$. Then, we have the p.d.f.'s:
\begin{align}
f_{W+\sum_{i=1}^{n+1}J_{i}^{+}+Z_{e}^{+}}\left(t\right) & =  \left(\frac{\lambda_{1}}{\lambda_{1}-\eta_{1}}\right)^{n}\eta_{1}\frac{e^{\frac{\left(\sigma\eta_{1}\right)^{2}}{2}}}{\sqrt{2\pi}}e^{-\eta_{1}t}Hh_{0}\left(-\frac{t}{\sigma}+\sigma\eta_{1}\right)+\label{eq:fn_++}\\
 &  -\sum_{i=1}^{n}\left(\frac{\lambda_{1}}{\lambda_{1}-\eta_{1}}\right)^{n-i+1}\eta_{1}\frac{e^{\frac{\left(\sigma\lambda_{1}\right)^{2}}{2}}}{\sqrt{2\pi}}\left(\sigma^{i}\lambda_{1}^{i}\right)e^{-\lambda_{1}t}Hh_{i}\left(\frac{-t}{\sigma}+\sigma\lambda_{1}\right)\ ,\ t>0,\, \lambda_{1}\neq\eta_{1},\nonumber 
\end{align}
 
\begin{align}
f_{W-\sum_{i=1}^{n+1}J_{i}^{-}+Z_{e}^{+}}\left(t\right) & =  \left(\frac{\lambda_{2}}{\lambda_{2}+\eta_{1}}\right)^{n}\eta_{1}\frac{e^{\frac{\left(\sigma\eta_{1}\right)^{2}}{2}}}{\sqrt{2\pi}}e^{-\eta_{1}t}Hh_{0}\left(-\frac{t}{\sigma}+\sigma\eta_{1}\right)+\label{eq:fn_-+}\\
 &   +\sum_{i=1}^{n}\left(\frac{\lambda_{2}}{\lambda_{2}+\eta_{1}}\right)^{n-i+1}\eta_{1}\frac{e^{\frac{\left(\sigma\lambda_{2}\right)^{2}}{2}}}{\sqrt{2\pi}}\left(\sigma^{i}\lambda_{2}^{i}\right)e^{\lambda_{2}t}Hh_{i}\left(\frac{t}{\sigma}+\sigma\lambda_{2}\right)\ ,\ t\in \R, \, \lambda_{2}\neq\eta_{1}.\nonumber 
\end{align}
Moreover,   they  admit the  recursive
relations: 
\begin{align}
f_{W+\sum_{i=1}^{n+1}J_{i}^{+}+Z_{e}^{+}}\left(t\right) & =\frac{\lambda_{1}}{\lambda_{1}-\eta_{1}}\left(f_{W+\sum_{i=1}^{n}J_{i}^{+}+Z_{e}^{+}}\left(t\right)-\eta_{1}\frac{e^{\frac{\left(\sigma\lambda_{1}\right)^{2}}{2}}}{\sqrt{2\pi}}\left(\sigma^{n}\lambda_{1}^{n}\right)e^{-\lambda_{1}t}Hh_{n}\left(\frac{-t}{\sigma}+\sigma\lambda_{1}\right)\right),\label{eq:fnrec_++}
\end{align}
 
\begin{align}
f_{W-\sum_{i=1}^{n+1}J_{i}^{-}+Z_{e}^{+}}\left(t\right) & =\frac{\lambda_{2}}{\lambda_{2}+\eta_{1}}\left(f_{W-\sum_{i=1}^{n+1}J_{i}^{-}+Z_{e}^{+}}\left(t\right)+\eta_{1}\frac{e^{\frac{\left(\sigma\lambda_{2}\right)^{2}}{2}}}{\sqrt{2\pi}}\left(\sigma^{n}\lambda_{2}^{n}\right)e^{\lambda_{2}t}Hh_{n}\left(\frac{t}{\sigma}+\sigma\lambda_{2}\right)\right).\label{eq:fnrec_-+}
\end{align}
 \end{proposition}
\begin{proof}  We start with the  p.d.f. for $W + Z_{e}^{+}$:
\[
f_{W+Z_{e}^{+}}\left(t\right)=\eta_{1}\frac{e^{\frac{\left(\sigma\eta_{1}\right)^{2}}{2}}}{\sqrt{2\pi}}e^{-\eta_{1}t}Hh_{0}\left(-\frac{t}{\sigma}+\sigma\eta_{1}\right).
\]
Using (\ref{eq:frec_++}), we also write 
\begin{align*}
f_{W+\sum_{i=1}^{n+1}J_{i}^{+}+Z_{e}^{+}}\left(t\right) & =\int_{-\infty}^{t}f_{\sum_{1}^{n+1}J_{i}^{+}+Z_{e}^{+}}\left(t-x\right)\frac{e^{-\frac{x{}^{2}}{2\sigma^{2}}}}{\sqrt{2\pi\sigma^{2}}}dx\\
 & =\frac{\lambda_{1}}{\lambda_{1}-\eta_{1}}\left(f_{W+\sum_{i=1}^{n}J_{i}^{+}+Z_{e}^{+}}\left(t\right)-\eta_{1}\frac{e^{\frac{\left(\sigma\lambda_{1}\right)^{2}}{2}}}{\sqrt{2\pi}}\left(\sigma^{n}\lambda_{1}^{n}\right)e^{-\lambda_{1}t}Hh_{n}\left(\frac{-t}{\sigma}+\sigma\lambda_{1}\right)\right),
\end{align*}
 which leads directly to 
\begin{eqnarray*}
f_{W+\sum_{i=1}^{n+1}J_{i}^{+}+Z_{e}^{+}}\left(t\right) & = & \left(\frac{\lambda_{1}}{\lambda_{1}-\eta_{1}}\right)^{n}\eta_{1}\frac{e^{\frac{\left(\sigma\eta_{1}\right)^{2}}{2}}}{\sqrt{2\pi}}e^{-\eta_{1}t}Hh_{0}\left(-\frac{t}{\sigma}+\sigma\eta_{1}\right)+\\
 &  & -\sum_{i=1}^{n}\left(\frac{\lambda_{1}}{\lambda_{1}-\eta_{1}}\right)^{n-i+1}\eta_{1}\frac{e^{\frac{\left(\sigma\lambda_{1}\right)^{2}}{2}}}{\sqrt{2\pi}}\left(\sigma^{i}\lambda_{1}^{i}\right)e^{-\lambda_{1}t}Hh_{i}\left(\frac{-t}{\sigma}+\sigma\lambda_{1}\right).
\end{eqnarray*}

To prove  (\ref{eq:fn_-+}), we apply  (\ref{eq:frec_-+}) to get  the recursive
expression 
\begin{align*}
f_{W-\sum_{i=1}^{n+1}J_{i}^{-}+Z_{e}^{+}}\left(t\right) & =\int_{\mathbb{R}}f_{-\sum_{i=1}^{n+1}J_{i}^{-}+Z_{e}^{+}}\left(t-x\right)\frac{e^{-\frac{x{}^{2}}{2\sigma^{2}}}}{\sqrt{2\pi\sigma^{2}}}dx\\
 & =\frac{\lambda_{2}}{\lambda_{2}+\eta_{1}}\left(f_{W-\sum_{i=1}^{n+1}J_{i}^{-}+Z_{e}^{+}}\left(t\right)+\eta_{1}\frac{e^{\frac{\left(\sigma\lambda_{2}\right)^{2}}{2}}}{\sqrt{2\pi}}\left(\sigma^{n}\lambda_{2}^{n}\right)e^{\lambda_{2}t}Hh_{n}\left(\frac{t}{\sigma}+\sigma\lambda_{2}\right)\right),
\end{align*}
 which can be written explicitly as  
\begin{eqnarray*}
f_{W-\sum_{i=1}^{n+1}J_{i}^{-}+Z_{e}^{+}}\left(t\right) & = & \left(\frac{\lambda_{2}}{\lambda_{2}+\eta_{1}}\right)^{n}\eta_{1}\frac{e^{\frac{\left(\sigma\eta_{1}\right)^{2}}{2}}}{\sqrt{2\pi}}e^{-\eta_{1}t}Hh_{0}\left(-\frac{t}{\sigma}+\sigma\eta_{1}\right)+\\
 &  & +\sum_{i=1}^{n}\left(\frac{\lambda_{2}}{\lambda_{2}+\eta_{1}}\right)^{n-i+1}\eta_{1}\frac{e^{\frac{\left(\sigma\lambda_{2}\right)^{2}}{2}}}{\sqrt{2\pi}}\left(\sigma^{i}\lambda_{2}^{i}\right)e^{\lambda_{2}t}Hh_{i}\left(\frac{t}{\sigma}+\sigma\lambda_{2}\right).
\end{eqnarray*}
\end{proof}

We can now calculate the tail probabilities that will allow us to
price the call option. 
\begin{prop}
Let $F_{X}(z)\equiv \Q\left\{ X\leq z\right\} $. Then,
\begin{align}
F_{W+\sum_{i=1}^{n+1}J_{i}^{+}+Z_{e}^{+}}(z) & =  \left(\frac{\lambda_{1}}{\lambda_{1}-\eta_{1}}\right)^{n}\eta_{1}\frac{e^{\frac{\left(\sigma\eta_{1}\right)^{2}}{2}}}{\sqrt{2\pi}}I_{0}\left(z;-\eta_{1},-\frac{1}{\sigma},-\eta_{1}\sigma\right)+\label{eq:fn_++-1}\\
 &   -\sum_{i=1}^{n}\left(\frac{\lambda_{1}}{\lambda_{1}-\eta_{1}}\right)^{n-i+1}\eta_{1}\frac{e^{\frac{\left(\sigma\lambda_{1}\right)^{2}}{2}}}{\sqrt{2\pi}}\left(\sigma^{i}\lambda_{1}^{i}\right)I_{i}\left(z;-\lambda_{1},-\frac{1}{\sigma},-\lambda_{1}\sigma\right)\ ,\quad z>0,\ \lambda_{1}\neq\eta_{1},\nonumber 
\end{align}
\begin{align}
F_{W-\sum_{i=1}^{n+1}J_{i}^{-}+Z_{e}^{+}} (z)& =  \left(\frac{\lambda_{2}}{\lambda_{2}+\eta_{1}}\right)^{n}\eta_{1}\frac{e^{\frac{\left(\sigma\eta_{1}\right)^{2}}{2}}}{\sqrt{2\pi}}I{}_{0}\left(z,-\eta_{1},-\frac{1}{\sigma},-\sigma\eta_{1}\right)+\label{eq:fn_-+-1}\\
 & +\sum_{i=1}^{n}\left(\frac{\lambda_{2}}{\lambda_{2}+\eta_{1}}\right)^{n-i+1}\eta_{1}\frac{e^{\frac{\left(\sigma\lambda_{2}\right)^{2}}{2}}}{\sqrt{2\pi}}\left(\sigma^{i}\lambda_{2}^{i}\right)I_{n} \left(z,\lambda_{2},\frac{1}{\sigma},-\sigma\lambda_{2}\right)\ ,\quad z\in \mathbb{R},\, \lambda_{2}\neq\eta_{1},\nonumber 
\end{align}
 In addition, these c.d.f.'s admit the following recursive
relations: 
\begin{align}
F_{W+\sum_{i=1}^{n+1}J_{i}^{+}+Z_{e}^{+}} (z)& =\frac{\lambda_{1}}{\left(\lambda_{1}-\eta_{1}\right)}\left(F_{W+\sum_{i=1}^{n}J_{i}^{+}+Z_{e}^{+}}(z)-\eta_{1}\frac{e^{\frac{\left(\sigma\lambda_{1}\right)^{2}}{2}}}{\sqrt{2\pi}}\left(\sigma^{n}\lambda_{1}^{n}\right)I_{n}\left(z;-\lambda_{1},-\frac{1}{\sigma},-\lambda_{1}\sigma\right)\right),\label{eq:fnrec_++-1}
\end{align}
 
\begin{align}
F_{W-\sum_{i=1}^{n+1}J_{i}^{-}+Z_{e}^{+}}(z) & =\frac{\lambda_{2}}{\lambda_{2}+\eta_{1}}\left(F_{W-\sum_{i=1}^{n+1}J_{i}^{-}+Z_{e}^{+}}(z)+\eta_{1}\frac{e^{\frac{\left(\sigma\lambda_{2}\right)^{2}}{2}}}{\sqrt{2\pi}}\left(\sigma^{n}\lambda_{2}^{n}\right)I{}_{n}\left(z,\lambda_{2},\frac{1}{\sigma},-\sigma\lambda_{2}\right)\right).\label{eq:fnrec_-+-1}
\end{align}
\end{prop}
\begin{proof}
The above expressions can all be derived by integrating the corresponding
densities given in (\ref{eq:fn_++}), (\ref{eq:fn_-+}), (\ref{eq:fnrec_++}), (\ref{eq:fnrec_-+}) and taking into
account the definitions of the $Hh$ and $I$ functions. 
\end{proof}

We remark that although we did not provide the tail probability formulas $F_{\sigma W_{T}+\sum_{i=1}^{m}\hat{J}_{i}^{+}-\hat{Z}_{e}^{-}}$  and $F_{\sigma W_{T}-\sum_{i=1}^{m}\hat{J}_{i}^{-}-\hat{Z}_{e}^{-}}$ we point out that they can be   derived by symmetry. For example, we have
\[
F_{\sigma W_{T}+\sum_{i=1}^{m}\hat{J}_{i}^{+}-\hat{Z}_{e}^{-}}(s) = 1  - F_{\sigma W_{T}-\sum_{i=1}^{m}\hat{J}_{i}^{+}+\hat{Z}_{e}^{-}}(-s).
\] In fact, these c.d.f.'s appear in the pricing formula as the functions $T_{i,j}$'s,   defined in Section \ref{sect-app1}.  We now give  the formula for the call option price. Recall  the call price can be written as  
\begin{equation}
C\left(S\right) = S\mathbb{E}\left\{ e^{(-\frac{\sigma^2}{2}-\frac{\alpha}{T}-\kappa\zeta)T+\sigma W_{T}+\sum_{i=1}^{N_T}J_{i}+Z_{e}}\mathbf{1}_{\{S_{T}>K\}}\vert S_0=S\right\} -e^{-rT}K\Q\left\{ S_{T}>K\vert S_0=S\right\} \label{eq:price}.
\end{equation}
Using the results above, we can write 
\begin{eqnarray}
 &  & \Q\left\{S_{T}>K\vert S_0=S\right\} \label{eq:rightpiece}\\
 & = & \sum_{n=1}^{\infty}\Q\left\{ N_{T}=n\right\} \Q\left\{ \sigma W_{T}+\sum_{i}^{n}J_{i}+Z_{e}>\mbox{log}\left(\frac{K}{S}\right)-\left(r-\frac{\sigma^{2}}{2}-\frac{\alpha}{T}-\kappa\zeta\right)T\right\} \nonumber \\
 & = & \sum_{n=1}^{\infty}\frac{\kappa^{n}e^{-\kappa}}{n!}
\bigg[u\sum_{m=1}^{n}\left(P_{n,m}\Q\left\{ \sigma W_{T}+\sum_{i}^{m}J_{i}^{+}+Z_{e}^{+}>k\right\} +Q_{n,m}\Q\left\{ \sigma W_{T}-\sum_{i}^{m}J_{i}^{-}+Z_{e}^{+}>k\right\} \right)+\nonumber \\
 &  & +w\sum_{m=1}^{n}\left(P_{n,m}\Q\left\{ \sigma W_{T}+\sum_{i}^{m}J_{i}^{+}-Z_{e}^{-}>k\right\} +Q_{n,m}\Q\left\{ \sigma W_{T}-\sum_{i}^{m}J_{i}^{-}-Z_{e}^{-}>k\right\} \right)\bigg]\nonumber \\
 & = & \sum_{n=1}^{\infty}\frac{\kappa^{n}e^{-\kappa}}{n!}\left[u\sum_{m=1}^{n}\left(P_{n,m}T_{1,n}\left(k,\Theta\right)+Q_{n,m}T_{2,n}\left(k,\Theta\right)\right)+w\sum_{m=1}^{n}\left(P_{n,m}T_{3,n}\left(k,\Theta\right)+Q_{n,m}T_{4,n}\left(k,\Theta\right)\right)\right] \nonumber,
\end{eqnarray}
where $k\equiv\mbox{log}\left(\frac{K}{S}\right)-\left(r-\frac{\sigma^{2}}{2}-\kappa\zeta\right)T-\alpha$. 

It remains to compute 
\begin{align}
   \mathbb{E}&\left\{ e^{\sigma W_{T}+\sum_{i=1}^{N_T}J_{i}+Z_{e}}\mathbf{1}_{\{S_{T}>K\}}\vert S_0=S\right\} \label{eq:leftpiece} \\
 & =  \sum_{n=1}^{\infty}\frac{\mbox{\ensuremath{\left(\kappa T\right)}}^{n}e^{-\kappa T}}{n!}\bigg[u\sum_{m=1}^{n}P_{n,m}\mathbb{E}\left\{ e^{\sigma W_{T}+\sum_{i=1}^{m}J_{i}^{+}+Z_{e}^{+}}\mathbf{1}_{\{ \sigma W_{T}+\sum_{i=1}^{m}J_{i}^{+}+Z_{e}^{+}>k \}}\right\} + \label{eq:Eexp1}\\
 &   +u\sum_{k=1}^{n}Q_{n,k}\mathbb{E}\left\{ e^{\sigma W_{T}-\sum_{i=1}^{m}J_{i}^{-}+Z_{e}^{+}}\mathbf{1}_{\{\sigma W_{T}-\sum_{i=1}^{m}J_{i}^{-}+Z_{e}^{+}>k\}}\right\} + \label{eq:Eexp2}\\
 &   +w\sum_{k=1}^{n}P_{n,k}\mathbb{E}\left\{ e^{\sigma W_{T}+\sum_{i=1}^{m}J_{i}^{+}-Z_{e}^{-}}\mathbf{1}_{\{\sigma W_{T}+\sum_{i=1}^{m}J_{i}^{+}-Z_{e}^{-}>k\}}\right\} + \label{eq:Eexp3}\\
 &   +w\sum_{k=1}^{n}Q_{n,k}\mathbb{E}\left\{ e^{\sigma W_{T}-\sum_{i=1}^{m}J_{i}^{-}-Z_{e}^{-}}\mathbf{1}_{\{\sigma W_{T}-\sum_{i=1}^{m}J_{i}^{-}-Z_{e}^{-}\}}\right\}\bigg].\label{eq:Eexp4}
\end{align}
In order to  compute the expectation    in (\ref{eq:Eexp1}), we apply the p.d.f. from 
Prop. \ref{prop:densities} to write
\begin{align}
  \mathbb{E}&\left\{ e^{\sigma W_{T}+\sum_{i=1}^{m}J_{i}^{+}+Z_{e}^{+}}\mathbf{1}_{\{\sigma W_{T}+\sum_{i=1}^{m}J_{i}^{+}+Z_{e}^{+}>k\}}\right\} = \nonumber \\
 & =   \int_{k}^{\infty}\int_{\mathbb{R}}\frac{e^{t-x}e^{-\frac{\left(t-x\right){}^{2}}{2\sigma^{2}T}}}{\sqrt{2\pi\sigma^{2}T}}\int_{\mathbb{R}}e^{x-y}f_{\sum_{i=1}^{m}J_{i}^{+}}\left(x-y\right)e^{y}f_{Z_{e}^{+}}\left(y\right)dydxdt\nonumber \\
 & =   \int_{k}^{\infty}\int_{\mathbb{R}}e^{\frac{\sigma^{2}T}{2}}\frac{e^{-\frac{\left(t-x-\sigma^{2}T\right){}^{2}}{2\sigma^{2}T}}}{\sqrt{2\pi\sigma^{2}T}}\int_{\mathbb{R}}\left(\frac{\lambda_{1}}{\lambda_{1}-1}\right)^{m}f_{\sum_{i=1}^{m}\hat{J}_{i}^{+}}\left(x-y\right)\frac{\eta_{1}}{\eta_{1}-1}f_{\hat{Z}_{e}^{+}}\left(y\right)dydxdt\nonumber \\
& = e^{\frac{\sigma^{2}T}{2}}\left(\frac{\lambda_{1}}{\lambda_{1}-1}\right)^m\frac{\eta_{1}}{\eta_{1}-1}\Q\left(\sigma W_{T}+\sum_{i=1}^{m}\hat{J}_{i}^{+}+\hat{Z}_{e}^{+}>k-\sigma^2T\right) \nonumber \\
& = e^{\frac{\sigma^{2}T}{2}}\left(\frac{\lambda_{1}}{\lambda_{1}-1}\right)^m\frac{\eta_{1}}{\eta_{1}-1}F_{\sigma W_{T}+\sum_{i=1}^{m}\hat{J}_{i}^{+}+\hat{Z}_{e}^{+}}\left(k-\sigma^2T\right) \label{eq:Eexp1_s}, 
\end{align}
where $\hat{J}_{i}^{+}\sim Exp\left(\lambda_{1}-1\right)$ and $\hat{Z}_{e}^{+}\sim Exp\left(\eta_{1}-1\right)$.
Expressions analogous to (\ref{eq:Eexp1_s}) can be found for the  terms (\ref{eq:Eexp2}), (\ref{eq:Eexp3}) and (\ref{eq:Eexp4}).  Substituting these expressions and rearranging terms, we get
\begin{align*}
 &\mathbb{E}\left\{ e^{(-\frac{\sigma^2}{2}-\frac{\alpha}{T}-\kappa\zeta)T+\sigma W_{T}+\sum_{i=1}^{N}J_{i}+Z_{e}}\mathbf{1}_{\{S_{T}>K\}}\vert S_0=S\right\} \\
 & =   e^{-\alpha}\sum_{n=1}^{\infty}\frac{\mbox{\ensuremath{\left(\hat{\kappa}T\right)}}^{n}e^{-\hat{\kappa}T}}{n!}\left(\hat{u}\sum_{k=1}^{n}\left(\hat{P}_{n,k}T_{1,n}\left(k,\Theta\right)+\hat{Q}_{n,k}T_{2,n}\left(k,\Theta\right)\right)+\hat{w}\sum_{k=1}^{n}\left(\hat{P}_{n,k}T_{3,n}\left(\hat{k},\Theta\right)+\hat{Q}_{n,k}T_{4,n}\left(\hat{k},\Theta\right)\right)\right),
\end{align*}
 where $\hat{P}_{n,k}$, $\hat{Q}_{n,k}$ are calculated as $P_{n,k}$,
$Q_{n,k}$ but with the parameters $\hat{\eta}_{1,2},\hat{\lambda}_{1,2}$
in place of $\eta_{1,2}$ and $\lambda_{1,2}$. In addition,  the Poisson intensity parameter has also been transformed to $\hat{\kappa}\equiv (m+1)\kappa$, with $m=p\frac{\lambda_1}{\lambda_1-1}+q\frac{\lambda_2}{\lambda_2+1}-1$. Finally, substituting the expressions  for (\ref{eq:rightpiece}) and (\ref{eq:leftpiece}) into  (\ref{eq:price}) concludes the proof.

\subsection{Proof of Propositions \ref{prop:lower_bound} and \ref{prop:upper_bound}}
The first part of Proposition \ref{prop:lower_bound} follows from   Jensen's inequality, namely, 
 
\begin{eqnarray}
C\left(t,S\right) & \geq & \mathbb{E}\left\{ e^{-r\tau}\left(Se^{r\tau+Z_{e}}-K\right)^{+}\right\}  \notag\\
 & = & \int_{\mathbb{R}^{+}}C_{BS}\left(\tau,S;\frac{\hat{\sigma}}{\sqrt{\tau}},K,r\right)\G\text{\ensuremath{\left(d\hat{\sigma}\right)}}
  \geq  C_{BS}\left(\tau,S;\frac{\hat{\sigma}_{min}}{\sqrt{\tau}},K,r\right),\label{eqineq}
\end{eqnarray}
where $\tau\equiv T-t$. In \eqref{eqineq}, the  equality follows from the tower property of the conditional
expectation and the last inequality follows from the monotonicity
of $C_{BS}$ w.r.t. the volatility parameter $\sigma$. The second part of the proposition is proved
in a similar fashion.

To prove Proposition \ref{prop:upper_bound}, we first observe that 

\begin{eqnarray*}
C\left(t,S\right) & = & \int_{\mathbb{R}^{+}\times\mathbb{R}^{+}}C_{BS}\left(\tau,S;\sqrt{\frac{\tilde{\sigma}^{2}+\hat{\sigma}^{2}}{\tau}},K,r\right)\H\left(d\tilde{\sigma}\right)\G\left(d\hat{\sigma}\right).
\end{eqnarray*}
For an ATM-forward option, i.e. $K=e^{r\tau}S$, we notice
that the Black-Scholes price $C_{BS}$ is concave in its volatility parameter $\sigma$. Consequently, by     Jensen's inequality, we obtain
the upper bound
\[
C(t,S)\leq C_{BS}\left(\tau ,S;\sqrt{\int_{\mathbb{R}^{+}}\frac{\tilde{\sigma}^{2}}{\tau}\H\left(d\tilde{\sigma}\right)+\int_{\mathbb{R}^{+}}\frac{\hat{\sigma}^{2}}{\tau}\G\left(d\hat{\sigma}\right)},K,r\right).
\]

\subsection{Proof of Propositions \ref{prop:asy1} and \ref{prop:asy2}}\label{app-strike}

Propositions \ref{prop:asy1} and \ref{prop:asy2} are  applications of the more general
results developed in \citet{Benaim2008}. Denote by  $M\left(\omega\right)\equiv\mathbb{E}\left\{ e^{\omega X}\right\} $
the m.g.f. of the r.v. $X$, and with $F$ its c.d.f. As mentioned in
Section \ref{sec:IV}, if $r^{*}\equiv\mbox{inf}\left\{ \omega\ s.t.\ M\left(\omega\right)<\infty\right\} $
is finite, than  it follows  that $\mbox{lim sup}_{x\rightarrow\infty}\frac{-{\log}\left(1-F\left(x\right)\right)}{x}=r^{*}$.
In turn, if $F$ is well-behaved,
the lim sup can be replaced by a lim and the tail asymptotics $\frac{-{\log}\left(1-F\left(x\right)\right)}{x}\sim r^{*}x$,
along with the condition $r^{*}>1$ are sufficient to provide the asymptotics for the implied volatility
\[
\frac{I^{2}\left(t;K,T\right)\left(T-t\right)}{{\log}\left(\frac{K}{S_{t}}\right)}\sim\mbox{\ensuremath{\xi} {\ensuremath{\left(p^{*}\right)}}}, \quad \text{as}\,\, \ensuremath{K}\rightarrow\infty\,;
\] see Theorems 9 and 10 of \citet{Benaim2008} for this result and the associated technical conditions. A symmetric argument holds for the negative
tails of $F$ and of the implied volatility.
Therefore, given that the log stock price under the extended Kou and Heston models  admits
a m.g.f, it remains to prove  that $\frac{-{\log}\left(1-F\left(x\right)\right)}{x}\sim r^{*}x$.
Theorem 7 and Theorem 8 in \citet{Benaim2008} provide sufficient
conditions for models admitting a m.g.f. to ensure that such condition holds. 

The m.g.f. $M$ of $X\equiv\hbox{log}\left(\frac{S_T}{S_t}\right)$ under the extended Heston model (\ref{eq:heston_dyn}), given $\sigma_t=\sigma$, satisfies\footnote{For the derivation of the Heston m.g.f.,  see, e.g.  \citet{BanoRollin2009}.}
\begin{align*}
{\log}M\left(\omega\right) & =  C+\omega D-\frac{\nu\vartheta}{\zeta^{2}}\left(2{\log}\left(\frac{1-ge^{-d\left(T-t\right)}}{1-g}\right)+d(T-t)\right)+\frac{\nu-\rho\zeta\omega -d}{\zeta^{2}}\frac{1-e^{-d\left(T-t\right)}}{1-ge^{-d\left(T-t\right)}}\sigma^{2}+\psi_e(\omega),\\
\text{ with} ~~g & =  \frac{\nu-\rho\zeta\omega -d}{\nu-\rho\zeta\omega +d}, \quad
d  = \mbox{\ensuremath{\sqrt{\left(\nu-\rho\zeta\omega \right)^{2}+\zeta^{2}\left(\omega -\omega^{2}\right)}}}, \quad
\psi_e(\omega)  =  {\log}\left(u\frac{\eta_{1}}{\eta_{1}-\omega}+w\frac{\eta_{2}}{\eta_{2}+\omega}\right),
\end{align*}
where $C$ and $D$ are constants. Clearly, we have $r^{*}=\min\left\{ p,\eta_{1}\right\}$ , where $p$ is the smallest positive solution of $1-ge^{-d\left(T-t\right)}\vert _{\omega=p}=0.$
In turn, the last equation is equivalent to 
\[
\nu-\rho\zeta p+\mbox{\ensuremath{\sqrt{\left(\nu-\rho\zeta p\right)^{2}+\zeta^{2}\left(p-p^{2}\right)}}}\mbox{coth}\left(\frac{\left(T-t\right)}{2}\mbox{\ensuremath{\sqrt{\left(\nu-\rho\zeta p\right)^{2}+\zeta^{2}\left(p-p^{2}\right)}}}\right)=0.
\]
Notice that, when $r^{*}=p$, the dominating term for $\omega\rightarrow p$
is $\frac{\nu-\rho\zeta\omega i-d}{\zeta^{2}}\frac{1-e^{-d\left(T-t\right)}}{1-ge^{-d\left(T-t\right)}}\sigma^{2}$.
Using l'Hopital's rule, it follows that 
\[
\frac{p-\omega}{1-ge^{-d\left(T-t\right)}}\rightarrow \text{constant, \quad as}~ \omega\rightarrow p^{-}.
\]
This means that $\frac{\nu-\rho\zeta\omega-d}{\zeta^{2}}\frac{1-e^{-d\left(T-t\right)}}{1-ge^{-d\left(T-t\right)}}\sigma_{t}^{2}$
is regularly varying of index 1 as a function of $\frac{1}{p-\omega}$
and Criterion II of  Theorem 8 in  \citet{Benaim2008} is satisfied. When
$r^{*}=\eta_{1}$, we have 
\[
M\left(\eta_{1}-z\right)\sim u\eta_{1}z^{-1}\mbox{ as \ensuremath{z\rightarrow0^{+}}},
\]
and thus Criterion I of   Theorem 7  in  \citet{Benaim2008}  is satisfied.
A similar argument holds for the negative tail.

Now consider the m.g.f. of $X\equiv\hbox{log}\left(\frac{S_T}{S_t}\right)$ under the extended Kou model:
\[
{\log}M\left(\omega\right)=\mu\omega+\frac{\sigma^{2}\omega^{2}}{2}+\kappa\left(p\frac{\lambda_{1}}{\lambda_{1}-\omega}+\left(1-p\right)\frac{\lambda_{2}}{\lambda_{2}+\omega}-1\right)+{\log}\left(u\frac{\eta_{1}}{\eta_{1}-\omega}+\left(1-u\right)\frac{\eta_{2}}{\eta_{2}+\omega}\right),
\]
where $\mu$ is a constant. This m.g.f.   satisfies the tail asymptotics condition with  $r^{*}=\min\left\{ \lambda_{1},\eta_{1}\right\}$, and 
\[
\begin{cases}
{\log}M\left(\lambda_{1}-z\right)\sim\kappa p\lambda_{1}z^{-1}\mbox{ as \ensuremath{z\rightarrow0^{+}}} & \quad if\ \lambda_{1}\leq\eta_{1},\\
M\left(\eta_{1}-z\right)\sim u\eta_{1}z^{-1}\mbox{ as \ensuremath{z\rightarrow0^{+}}} & \quad if\ \lambda_{1}>\eta_{1}. 
\end{cases}
\]
Hence,  $M$ satisfies either Criterion I or II (depending on  whether $\lambda_1>\eta_1$ or not) of    
Theorems 7 and 8 in  \citet{Benaim2008}. A similar argument holds for the negative tail.

\subsection{Proof of Proposition \ref{prop:american_bounds}}

Let $A\left(t,S;u\right)$ be the American put price when the earnings announcement is  scheduled at time $T_{e}=u,\ t<u\leq T$. Our goal is to show that $A\left(t,S;u\right)\geq A\left(t,S;l\right)$, for $t< u\leq l$.
W.l.o.g., let $t=0$ and write $A\left(S;l\right)\equiv A\left(0,S;l\right)$. Let $X_s = \log(S_s/S_0)  - 1_{\{ s\ge T_e\}}Z_e$ be the log stock price excluding the EA jump, and denote by  $(\mathcal{F}^u_s)_{0 \le s \le T}$ (resp.  $(\mathcal{F}^l_s)_{0 \le s \le T}$) the filtration generated by $S$  with  $T_e=u$ (resp. $T_e = l$). Notice that $\mathcal{F}^l_s =  \mathcal{F}^u_s$ for $s<u$ or $s\ge l$, and  $\mathcal{F}^l_s\subset \mathcal{F}^u_s$, for $  u \le s < l$. Therefore,  the sets of   stopping times w.r.t to $\mathcal{F}^u$ and $\mathcal{F}^l$, denoted respectively by  $\mathcal{T}^u$ and $\mathcal{T}^l$, satisfy $\mathcal{T}^l\subset \mathcal{T}^u$. Thus, for any  candidate  stopping time $\tau\in \mathcal{T}^l$, we have 
\begin{align}
&\mathbb{E}\left\{ e^{-r\tau}\left(K-Se^{X_{\tau}+1_{\left\{ \tau\geq u\right\} }Z_{e}}\right)^{+} | \,\F^l_{l-}  \right\}\label{ineq0}\\
&=\mathbb{E}\left\{ 1_{\{ \tau <l \}}e^{-r\tau}\left(K-Se^{X_{\tau}+1_{\left\{ \tau\geq u\right\} }Z_{e}}\right)^{+} | \,\F^l_{l-}  \right\}  + 
\mathbb{E}\left\{ 1_{\{ \tau  \ge l\}}e^{-r\tau}\left(K-Se^{X_{\tau}+1_{\left\{ \tau\geq u\right\} }Z_{e}}\right)^{+} | \,\F^l_{l-}  \right\} \label{ineq1}  \\
&\ge \mathbb{E}\left\{ 1_{\{ \tau <l \}}e^{-r\tau}\left(K-Se^{X_{\tau}}\right)^{+} | \,\F^l_{l-}  \right\}  + 
\mathbb{E}\left\{ 1_{\{ \tau  \ge l\}}e^{-r\tau}\left(K-Se^{X_{\tau}+1_{\left\{ \tau\geq l\right\} }Z_{e}}\right)^{+} | \,\F^l_{l-}  \right\}\label{ineq2}  \\&= \mathbb{E}\left\{e^{-r\tau}\left(K-Se^{X_{\tau}+1_{\left\{ \tau\geq l\right\} }Z_{e}}\right)^{+} | \F^l_{l-} \right\}.
\label{ineq3}
\end{align}
The inequality  holds as follows.  In  the first term of \eqref{ineq1}, given the information in $\F^l_{l-}$ and on $\{\tau < l\}$ the values of $\tau$ and $X_\tau$ are known, we   apply  Jensen's inequality to get the first term in \eqref{ineq2}. In addition, since  $\tau \ge l$  implies that $\tau \ge u$,   the second terms of \eqref{ineq1} and \eqref{ineq2} are equal. The last equality follows from the fact that $1_{\{ \tau \ge l\}}=0$ on $\{ \tau <l\}$. 

In turn, taking expectations in \eqref{ineq0} and \eqref{ineq3} and maximizing  both sides over   $\mathcal{T}^l$,  we get  \[A\left(S;u\right)\geq \sup_{\tau  \in \mathcal{T}^l} \mathbb{E}\left\{ e^{-r\tau}\left(K-Se^{X_{\tau}+1_{\left\{ \tau\geq u\right\} }Z_{e}}\right)^{+}  \right\}  \ge  A\left(S;l\right),\]where the first inequality follows from the inclusion $ \mathcal{T}^l\subset \mathcal{T}^u$.  
%which concludes the proof because
%\begin{eqnarray*}
%A\left(S;l\right) & = & \sup_{0\leq\tau\leq T}\mathbb{E}\left\{ \mathbb{E}\left\{ e^{-r\tau}\left(K-Se^{X_{\tau}+1_{\left\{ \tau\geq l\right\} }Z_{e}}\right)^{+}\vert \mathcal{F}^X_{T}\right\} \right\} .
%\end{eqnarray*} 

\subsection{Barone-Adesi Approximation} Here we  give a sketch of the derivation of the approximation formula \eqref{eq:AM_DE_apx-1}. The     arguments are adapted from  \citet{Barone-Adesi1987} and \citet{Kou2004}. We start by writing
\[
A\left(t,S\right)=P_{E}\left(t,S\right)+\epsilon\left(t,S\right),
\]
where $P_{E}$ is the European put price with an EA jump, and $\epsilon$ is a correction
term. Note that $P_E$ is computed  using   the results in Proposition
\ref{prop:KOU_extension}. In the continuation region, $\epsilon$ must satisfy the same PIDE as that of  $P_{E}$ and $A$, namely
\begin{equation}
r\epsilon\left(t,S\right)-\partial_t\epsilon\left(t,S\right)-\mathcal{L}\epsilon\left(t,S\right)=0,\label{eq:BA_original}
\end{equation}
where the operator $\mathcal{L}$ is defined  in (\ref{eq:Kou_generator}). The idea of the approximation in \citet{Barone-Adesi1987} is to remove the term $\epsilon_{t}$
in \eqref{eq:BA_original}. This involves letting   $\epsilon\left(t,S\right)\equiv g\left(z,S\right)z$,
with  $z\equiv1-e^{-r\left(T-t\right)}$, substituting in the above PIDE, and  ignoring the term $\left(1-z\right)g_{z}$. This results in  the OIDE
\begin{equation}
\frac{r}{z}\epsilon\left(t,S\right)-\mathcal{L}\epsilon\left(t,S\right)=0.\label{eq:BA_OIDE}
\end{equation}
While (\ref{eq:BA_OIDE}) holds in the continuation region, in the exercise reason we have $\epsilon\left(t,S\right)=K-S-P_{E}\left(t,S\right)$.  Following  \citet{Kou2004}, we consider  the ansatz\[
\epsilon\left(t,S\right)=\begin{cases}
\gamma_{1}\left(t\right)S^{-\beta_{1}}+\gamma_{2}\left(t\right)S^{-\beta_{2}} & \quad,\quad S>\alpha\left(t\right),\\
K-S-P_{E}\left(t,S\right) & \quad,\quad S\leq\alpha\left(t\right),
\end{cases}
\]
where $\alpha\left(t\right)$ is the boundary at time $t$. One can directly verify that the ansatz solves 
the OIDE (\ref{eq:BA_OIDE}) if $\beta_{1,2}$ 
are   two positive solutions to  \ref{eq:betas}, and if 
\begin{equation}
\frac{K}{\lambda_{2}}-\frac{\alpha\left(t\right)}{1+\lambda_{2}}-\int_{-\infty}^{0}P_{E}\left(t,\alpha e^{y}\right)e^{\lambda_{2}y}dy=\gamma_{1}\left(t\right)\frac{\alpha^{-\beta_{1}}}{\lambda_{2}-\beta_{1}}+\gamma_{2}\left(t\right)\frac{\alpha^{-\beta_{2}}}{\lambda_{2}-\beta_{2}},\label{eq:condition}
\end{equation}
In turn, we   impose the continuous-fit and smooth-pasting conditions. The first condition yields (\ref{eq:alphat}) while the second  condition, altogether with  
(\ref{eq:condition}), yields the constants $\gamma_1$ and $\gamma_2$ in (\ref{eq:gamma1}) and (\ref{eq:gamma2}).

\bibliographystyle{apa}
\linespread{-1} 
\begin{small}

\bibliography{biblio}

\begin{thebibliography}{}

\bibitem[\protect\astroncite{Andersen and Andreasen}{2000}]{Andersen2000}
Andersen, L. and Andreasen, J. (2000).
\newblock Jump diffusion models: volatility smile fitting and numerical methods
  for pricing.
\newblock {\em Review of Derivatives Research}, 4:231--262.

\bibitem[\protect\astroncite{Barone-Adesi and Whaley}{1987}]{Barone-Adesi1987}
Barone-Adesi, G. and Whaley, R. (1987).
\newblock Efficient analytic approximation of {A}merican option values.
\newblock {\em The Journal of Finance}, 17(2):301--320.

\bibitem[\protect\astroncite{Barth et~al.}{2011}]{Barth2011}
Barth, M., Johnson, T., and So, C. (2011).
\newblock Dynamics of earnings announcement news: Evidence from option prices.
\newblock {\em Working paper}.

\bibitem[\protect\astroncite{Bates}{1996}]{Bates1996}
Bates, D. (1996).
\newblock Jumps and stochastic volatility: {T}he exchange rate processes
  implicit in {D}eutschemark options.
\newblock {\em Review of Financial Studies}, 9(1):69--107.

\bibitem[\protect\astroncite{Benaim and Friz}{2008}]{Benaim2008}
Benaim, S. and Friz, P. (2008).
\newblock Smile asymptotics {II}: models with known moment generating
  functions.
\newblock {\em Journal of Applied Probability}, 45 (1):1--291.

\bibitem[\protect\astroncite{Benaim et~al.}{2012}]{Benaim2008a}
Benaim, S., Friz, P., and Lee, R. (2012).
\newblock On the {B}lack-{S}choles implied volatility at extreme strikes.
\newblock In Cont, R., editor, {\em Frontiers in Quantitative Finance:
  Volatility and Credit Risk Modeling}, pages 19--45. Wiley \& Sons.

\bibitem[\protect\astroncite{Bensoussan and Lions}{1984}]{Bensoussan1984}
Bensoussan, A. and Lions, J. (1984).
\newblock {\em Impulse control and quasivariational inequalities.}
\newblock Gaunthier-Villars.

\bibitem[\protect\astroncite{Billings and Jennings}{2011}]{Billings2010}
Billings, M. and Jennings, R. (2011).
\newblock The option marketÕs anticipation of information content in earnings
  announcements.
\newblock {\em Review of Accounting Studies Conference Version}, 16:587--619.

\bibitem[\protect\astroncite{Black and Scholes}{1973}]{Black1973}
Black, F. and Scholes, M. (1973).
\newblock The pricing of options and corporate liabilities.
\newblock {\em Journal of Political Economy}, 81:637--654.

\bibitem[\protect\astroncite{Broadie et~al.}{2009}]{Broadie2009}
Broadie, M., Chernov, M., and Johannes, M. (2009).
\newblock Understanding index option returns.
\newblock {\em Review of Financial Studies}, 22(11):4493--4529.

\bibitem[\protect\astroncite{Carr et~al.}{2002}]{Carr2002}
Carr, P., Geman, H., Madan, D., and Yor, M. (2002).
\newblock The fine structure of asset returns: An empirical investigation.
\newblock {\em Journal of Business}, 75(2):305--332.

\bibitem[\protect\astroncite{Carr and Madan}{1999}]{Carr1999}
Carr, P. and Madan, D. (1999).
\newblock Option pricing and the fast {F}ourier transform.
\newblock {\em Journal of Computational Finance}, 2:61--73.

\bibitem[\protect\astroncite{Chatrah et~al.}{2009}]{Kiseop}
Chatrah, A., Christie-David, R., and Lee, K. (2009).
\newblock How potent are news reversals?: {E}vidence from the futures markets.
\newblock {\em Journal of Futures Market}, 29:42--73.

\bibitem[\protect\astroncite{Chordia and Shivakumar}{2006}]{Chordia2006}
Chordia, T. and Shivakumar, L. (2006).
\newblock Earnings and price momentum.
\newblock {\em Journal of Financial Economics}, 80(3):627--656.

\bibitem[\protect\astroncite{Coleman and Li}{1994}]{Coleman1994}
Coleman, T. and Li, Y. (1994).
\newblock On the convergence of reflective {N}ewton methods for large-scale
  nonlinear minimization subject to bounds.
\newblock {\em Mathematical Programming}, 67, 2:189--224.

\bibitem[\protect\astroncite{Coleman and Li}{1996}]{Coleman1996}
Coleman, T. and Li, Y. (1996).
\newblock An interior, trust region approach for nonlinear minimization subject
  to bounds.
\newblock {\em SIAM Journal on Optimization}, 6:418--445.

\bibitem[\protect\astroncite{Cont and Tankov}{2002}]{Cont2002}
Cont, R. and Tankov, P. (2002).
\newblock Calibration of jump-diffusion option-pricing models: a robust
  non-parametric approach.
\newblock {\em Working paper}.

\bibitem[\protect\astroncite{Cont and Voltchkova}{2005}]{Cont2005}
Cont, R. and Voltchkova, E. (2005).
\newblock A finite difference scheme for option pricing in jump diffusions and
  exponential {L}\'{e}vy models.
\newblock {\em SIAM Journal on Numerical Analysis}, 43:1596--1626.

\bibitem[\protect\astroncite{del Ba{\~{n}}o~Rollin
  et~al.}{2009}]{BanoRollin2009}
del Ba{\~{n}}o~Rollin, S., Ferreiro-Castilla, A., and Utzet, F. (2009).
\newblock A new look at the {H}eston characteristic function.
\newblock Preprint.

\bibitem[\protect\astroncite{Dennis}{1977}]{Dennis1977}
Dennis, J. (1977).
\newblock Nonlinear least-squares and equations.
\newblock In Jacobs, D., editor, {\em State of the Art in Numerical Analysis},
  pages 269--312. Academic Press.

\bibitem[\protect\astroncite{Donders and Vorst}{1996}]{Donders1996}
Donders, M. and Vorst, T. (1996).
\newblock The impact of firm specific news on {IV}s.
\newblock {\em Journal of Banking and Finance}, 20:1447--1461.

\bibitem[\protect\astroncite{Dubinsky and Johannes}{2006}]{Dubinsky2006}
Dubinsky, A. and Johannes, M. (2006).
\newblock Fundamental uncertainty, earning announcements and equity options.
\newblock Working paper.

\bibitem[\protect\astroncite{Duffie and Singleton}{2000}]{Duffie2000}
Duffie, D., J.~P. and Singleton, K. (2000).
\newblock Transform analysis and option pricing for affine jump-diffusions.
\newblock {\em Econometrica}, 68:1343--1376.

\bibitem[\protect\astroncite{Isakov and Perignon}{2001}]{Isakov2001}
Isakov, D. and Perignon, C. (2001).
\newblock Evolution of market uncertainty around earnings announcements.
\newblock {\em Journal of Banking and Finance}, 25:1769--1788.

\bibitem[\protect\astroncite{Jackson et~al.}{2008}]{Jackson2008}
Jackson, K., Jaimungal, S., and Surkov, V. (2008).
\newblock Fourier space time-stepping for option pricing with {L}\'{e}vy
  models.
\newblock {\em Journal of Computational Finance}, 12(2):1--29.

\bibitem[\protect\astroncite{Kou}{2002}]{Kou2002}
Kou, S. (2002).
\newblock A jump-diffusion model for option pricing.
\newblock {\em Management Science}, 48:1086--1101.

\bibitem[\protect\astroncite{Kou and Wang}{2004}]{Kou2004}
Kou, S.~G. and Wang, H. (2004).
\newblock Option pricing under a double exponential jump diffusion model.
\newblock {\em Management Science}, 50(9):1178--1192.

\bibitem[\protect\astroncite{Lee}{2004}]{Lee2004}
Lee, R. (2004).
\newblock Option pricing by transform methods: extensions, unification, and
  error control.
\newblock {\em Journal of Computational Finance}, 7:51--86.

\bibitem[\protect\astroncite{Lee and Mykland}{2008}]{Lee2008}
Lee, S. and Mykland, P. (2008).
\newblock Jumps in financial markets: A new nonparametric test and jump
  dynamics.
\newblock {\em Review of Financial Studies}, 21(6):2535--2563.

\bibitem[\protect\astroncite{Lord et~al.}{2008}]{Lord08afast}
Lord, R., Fang, F., Bervoets, F., and Oosterlee, C.~W. (2008).
\newblock A fast and accurate {FFT}-based method for pricing early-exercise
  options under {L}\'{e}vy processes.
\newblock {\em SIAM Journal of Scientific Computing}, 30(4):1678--1705.

\bibitem[\protect\astroncite{Madan et~al.}{1998}]{Madan1998}
Madan, D., Carr, P., and Chang, E. (1998).
\newblock The variance gamma process and option pricing.
\newblock {\em European Finance Review}, 2(8):79--105.

\bibitem[\protect\astroncite{Maheu and McCurdy}{2004}]{Maheu2004}
Maheu, J. and McCurdy, T. (2004).
\newblock News arrival, jump dynamics and volatility components for individual
  stock returns.
\newblock {\em Journal of Finance}, 59:755--793.

\bibitem[\protect\astroncite{Mehra et~al.}{2014}]{jpmorgan}
Mehra, A., Kolanovic, M., and Kaplan, B. (2014).
\newblock Earnings and option volatility monitor.
\newblock Technical report, J.P. Morgan.

\bibitem[\protect\astroncite{Merton}{1973}]{Merton1973}
Merton, R. (1973).
\newblock Theory of rational option pricing.
\newblock {\em Bell Journal of Economics and Management Science}, 4:141--183.

\bibitem[\protect\astroncite{Merton}{1976}]{Merton1976}
Merton, R. (1976).
\newblock Option pricing when underlying stock returns are discontinuous.
\newblock {\em Journal of Financial Economics}, 3:125--144.

\bibitem[\protect\astroncite{{\O}ksendal}{2003}]{Oeksendal2003}
{\O}ksendal, B. (2003).
\newblock {\em Stochastic Differential Equations: An Introduction with
  Applications}.
\newblock Springer.

\bibitem[\protect\astroncite{Patell and Wolfson}{1981}]{Patell1981}
Patell, J. and Wolfson, M. (1981).
\newblock The ex ante and ex post price effect of quarterly earnings
  announcements reflected in option and stock prices.
\newblock {\em Journal of Accounting Research}, 19:434--458.

\bibitem[\protect\astroncite{Patell and Wolfson}{1984}]{Patell1984}
Patell, J. and Wolfson, M. (1984).
\newblock The intraday speed of adjustment of stock prices to earnings and
  dividend announcements.
\newblock {\em Journal of Financial Economics}, 13:223--252.

\bibitem[\protect\astroncite{Piazzesi}{2005}]{Piazzesi2005}
Piazzesi, M. (2005).
\newblock Bond yields and the {F}ederal {R}eserve.
\newblock {\em Journal of Political Economy}, 113:311--344.

\bibitem[\protect\astroncite{Raible}{2000}]{Raible2000}
Raible, S. (2000).
\newblock {\em {L}\'{e}vy processes in finance: Theory, numerics, and empirical
  facts}.
\newblock PhD thesis, Univ. Freiburg.

\bibitem[\protect\astroncite{Rogers et~al.}{2009}]{Rogers2009}
Rogers, J., Skinner, D., and Van~Buskirk, A. (2009).
\newblock Earnings guidance and market uncertainty.
\newblock {\em Journal of Accounting and Economics}, 48:90--109.

\end{thebibliography}
\end{small}

\end{document}